\documentclass[11pt,a4paper]{article}
\usepackage[english]{babel}   			   
\usepackage[OT1]{fontenc} 	               
\usepackage[utf8]{inputenc}              
\usepackage[shortcuts]{extdash}
\usepackage[left=3cm, right=3cm, top=3cm, bottom=3cm]{geometry} 
\usepackage{setspace} 
\usepackage{microtype}
\usepackage[medium]{titlesec}

\usepackage{amsmath,amsthm,amssymb}
\usepackage{marvosym,nicefrac}
\usepackage{bm} 
\usepackage{listings} 
\usepackage{stmaryrd}
 \usepackage[longnamesfirst,round]{natbib}
\usepackage[utf8]{inputenc}
\usepackage{mathtools}    

\usepackage{setspace}
\usepackage[hang]{footmisc}
\usepackage{fancyhdr}
\usepackage{adjustbox}
\usepackage{rotating}
\usepackage{mathrsfs}
\usepackage{commath}
\usepackage{etex}
\usepackage{float}
\usepackage{subfig}
\usepackage[colorlinks,citecolor=blue,urlcolor=blue,breaklinks]{hyperref}
\usepackage{breakurl}
\usepackage{color,graphicx,epsfig,epstopdf}
\definecolor{DarkBlue}{rgb}{0.1,0,0.55}
\definecolor{DarkGreen}{RGB}{24,126,35}
\definecolor{lgrey}{RGB}{245,245,245}     
\definecolor{colKeys}{RGB}{0,0,255}       
\definecolor{colIdentifier}{RGB}{0,0,0}	  
\definecolor{colComments}{RGB}{34,139,34} 
\definecolor{colString}{RGB}{160,32,240}  

\usepackage{caption}
\usepackage{booktabs}
\usepackage{bigstrut}
\usepackage{multirow,multicol}
\usepackage{longtable, tabularx, array, pbox, rotating,fullpage}
\usepackage{enumitem} 
\usepackage[flushleft]{threeparttable}
\usepackage{threeparttablex}
\usepackage{longtable}
\usepackage{todonotes} 				   
\RequirePackage[l2tabu, orthodox]{nag} 

\newtheoremstyle{standard}
  {0,5cm}      
  {0,5cm}      
  {\upshape} 
  {}         
  {\bfseries}
  {}        
  {\newline} 
  {}         

\usepackage{colortbl,multirow,hhline}
\usepackage{xcolor}
\usepackage{etoolbox}
\newcounter{saveenumi}

\newcommand{\1}{\mbox{$\mathrm{1\hspace*{-2.5pt}l}$\,}}

\newcommand{\beq}{\begin{equation}}
\newcommand{\eeq}{\end{equation}}
\newcommand{\beqno}{\begin{equation*}}
\newcommand{\eeqno}{\end{equation*}}
\newcommand{\beqn}{\begin{eqnarray}}
\newcommand{\eeqn}{\end{eqnarray}}
\newcommand{\beqnn}{\begin{eqnarray*}}
\newcommand{\eeqnn}{\end{eqnarray*}}
\newcommand{\balgn}{\begin{align}}
\newcommand{\ealgn}{\end{align}}
\newcommand{\balgnn}{\begin{align*}}
\newcommand{\ealgnn}{\end{align*}}

\newcommand{\ben}{\begin{enumerate}}
\newcommand{\een}{\end{enumerate}}
\newcommand{\bit}{\begin{itemize}}
\newcommand{\eit}{\end{itemize}}

\newcommand{\bbm}{\begin{bmatrix}}
\newcommand{\ebm}{\end{bmatrix}}

\newtoggle{withcomments}

\newcommand{\E}[1]{\mathbb{E}\left[#1\right]}

\newcommand{\Prob}[1]{\mathbb{P}\left(#1\right)}
\newcommand{\trace}[1]{\text{tr}\left(#1\right)}
\newcommand{\traces}[1]{\text{tr}\left[#1\right]}
\newcommand{\spec}[1]{\left\lVert #1\right\lVert}
\newcommand{\frob}[1]{\left\lVert #1\right\lVert_F}

\setlist[itemize]{leftmargin=2cm}
\theoremstyle{plain}
\newtheorem{assumption}{Assumption}
\newtheorem{theorem}{Theorem}
\newtheorem{lemma}[theorem]{Lemma}
\usepackage{mathtools}           
\numberwithin{equation}{section}

\usepackage{changepage}
\usepackage[flushleft]{threeparttable}
\usepackage{chngcntr}
\usepackage{apptools}
\counterwithin{assumption}{section}
\counterwithin{theorem}{section}
\renewenvironment{proof}{{\bfseries Proof.}}{\qed}
\pdfoutput=1
\graphicspath{ {figures/p12_r10/} }
\usepackage{amssymb}
\usepackage{tikz}
\numberwithin{equation}{section}
\numberwithin{figure}{section}
\newcommand{\blind}{0}
\begin{document}

\if0\blind
{
    \title{\renewcommand{\thefootnote}{\fnsymbol{footnote}}\vspace{-1.5cm}\textbf{A Regularized Factor-augmented Vector Autoregressive Model}}
    
    \author{
        \Large{Maurizio Daniele}\vspace{.2cm}
        \\ Department of Economics, University of Konstanz\\
        \texttt{\href{mailto:maurizio.daniele@uni-konstanz.de}{maurizio.daniele@uni-konstanz.de}}
        \vspace{.2cm} \\
        and 
        \vspace{.2cm} \\
        \Large{Julie Schnaitmann}\footnote{Corresponding Author. We thank Ralf Br\"uggemann, Mehmet Caner, and the participants at the $($EC$)^2$ 2018 in Rome, the QFFE 2019 in Marseille and the IAAE 2019 in Nicosia for useful comments on earlier versions of the paper. Financial support by the Graduate School of Decision Sciences (GSDS) and the German Science Foundation (DFG) grant number BR 2941/3-1 is gratefully acknowledged.}\vspace{.2cm}
        \\Department of Economics, University of Konstanz\\
        \texttt{\href{mailto:julie.schnaitmann@uni-konstanz.de}{julie.schnaitmann@uni-konstanz.de}}
    }
    \date{\today}
    
    \maketitle
} \fi

\if1\blind
{
    \bigskip
    \bigskip
    \bigskip
    $ $
    \vspace{1cm}
    $ $
    \begin{center}
        {\LARGE\bf A Regularized Factor-augmented Vector Autoregressive Model} \\
        \bigskip
        \today
    \end{center}
    \medskip
} \fi

\bigskip
\begin{abstract}\footnotesize
    We propose a regularized factor-augmented vector autoregressive (FAVAR) model that allows for sparsity in the factor loadings. 
    In this framework, factors may only load on a subset of variables which simplifies the factor identification and their economic interpretation. We identify the factors in a data-driven manner without imposing specific relations between the unobserved factors and the underlying time series. 
    Using our approach, the effects of structural shocks can be investigated on economically meaningful factors and on all observed time series included in the FAVAR model.
    We prove consistency for the estimators of the factor loadings, the covariance matrix of the idiosyncratic component, the factors, as well as the autoregressive parameters in the dynamic model.
    In an empirical application, we investigate the effects of a monetary policy shock on a broad range of economically relevant variables. We identify this shock using a joint identification of the factor model and the structural innovations in the VAR model. 
    We find impulse response functions which are in line with economic rationale, both on the factor aggregates and observed time series level.\\ 

\end{abstract}

\newpage

\section{Introduction} \label{sec:intro}
In this paper, we propose a regularized factor-augmented vector autoregressive (FAVAR) model to investigate the effects of a structural macroeconomic shock on the economy in the presence of many observed time series.
Since the seminal work of \cite{sims1980macroeconomics}, small scale vector autoregressive (VAR) models are conventionally used to analyze the dynamic effects of structural shocks on economic systems.
Structural shocks arise as linear combinations of the reduced form innovations which depend on the variables included. Since VAR models typically incorporate only relatively few variables, the information set spanned by those models is rather limited and important characteristics of the underlying economy may be omitted. This problem is known as ``non-fundamentalness'', which implies that there is no direct mapping between the reduced form and structural innovations and misspecified structural shocks are obtained.\footnote{An overview of the literature is provided in a review paper by \cite{AlessiBarigozziCapasso2011}.} In an empirical application, this informational deficiency may result in misleading impulse response patterns, e.g.\ in the form of price puzzles as described in \cite{Sims1992} and \cite{ramey2016macroeconomic}.
Hence, in order to span the entire space of the structural shocks, it is crucial to incorporate all relevant variables. However, as the number of parameters increases with the square of the number of included variables, the extent to which a VAR model can be enlarged is limited by the number of observations.


To circumvent these drawbacks, dimension reduction techniques that enable the use of the informational content of many time series for structural analysis have obtained increasing attention. A frequently used approach is the FAVAR model introduced by \cite{bernanke2005measuring}. 
FAVAR models have been used to investigate structural monetary policy, fiscal or oil price shocks. An incomplete list comprises the studies by \cite{del200799}, \cite{boivin2007global}, \cite{boivin2009sticky}, \cite{mumtaz2009transmission}, \cite{kilian2011does} and \cite{stock2016factor}.

A FAVAR model decomposes the co-movement of the observed time series into a common and idiosyncratic component. The common component is composed of latent and observed factors that affect the underlying time series according to the corresponding weights represented by the factor loadings matrix.
The idiosyncratic innovations allow for weak cross-sectional and serial correlations in the spirit of \cite{chamberlain1983arbitrage}.
As the number of latent and observed factors is much smaller than the number of time series included, the information contained in a large panel of variables is condensed into a small number of factors.
To allow for dynamics, the FAVAR model incorporates autoregressive structures for both types of factors. Due to the large number of time series incorporated in the FAVAR model, the model accommodates enough variables to span macroeconomic shocks without an omitted variable bias.

There is an identification problem associated with the FAVAR approach. In contrast to the observed factors that can be interpreted, only the common component corresponding to the latent factors can be estimated consistently. To identify the latent factor and factor loadings estimates, restrictions have to be imposed on the FAVAR model. 
In the context of structural analysis, it is crucial to introduce a scheme that allows an economic interpretation of the factors. The named factor identification scheme is used in the factor literature which deals with the analysis of structural shocks. Examples of named factor identification schemes can be found in \cite{bernanke2005measuring}, \cite{stock2016factor} or \cite{bai2016estimation}. This scheme associates each latent factor with a unique observed time series, by which the factor is defined. However, these time series may not represent an entire economic sector appropriately. Furthermore, the naming time series potentially impose implausible relations between the factors and the remaining observed time series that affect the structure of the factor model and its dynamics.

In this paper, we propose a regularized FAVAR model that allows for sparsity in the factor loadings matrices of the latent and observed factors. 
In this framework, factors may only load on a subset of variables which represent different economic sectors. Hence, our framework also incorporates weak factors. Therefore, this allows to identify and interpret the factors economically without imposing restrictive identification restrictions through a named factor scheme. 
Our estimation procedure relies on a penalized quasi-maximum likelihood approach and is based on $L_1$-norm regularization of the factor loadings. The estimation of the factor loadings for the latent and observed factors is conducted in a single step. This allows for a similar degree of shrinkage for both type of factors, which is necessary to retain their structural interpretation.

In the theoretical part of the paper, we prove consistency for the estimators of the factor loadings, the factors and the covariance matrix of the idiosyncratic component under the average Frobenius norm. Moreover, the autoregressive parameters in the dynamic model are consistently estimated as well.

Using our approach, the impact of structural shocks can be investigated on economically meaningful factors and observed time series using impulse response functions. We propose a joint identification scheme of the factor model and the structural innovations in the VAR model. More specifically, the scheme focuses on the identification of shocks to the observed factors. In a monetary policy application, the key interest lies in the monetary policy shock which amounts to a shock in the structural innovation of the Federal Funds rate (FFR). As many time series react contemporaneously to changes in the FFR, it is commonly treated as an observed factor, see e.g.\ \cite{bernanke2005measuring} or \cite{boivin2009sticky}.

Using monthly US macroeconomic data, we apply our framework to investigate the effects of a monetary policy shock on the factors and the underlying time series. We are able to identify the factors in a data-driven manner and 
extract five latent factors that relate to the labor market, prices, industrial production, the stock market and credit spreads. 
The impulse response functions are in line with economic rationale both on the factor aggregates and the observed time series level. In particular, following a tight monetary policy, prices and industrial production fall, credit conditions deteriorate and the level of employment decreases. These results are similar to those obtained by \cite{ForniGambetti2010} in a structural dynamic factor model framework. However, the latter is very sensitive to changes in the model specification in comparison to our regularized FAVAR model.  

The remainder of this paper is organized as follows. 
In Section \ref{sec:model} we describe the regularized FAVAR approach and propose an estimation procedure, where we impose shrinkage on the factor loadings of both the observed and latent factors. We further provide an identification scheme that jointly identifies the factor model and the structural innovations in the dynamic equation. 
In Section \ref{sec:theo} the large sample properties of the sparse factor loadings, the latent factor estimator, the covariance matrix of the idiosyncratic innovations, as well as the coefficients of the dynamic equation are provided. 
In Section \ref{sec:appl} our empirical application investigates the effects of a monetary policy shock on the US economy. Finally, Section \ref{sec:con} concludes.

The following notation is used throughout the paper: $\pi_{\max}(A)$ and $\pi_{\min}(A)$ are the maximum and minimum eigenvalue of a matrix $A$. Further, $\spec{A}$ and $\frob{A}$ denote the spectral and Frobenius norm, respectively. They are defined as $\spec{A} = \sqrt{\pi_{\max}(A'A)}$ and  $\frob{A} = \sqrt{\trace{A'A}}$. For some constant $c > 0$ and a non-random sequence $b_N$, we use the notation $b_N = \mathcal{O}(N)$, if $N^{-1} b_N \to c$, for $N \to \infty$. Moreover, $b_N = o(N)$, if $N^{-1} b_N \to 0$, for $N \to \infty$. Similarly, for a random sequence $d_N$, we say $d_N = \mathcal{O}_p(N)$, if $N^{-1} d_N \overset{p}{\to} c$, for $N \to \infty$ and $d_N = o_p(N)$, if $N^{-1} d_N \overset{p}{\to} 0$, for $N \to \infty$, where $\overset{p}{\to}$ denotes convergence in probability.


\section{Econometric modeling framework}\label{sec:model}

\subsection{Regularized factor-augmented vector autoregressive model}

\noindent We define the factor-augmented vector autoregressive model as 
\begin{equation}\label{favar1}
    x_t = \Lambda^f f_t + \Lambda^g g_t + e_t,
\end{equation}
\begin{equation}\label{favar2}
    \left[\begin{matrix}
        f_t \\ g_t \\ \end{matrix} \right] = \left[ \begin{matrix}
        \Phi^{f f}(L) & \Phi^{f g}(L) \\ \Phi^{g f}(L) & \Phi^{gg}(L) \\ \end{matrix} \right] \left[\begin{matrix}
        f_{t-1} \\ g_{t-1} \\ \end{matrix} \right] + \left[ \begin{matrix}
        \eta^f_t \\ \eta^g_t \\
    \end{matrix} \right], \quad t = 1, \hdots, T,
\end{equation}

\noindent where $x_t$ is a $(N \times 1)$ vector of the observable time series, $f_t$ is a $(r_1 \times 1)$ vector of latent factors, $g_t$ is a $(r_2 \times 1)$ vector of observed factors, $\Lambda^f$ is a $(N \times r_1)$ matrix of factor loadings of the latent factors, and $\Lambda^g$ is a $(N \times r_2)$ matrix of factor loadings of the observed factors. $e_t$ is a $(N \times 1)$ vector of idiosyncratic innovations which may be cross-sectionally and serially correlated, as in the approximate factor model framework. Their covariance matrix is given by $\Sigma_e = \E{e_t e_t'}$. Moreover, $\eta^f_t$ denotes a $(r_1 \times 1)$ and $\eta^g_t$ a $(r_2 \times 1)$ vector of factor innovations associated with the latent and observed factors, respectively, and $\Phi^{ff}(L)$, $\Phi^{fg}(L)$, $\Phi^{gf}(L)$ and $\Phi^{gg}(L)$ are lag polynomials of order $p$. We define the vector of latent and observed factors as $h_t = \left[f_t', g_t'\right]'$.


In the FAVAR specification the latent factors are not identified without further restrictions. 
Their interpretability hinges on the identification restrictions imposed on the factor model. For an exact identification of the model, $r_1^2 + r_1 r_2$ restrictions have to be imposed, see e.g.\ \cite{bernanke2005measuring} or \cite{bai2016estimation}. Among alternative identification schemes, \cite{bai2016estimation} analyze the named factor scheme which yields estimated factors that can be interpreted economically. Intuitively, the idea is to re-order the underlying observed time series such that the first $r_1$ time series define the latent factors. It is assumed that these time series only load onto one factor.\footnote{The resulting factor loadings matrix is given by $\Lambda^f_{NF} = \left[ 
    I_{r_1}, \Lambda^{f\cdot'}\right]'$, where $I_{r_1}$ is a $(r_1 \times r_1)$ identity matrix and $\Lambda^{f\cdot} = \Lambda^{f}_2 \left(\Lambda^{f}_1\right)^{-1}$, with $\Lambda^f = \left[ 
    \Lambda^{f'}_1, \Lambda^{f'}_2 \right]'$ and $\Lambda^{f'}_1$ denotes the upper $(r_1 \times r_1)$ block.}
Hence, the ordering of the time series is crucial for the structure of the factor model.
To circumvent this a priori selection of the factors, we introduce a regularized FAVAR model where we impose sparsity on the factor loadings matrices. This leads to a factor loadings matrix which is economically interpretable and enables us to attribute economic meaning to the estimated factors. Moreover, by introducing sparsity, we can incorporate weak factors which only load on a subset of the time series. As shown by \cite{Onatski2012}, it is crucial to account for weakly influential factors in the estimation procedure. Otherwise, conventional methods (e.g. principal component analysis (PCA)) estimate those factors with large errors. 



A second issue that we target is the estimation of the regularized FAVAR model. The complication in this model arises because it consists of latent and observable factors. This implies that the conventionally used methodologies to extract the latent factors do not work in presence of the observed factors. The effect of the observed factors in explaining the covariance structure of the data has to be accounted for. The approaches introduced in the literature have tackled this problem in various ways. 
\cite{bernanke2005measuring} propose a two-step principal components approach to estimate the FAVAR model which does not take into account that the observed factors also contribute to the factor space used to estimate the latent factors. 
\cite{bai2016estimation} address this issue by introducing a two-step quasi maximum likelihood (QML) approach to estimate the FAVAR model. In a first step the impact of the observed factors is linearly projected out and they extract the latent factors and their factor loadings using QML. In a second step, they estimate the factor loadings of the observed factors by linear projection.

We propose a procedure for the joint estimation of the factors loadings of both types of factors. 
Hence, the informational content explained by the observed factors is directly accounted when estimating the latent factors. 
This allows us to jointly impose shrinkage on the factor loadings estimates which leads to a similar degree of shrinkage for both types of factors.

\subsection{Estimation of the regularized FAVAR model}

The factors and factor loadings in equation \eqref{favar1} can be estimated by either PCA\footnote{See e.g.\ \cite{Bai2002} or \cite{Stock2002} for a detailed treatment of the PCA, in approximate factor models.} or quasi maximum likelihood (QML) estimation under normality.\footnote{\cite{BaiLi2012} deal with the consistent estimation of the strict factor model, whereas \cite{Bai2016} analyze the approximate factor model estimation by QML.} In the following we pursue estimating the factor model by QML, which allows us to introduce sparsity 
in the factor loadings by penalizing the likelihood function.

Intuitively, we can estimate the factor loadings of the observed and latent factors, $\Lambda^f$ and $\Lambda^g$, jointly using the covariance matrix of the observed time series, $x_t$.
In matrix notation, the observation equation of the model is defined as
\begin{align}
    X = \Lambda^f F' + \Lambda^g G' + e = \left[ \begin{matrix} \Lambda^f & \Lambda^g \end{matrix} \right] \left[ \begin{matrix} F' \\ G' \\ \end{matrix} \right] + e = \Lambda H + e, \label{favar_m}
\end{align}
where $X = \left[ x_1, \hdots, x_T \right]$ and $e = \left[ e_1, \hdots, e_T \right]$ are $(N \times T)$ matrices, $F = \left[ f_{1t}, \hdots, f_{r_1t} \right]$ and $G = \left[ g_{1t}, \hdots, g_{r_2t} \right]$ are $(T \times r_1)$ and $(T \times r_2)$ matrices, respectively. 
The covariance matrix of the observed time series $X$ based on the factor model in \eqref{favar_m} can be written as $ \Sigma = \Lambda \Sigma_H \Lambda' + \Sigma_e$,  where $\Sigma_H$ is the composite covariance matrix of the observed and latent factors and $\Sigma_e$ denotes the covariance matrix of the idiosyncratic component. 


Using the previous result, we obtain the following expression for the negative quasi log-likelihood function for the covariance matrix of the data in the FAVAR model
\begin{align}
    \mathcal{L}(\Lambda, \Sigma_{H}, \Sigma_{e}) = \log \left| \Lambda \Sigma_{H} \Lambda' +  \Sigma_e \right| +  \traces{S_x \left(\Lambda \Sigma_{H} \Lambda' + \Sigma_e\right)^{-1}}, \label{neg_log_lik}
\end{align}
where $S_{x} = \frac{1}{T} \sum_{t = 1}^{T} x_tx_t'$ is the sample covariance matrix of the observed data.
In the first step of our model estimation, we treat $\Sigma_{e}$ as a diagonal matrix and define $\Phi_e = \text{diag}\left(\Sigma_{e}\right)$ denoting a diagonal matrix that contains only the main diagonal elements of $\Sigma_{e}$ to reduce the number of parameters in the estimation. 
Thus, our unpenalized objective function reduces to
\begin{align}\label{quasi_lik}
    \mathcal{L}(\Lambda, \Sigma_{H}, \Phi_{e}) = \log \left| \Lambda \Sigma_{H} \Lambda' + \Phi_e \right| +  \traces{S_x \left(\Lambda \Sigma_{H} \Lambda' +  \Phi_e\right)^{-1}}.  
\end{align}
As $\Sigma_{e}$ incorporates correlations of general form, equation \eqref{quasi_lik} may be seen as a quasi log-likelihood function, since it imposes the innovation term structure of a strict factor model. However, \cite{Bai2016} show in the approximate factor model framework that imposing this restrictions on $\Sigma_{e}$ does not affect the consistency of the QML estimator. The diagonality assumption on $\Sigma_{e}$ is relaxed in a second step based on the soft-thresholding estimator introduced in Section \ref{sec:poet}.

In order to introduce sparsity in the factor loadings matrix $\Lambda$, we shrink each element of $\Lambda$ towards zero. This is incorporated based on a penalized maximum likelihood estimation of the objective function in \eqref{quasi_lik} by separate $L_1$-norm penalties on the factor loadings associated with the observed and latent factors, respectively. More specifically, we focus on the following penalized optimization problem
\begin{align}
    \begin{split}\label{pen_favar}
        \underset{\{\Lambda,\Phi_{e} \} }{\text{min}} &\log \left| \Lambda \Sigma_{H} \Lambda' + \Phi_e \right| +  \traces{S_x \left(\Lambda \Sigma_{H} \Lambda' + \Phi_e\right)^{-1}}  \\
        &\quad + \mu_1 \sum_{i=1}^N \sum_{k=1}^{r_1} \left| \lambda^f_{ik} \right| + \mu_2 \sum_{i=1}^N \sum_{l=r_1+1}^{r} \left| \lambda^g_{il} \right|,
    \end{split}
\end{align}
where $\mu_1$ and $\mu_2$ determine the degree of penalization of the factor loadings corresponding to the latent and observed factors, respectively. A clear separation of both sets of factor loadings has the advantage of offering a more flexible treatment of both components.

The latent factors $f_t$ can be estimated by generalized least squares (GLS)
\begin{align}\label{gls_factors}
    \tilde{f}_t = \left(\tilde{\Lambda}^{f'} \tilde{\Phi}_{e}^{-1} \tilde{\Lambda}^f \right)^{-1}\tilde{\Lambda}^{f'} \tilde{\Phi}_{e}^{-1} x_t \, ,
\end{align}
where the estimates $\tilde{\Lambda}^f$ and $\tilde{\Phi}_{e}$ are the ones obtained from the optimization of the objective function in \eqref{pen_favar}.


%
%
%
%
%
%
%



\subsection{Estimation of the idiosyncratic component covariance matrix} \label{sec:poet}
As the diagonality assumption on $\Sigma_{e}$ that we impose in the first step of our estimation is restrictive, we introduce a procedure that relaxes the assumption and allows for the estimation of a possibly sparse idiosyncratic innovation covariance matrix. 
More specifically, we re-estimate $\Sigma_{e}$ by means of the principal orthogonal complement thresholding (POET) estimator introduced by \cite{FanLiaoMincheva2013}. 
The POET estimator allows for sparsity in the idiosyncratic error covariance matrix by shrinking the off-diagonal elements of the sample covariance matrix of the residuals obtained from the estimation of our regularized FAVAR model towards zero using the soft-thresholding method. The estimated idiosyncratic error covariance matrix $\tilde{\Sigma}_{e}^{\tau}$ based on the POET method is defined as
\begin{align*}
    \tilde{\Sigma}_{e}^{\tau} = s_{ij}^{\tau}, \quad s_{ij}^{\tau} = \left\{\begin{array}{ll}
        s_{e,ii}, & i = j\\
        \mathcal{S}(s_{e,ij}, \tau), & i \neq j 
    \end{array}\right. ,
\end{align*}
where $s_{e,ij}$ is the $ij$-th element of the sample covariance matrix $S_e = \frac{1}{T} \sum_{t = 1}^{T} (x_{t} - \tilde{\Lambda}\tilde{h}_t)(x_{t} - \tilde{\Lambda}\tilde{h}_t)'$ of the estimated factor model residuals, $\tilde{h}_t = \left( \tilde{f}_t', g_t' \right)'$ are the estimated factors, $\tau = \frac{1}{\sqrt{N}}+\sqrt{\frac{\log N}{T}}$ is a threshold and $\mathcal{S}(\cdot)$ denotes the soft-thresholding operator defined as
\begin{align}\label{soft_t}
    \mathcal{S}(\sigma_{e,ij}, \tau) = \text{sign}(\sigma_{e,ij})(|\sigma_{e,ij}| - \tau)_+  \, .
\end{align}


\subsection{Identification restrictions}
\label{sec:ident}

The FAVAR model in \eqref{favar1} and \eqref{favar2} is not identified, as it can be expressed as
\begin{align*}
    x_t &= \Lambda A^{-1} \cdot A \, h_t + e_t \\ 
    A h_t &= A\Phi_1 A^{-1} \cdot A h_{t-1} + \cdots + A\Phi_p A^{-1} \cdot A h_{t-p} + A \eta_t, 
\end{align*}
where $h_t = \left(f_t', g_t'\right)'$ and $A$ denotes an invertible $\left[(r_1+ r_2) \times (r_1+ r_2)\right]$ matrix. Moreover, $\Lambda^* = \Lambda A^{-1}$, $h_{t-j}^* = A h_{t-j}$ for all $j=0,\hdots,p$, $\Phi^*_i = A \Phi_i A^{-1}$ for all $i =1, \hdots,p$ and $\eta^*_t = A \eta_t$ are the uniquely identified quantities given a specific choice of $A$.
\cite{bai2016estimation} show that $r_1^2 + r_1 r_2$ restrictions are necessary to uniquely identify the FAVAR model. Hence, we impose restrictions on the factor model and the dynamics of the factors.

The covariance matrix of the VAR innovations $\eta_t^* = \left({\eta_t^{*f}}', {\eta_t^{*g}}'\right)'$ is given by
\begin{equation}
    \Omega^* = \left[\begin{array}{cc}
        \E{\eta^{*f}_t {\eta_t^{*f}}'} & \E{\eta^{*f}_t {\eta_t^{*g}}'} \\
        \E{\eta^{*g}_t {\eta_t^{*f}}'} & \E{\eta^{*g}_t {\eta_t^{*g}}'}
    \end{array}\right] = \left[\begin{array}{cc}
        \Omega^*_{ff} & \Omega^*_{fg} \\
        \Omega^*_{gf} & \Omega^*_{gg} \end{array}\right].\label{error_cov}
\end{equation} 
We consider the following sets of identification restrictions on the covariance matrix of the factor innovations and the factor model.\footnote{Similar identification restrictions are common in the factor model literature, see e.g. \cite{bai2013principal} or \cite{bai2014identification}.}
\begin{itemize}
    \item[\textit{IRa}:] $\Omega^*_{fg} = 0$ and $\Sigma_F = I_{r_1}$.
    \item[\textit{IRb}:] $\Omega^*_{fg} = 0$ and $\Lambda_1 = I_{r_1}$, where $\Lambda_1$ is the upper $r_1 \times r_1$ submatrix of $\Lambda$.
\end{itemize}

Both identification schemes share the restriction that $\Omega^*_{fg} = 0$. This restriction assures that the observed factors $g_t$ are not rotated. Moreover, as this assumption is imposed on the covariance matrix of the structural factor innovations this implies that the structural shocks associated with the observed factors are contemporaneously uncorrelated with those of the latent ones. 


The first set of identifying restrictions \textit{IRa} is used for our regularized FAVAR model. As long as a sufficient amount of sparsity is imposed on the factor loadings, the regularized FAVAR model is identified up to a unitary generalized permutation matrix, by the combination of \textit{IRa} and the $L_1$-norm penalties on $\Lambda$. Hence, the regularized FAVAR model is fully identified, by ordering the columns of the estimated factor loadings matrix $\tilde{\Lambda}$ according to their sparsity and by fixing the sign of each column. We choose the sign of the estimated factors such that they align to the corresponding observable time series. Numerically we can verify, whether enough sparsity is imposed on the factor loadings matrix to achieve identification, by evaluating the rank condition for local identification proposed by \cite{Bekker1986}.

However, if there is no sparsity imposed on the factor loadings matrix (i.e.\ $\mu_1 = \mu_2 = 0$) or the local identification condition is not satisfied, we need restrictions in addition to \textit{IRa} to identify the FAVAR model. For this case, we further use the normalization $\frac{1}{N}\Lambda'\Sigma^{-1}_{e}\Lambda = Q$, where the diagonal entries of $Q$ are assumed to be distinct and arranged in a decreasing order. These restrictions are usually imposed in the maximum likelihood framework for the approximate factor model (see e.g.\ \cite{Lawley1971}).\footnote{Alternatively, we could also impose the restriction $\Omega^*_{ff} = I_{r_1}, \Omega^*_{fg} = 0$ and $\frac{1}{N}\Lambda'\Sigma^{-1}_{e}\Lambda = Q$ 
    (see e.g. identification restriction \textit{IRa} in \cite{bai2016estimation}). This imposes an orthogonality restriction on the covariance matrix of the dynamic factor innovations. However, this version implies that the factor loadings of the latent factors are rotated.} The second restriction \textit{IRb} is conventionally\footnote{Restriction \textit{IRb} is also employed by \cite{bai2016estimation}.} referred to as named factor identification. The first $r_1$ time series offer a direct identification for the latent factors and are assumed to only load on the specific factors, respectively. However, as a consequence the ordering of the observed time series matters crucially for this set of restrictions. 

Based on the different sets of identification schemes the resulting estimates change.
In scheme \textit{IRa}, the identification restrictions are chosen to have no impact on the rotation of the factor loadings corresponding to the latent factors $\tilde{\Lambda}^f$. Hence, the interpretation of $\hat{\Lambda}^{f}$ is not distorted by any rotation and the factor loadings estimates are unchanged $\hat{\Lambda}^{f} = \tilde{\Lambda}^f$. The factor loadings associated with the observed factors are transformed according to $\hat{\Lambda}^{g} = \tilde{\Lambda}^{g} + \tilde{\Lambda}^{f} \tilde{\Omega}_{fg}\tilde{\Omega}^{-1}_{gg}$. Moreover, we rotate the estimated factors by $\hat{F} = \tilde{F} - \tilde{\Omega}_{fg}\tilde{\Omega}_{gg}^{-1} G$ and the autoregressive parameters by $\hat{\Phi}_i = \tilde{A} \tilde{\Phi}_i \tilde{A}^{-1}$, for $i = 1,\dots,p$, with the rotation matrix
$\tilde{A} = \left[\begin{array}{cc}
I_{r_1} &  -\tilde{\Omega}_{fg}\tilde{\Omega}_{gg}^{-1}	\\
0 & I_{r_2} 
\end{array} \right]$. 
In case of \textit{IRb}, let $\tilde{\Lambda}_1$ be the first $r_1 \times r_1$ block of $\tilde{\Lambda}$. The resulting estimated factor loadings are $\hat{\Lambda}^{f} = \tilde{\Lambda}\tilde{\Lambda}_1^{-1}$ and $\hat{\Lambda}^{g} = \tilde{\Lambda}^{g} + \tilde{\Lambda}^{f} \tilde{\Omega}_{fg}\tilde{\Omega}^{-1}_{gg}$. The estimated factors are given by $\hat{F} = \left(\tilde{F} - \tilde{\Omega}_{fg}\tilde{\Omega}^{-1}_{gg}G\right)\tilde{\Lambda}_1'$ and the autoregressive parameters are obtained as $\hat{\Phi}_i = \tilde{A} \tilde{\Phi}_i \tilde{A}^{-1}$, for $i = 1,\dots,p$, where we use the rotation matrix
$\tilde A = \left[\begin{array}{cc}
\tilde{\Lambda}_1 &  -\tilde{\Lambda}_1\tilde{\Omega}_{fg}\tilde{\Omega}_{gg}^{-1}	\\
0 & I_{r_2} 
\end{array} \right]$.
Note that quantities with a ' $\hat{ }$ ' are estimates of the identified quantities denoted by a '$ ^\ast$'.


\subsection{Impulse responses}\label{sec:impl}

Starting from the VAR representation of the factors in equation \eqref{favar2}, we obtain the vector moving average representation of the model and the impulse response functions. By rewriting the dynamic equation of the factors as $\left( I_r - \Phi_1 L - \cdots - \Phi_p L^p \right) h_t = \eta_t$ and using the inverted lag polynomial to write the factors $h_t$ as a function of their innovations $\eta_t$, we obtain
\begin{equation*}
    \label{vma}
    h_t = \left( I_r - \Phi_1 L - \cdots - \Phi_p L^p \right)^{-1} \eta_t =\Psi(L) \eta_t,
\end{equation*}
where $\Psi(L) = I_r + \Psi_1 L + \Psi_2 L^2 + \cdots$ and the $\Psi_i$ are moving average coefficients. Hence, $\Psi(L)$ is the conventional inverted lag polynomial and the $\Psi_i$ can be interpreted as the matrices of responses to the innovations $\eta_t$. 
Furthermore, we accumulate the moving average parameters to analyze the effect of a shock in the factor innovations on the level of the factors, when the underlying observed time series $x_t$ are given in first-differences.

In the factor-augmented VAR model, we can also calculate the responses of the observed time series by employing
\begin{equation*}
    \label{vma_x}
    x_t = \Lambda \Psi(L) \eta_t.
\end{equation*}
In addition, we can investigate the impact of a shock in the factor innovations on all time series included in the factor model. Impulse responses on the individual time series in levels are obtained by pre-multiplying the moving average coefficients $\Psi_i$ by the factor loadings matrix $\Lambda$ and accumulating them, when necessary.

Shocks to the factor innovations $\eta_t$ cannot be interpreted as structural shocks as they are contemporaneously correlated. However, as we identify the factor model and the structural innovations jointly by the schemes described in Section \ref{sec:ident}, we obtain innovations with a block diagonal structure. The structural innovations are given by pre-multiplying the factors and factor innovations by the rotation matrix. The contemporaneous impact matrix is obtained by inverting the rotation matrix and is given by 
\begin{equation}
    \label{A_mat}
    A^{-1} = \left[\begin{array}{cc}
        I_{r_1} &  \Omega_{fg}\Omega_{gg}^{-1}	\\
        0 & I_{r_2} 
    \end{array} \right],
\end{equation}
where $\Omega_{fg}$ and $\Omega_{gg}$ are the corresponding elements of the covariance matrix of the factor innovations $\eta_t$, which is denoted as
$\Omega  
= \left[\begin{array}{cc}
\Omega_{ff} & \Omega_{fg}	\\
\Omega_{gf} & \Omega_{gg} 
\end{array} \right]$. The covariance matrix of the structural innovations $\eta^*_t$ is given by the block diagonal matrix
$
\Omega^*  
= \left[\begin{array}{cc}
\Omega_{f \cdot g} & 0	\\
0 & \Omega_{gg} 
\end{array} \right]
$, where $\Omega_{f \cdot g} = \Omega_{ff} -\Omega_{fg}\Omega_{gg}^{-1}\Omega_{gf}$.
This implies that a structural shock in the innovation of the observed factors are contemporaneously uncorrelated with those of the latent factors. Shocks to the structural innovation of $g_t$ can impact all latent factors contemporaneously. We do not impose zero restrictions on the last $r_2$ columns of the impact matrix $A^{-1}$ in \eqref{A_mat}.
For the named factor identification scheme (\textit{IRb}), we obtain a block diagonal covariance matrix of the following form $\Omega^* = \left[\begin{array}{cc}
\Lambda_1 \Omega_{f \cdot g} \Lambda_1' & 0	\\
0 & \Omega_{gg} 
\end{array} \right]$. 

In case of only one observable factor, i.e.\ $r_2 = 1$, the covariance of the structural innovations $\Omega_{gg}$ is a scalar. Hence, in the FAVAR context the structural innovations of the observed factors are identified by the factor model identification irrespective of the identification scheme used. 

For more than one observable factor, i.e.\ $r_2 > 1$, we can generalize our model. The inverse of the contemporaneous impact matrix can be written as $A = \left[\begin{array}{cc}
A_{11} &  A_{12}	\\
0 & A_{22}
\end{array} \right]$ and the structural innovations of the observed factors are given by $\eta^*_{g,t} = A_{22} \eta_{g,t}$. Hence, their covariance matrix is expressed as
\begin{equation}
    \Omega^*_{gg} = \mathbb{E} \left[ \eta^*_{g,t} {\eta^*_{g,t}}' \right] = A_{22} \Omega_{gg} A_{22}'.
\end{equation}
In this setting, we can impose restrictions on $A_{22}$ to achieve identification of the structural innovations with respect to the observable factors $g_t$. For example, a recursive structure within the observable factors could be used.

Moreover, we derive the impact matrix for the observed data $x_t$. It describes the contemporaneous impacts of an exogenous structural shock on the observed time series. Starting from the implied structural moving average representation for the observed data in lag operator notation, $x_t = \Lambda \Psi(L) A^{-1}  \eta_t^*$, the model in period $t$ reads
$x_t = \Lambda A^{-1} \eta_t^* + \Lambda \Psi_1 A^{-1} \eta^*_{t-1} + \Lambda \Psi_2 A^{-1} \eta^*_{t-2} + \cdots$.
The impact matrix is denoted as $B_0 = \Lambda^* = \Lambda A^{-1}$. Using the estimated quantities, the impact matrix is given by
\begin{equation}
    \hat{B}_0 = \left[ \tilde{\Lambda}^f   \quad \tilde{\Lambda}^f \tilde{\Omega}_{fg} \tilde{\Omega}_{gg}^{-1} + \tilde{\Lambda}^g \right] = \hat{\Lambda}.
\end{equation}
Hence, the sparsity in $\hat{\Lambda}$ provides structure to the contemporaneous impact matrix. The zeros in the factor loadings yield timing restrictions which are data-driven but can be explained economically. Moreover, the strength and direction of the effect is guided by the regularized factor loadings.

Inference on the impulse response functions of the regularized FAVAR model is conducted using a residual-bootstrap based on the dynamic factor equation \eqref{favar2}. As shown in the following section, all estimators based on the regularized FAVAR are consistent. Hence, the estimated factors are assumed to be known in the VAR bootstrap. This approach is in the spirit of the factor augmented regression literature as in \cite{bai2006confidence} and \cite{gonccalves2014bootstrapping}.

\section{Asymptotic properties}\label{sec:theo}
In this section we establish the consistency of the regularized FAVAR model estimator introduced in equation \eqref{pen_favar}. The proofs of the stated theorems are extensions of those in \cite{Daniele2018} and \cite{bai2016estimation}. Subsequently, we state the necessary assumptions.

\begin{assumption}[Data generating process]\label{data_assum}
\leavevmode
\begin{enumerate}[label=(\roman*)]
\item $\left\{e_t, h_t\right\}_{t\geq 1}$ are strictly stationary, for $h_t = (f_t',g_t')'$. In addition, $\frac{1}{T} \sum_{t = 1}^T f_tg_t' = 0$ and $\E{e_{it}} = \E{e_{it}h_{kt}} = 0$, for all $i \leq N$, $k \leq r$ and $t \leq T$.\label{as11}
\item \label{lam_ass} There exists a constant $c > 0$ such that, for all $N$,
\begin{align*}
c^{-1} < \pi_{\min}\left(\frac{\Lambda'\Lambda}{N^{\beta}}\right) \leq \pi_{\max}\left(\frac{\Lambda'\Lambda}{N^{\beta}}\right) < c, \text{ where } 1/2\leq \beta \leq 1.\footnotemark
\end{align*}
\item There exists $r_1, r_2 > 0$ and $b_1, b_2 > 0$, such that for any $s > 0$, $i \leq N$ and $k \leq r$,\label{ass_exp}
\footnotetext{The lower limit 1/2 for $\beta$ is necessary for a consistent estimation of the factors. See \lemref{lem_est_factor} in Section \ref{subsubsec:consistency_lam} in the Appendix.}
\begin{align*}
\Prob{|e_{it}| > s} \leq \exp(-(s/b_1)^{r_1}), \quad \Prob{|h_{kt}| > s} \leq \exp(-(s/b_2)^{r_2})
\end{align*}
\item Define the mixing coefficient:
\begin{align*}
\alpha(T) = \sup_{A\in \mathcal{F}_{-\infty}^0, B\in \mathcal{F}_{T}^{\infty}} \left|\Prob{A}\Prob{B} - \Prob{AB}\right|,
\end{align*}
where $\mathcal{F}_{-\infty}^0$ and $\mathcal{F}_{T}^{\infty}$ denote the $\sigma$-algebras generated by $\{(h_t, e_t): -\infty \leq t \leq 0\}$ and $\{(h_t, e_t): T \leq t \leq \infty\}$\\
Strong mixing: There exist $r_3 > 0$ and $C > 0$ satisfying: for all $T \in \mathcal{Z}^+$,
\label{ass_sm}
\begin{align*}
\alpha(T) \leq \exp(-CT^{r_3})
\end{align*}
\item There exist constants $c_1, c_2, c_3 > 0$ such that $c_2 \leq \pi_{\min}\left(\Sigma_{e0}\right) \leq \pi_{\max}\left(\Sigma_{e0}\right) \leq c_1$ and $\max_{i \leq N} \max_{k \leq r} |\lambda_{ik0}| < c_3$.\label{assum_sig}

\item The factors $h_t = (f_t', g_t')'$ follow the VAR representation in \eqref{favar2}, where $\eta_t$ is an iid process with $\E{\eta_{t}} = 0$ and $\E{\eta_{t}\eta_t'} = \Omega$ and all the roots of the polynomial $\Phi(L) = (I_r - \Phi_1L - \cdots - \Phi_p L^p) = 0$ are outside the unit circle. Moreover, $\eta_s$ and $e_{it}$ are independent for all $i,t,s$. \label{var_assm}
\end{enumerate}
\end{assumption}
The assumptions in \ref{data_assum} impose regularity conditions on the data generating process and are similar to those imposed by \cite{Bai2016b}. Condition \textit{\ref{as11}} induces strict stationarity for $e_t$ and $h_t$ and requires that both terms are uncorrelated. Furthermore, it implies that the observed and latent factors are orthogonal, which simplifies our technicalities. Condition \textit{\ref{lam_ass}} relaxes the pervasiveness assumption commonly imposed in the approximate factor model literature and allows for weak factors. Condition \textit{\ref{ass_exp}} requires exponential-type tails, which allows to use large deviation theory for $\frac{1}{T} \sum_{t = 1}^{T} e_{it} e_{jt} - \sigma_{e, ij}$ and $\frac{1}{T} \sum_{t = 1}^{T} h_{jt} e_{it}$.
Condition \textit{\ref{ass_sm}} imposes a strong mixing condition to allow for weak serial dependence. Further, Condition \textit{\ref{assum_sig}} implies bounded eigenvalues of the idiosyncratic error covariance matrix, which is a common identifying assumption in the factor model framework. Condition \textit{\ref{var_assm}} imposes common VAR assumptions on the factor processes.

To control the sparsity in both $\Lambda$ and $\Sigma_{e}$, we impose the following sparsity assumptions
\begin{assumption}[Sparsity]\label{sparsity_assum}
\leavevmode
\begin{enumerate}[label=(\roman*)]
	\item $L_N = \sum_{k=1}^r \sum_{i = 1}^{N}\1\left\{\lambda_{ik} \neq 0 \right\} = \mathcal{O}\left(N\right),$ \label{spars1}
	\item $S_N = \max_{i \leq N} \sum_{j = 1}^{N} \1\left\{\sigma_{e,ij} \neq 0 \right\}$\label{spars2}, $S_N^2 d_T = o_p(1)$ and $S_N \max(\mu_1, \mu_2) = o(1)$,
\end{enumerate}
where $\1\{\cdot\}$ defines an indicator function that is equal to one if the boolean argument in braces is true.
\end{assumption}
Condition \textit{\ref{spars1}} defines the quantity $L_N$ that represents the number of non-zero elements in the factor loadings matrix $\Lambda$. As the number of factors $r$ are assumed to be constant,\textit{ \ref{spars1}} upper bounds the number of non-zero elements in each column of $\Lambda$ by $N$. Condition \textit{\ref{spars2}} specifies $S_N$ that corresponds to the maximum number of non-zero elements in each row of $\Sigma_{e}$, following the definition in \cite{FanLiaoMincheva2013}.
\begin{theorem}[Consistency of the estimators of the regularized FAVAR model before rotation]\label{theo_consistency_lam}
\leavevmode\\
Under Assumptions \ref{data_assum} and \ref{sparsity_assum} the regularized FAVAR model in \eqref{pen_favar} satisfies for $T$ and $N \to \infty$, the following properties
\begin{align*}
\frac{1}{N} \frob{\tilde{\Lambda} - \Lambda}^2 &= \mathcal{O}_p\left(\max(\mu_1,\mu_2) + \sqrt{\frac{N (d_T + \max(\mu_1,\mu_2))}{L_N}}\right),\\
\frac{1}{N} \frob{\tilde{\Phi}_{e} -\Phi_{e}}^2 &= \mathcal{O}_p\left(\frac{\log N^{\beta}}{N} + \frac{\log N}{T}\right),
\end{align*}
where $d_T = \frac{\log N^{\beta}}{N} + \frac{1}{N^{\beta}} \frac{\log N}{T}$, for $1/2 \leq \beta \leq 1$.\\
Hence, for $\log(N) = o(T)$ and the regularization parameters $\mu_1 = o(1)$, $\mu_2 = o(1)$, we obtain
\begin{align*}
\frac{1}{N} \frob{\tilde{\Lambda} - \Lambda}^2 &= o_p(1), \qquad \frac{1}{N} \frob{\tilde{\Phi}_{e} -\Phi_{e}}^2 = o_p(1).
\end{align*}
Furthermore, for all $t \leq T$:
\begin{align*}
\spec{\tilde{f}_t - f_t} = o_p(1)
\end{align*}
For the second step estimator of the idiosyncratic error covariance matrix, specified in Section \ref{sec:poet}, we get
\begin{align*}
\spec{\tilde{\Sigma}_e^{\tau} - \Sigma_e} = \mathcal{O}_p\left(S_N\sqrt{\max(\mu_1,\mu_2)^2 + \frac{N (d_T + \max(\mu_1,\mu_2))}{L_N}}\right).
\end{align*}
Moreover, for the autoregressive matrices in the dynamic equation \eqref{favar2} we have
\begin{align*}
	\frob{\tilde{\Phi}_i - \Phi_i} = \left(\sum_{t = \bar{p}}^T \eta_t \psi_t'\right) \left(\frac{1}{\bar{T}}\sum_{t = \bar{p}}^T \psi_t \psi_t'\right)^{-1}\left(\iota_i \otimes I_r\right) + o_p(1),
\end{align*}
where $\iota_i$ is the $i$-th column of the $r \times r$ identity matrix.
\end{theorem}
The proof of Theorem \ref{theo_consistency_lam} is given in Section \ref{subsubsec:consistency_lam} in the Appendix. 
Under the imposed regularity conditions Theorem \ref{theo_consistency_lam} establishes the average consistency in the Frobenius norm of the estimators of the factor loadings matrix and idiosyncratic error covariance matrix based on our regularized FAVAR model before rotation. More specifically, we can see that $\Lambda$ and $\Phi$ are estimated consistently, even though we impose the strict factor model structure on $\Sigma_{e}$ in the first step of our estimation procedure. Consequently, the latent factors $f_t$ estimated based on GLS are consistent as long as $1/2 \leq \beta \leq 1$. Intuitively, the lower limit on $\beta$ ensures that the factors are not too weak such that there is still a clear distinction between the common and idiosyncratic component. Furthermore, the second step estimator for the idiosyncratic error covariance matrix introduced in Section \ref{sec:poet} is consistent under the spectral norm as long as $S_N^2 d_T = o_p(1)$ and $S_N \max(\mu_1, \mu_2) = o(1)$. Finally, the autoregressive parameter matrices in the dynamic VAR equation in \eqref{favar2} are consistently estimated under the Frobenius norm as well.

The following theorem summarizes the consistency results for the complete set of estimators based on our regularized FAVAR model after rotating them with matrix $\tilde{A} = \left[\begin{array}{cc}
I_{r_1} &  -\tilde{\Omega}_{fg}\tilde{\Omega}_{gg}^{-1}	\\
0 & I_{r_2} 
\end{array} \right]$.
\begin{theorem}[Consistency of the rotated estimators based on the regularized FAVAR model]\label{theo_consistency_lam_rotated}
\leavevmode
Under Assumptions \ref{data_assum} and \ref{sparsity_assum} the regularized FAVAR model in \eqref{pen_favar} satisfies for $T$ and $N \to \infty$, the following 
\begin{align*}
\max_{i \leq N}\spec{\hat{\lambda}_i^g - \lambda_i^{*g}} &= \mathcal{O}_p\left(\sqrt{\mu_1} + \sqrt{\mu_2} + \sqrt{d_T}\right) \\
\spec{\hat{f}_t - f_t^*} &= \mathcal{O}_p\left(N^{-\beta/2}\right) + \mathcal{O}_p\left(\frac{\mu_1 \sqrt{L_N}}{N^{\beta}}+\sqrt{\frac{N(d_T + \mu_2)}{N^{2\beta}}}\right),  \\
\begin{split}
\frob{\hat{\Phi}_i - \Phi_i^*} &= \left(\sum_{t = \bar{p}}^T \eta_t \psi_t'\right) \left(\frac{1}{\bar{T}}\sum_{t = \bar{p}}^T \psi_t \psi_t'\right)^{-1} \left(\iota_i \otimes I_r\right) \\
&+ \mathcal{O}_p\left(\frac{\mu_1 \sqrt{L_N}}{N^{\beta}}+\sqrt{\frac{N(d_T + \mu_2)}{N^{2\beta}}}\right),
\end{split}
\end{align*}
where $\iota_i$ is the $i$-th column of the $r \times r$ identity matrix and $1/2\leq \beta \leq 1$.
\end{theorem}
The proof of Theorem \ref{theo_consistency_lam_rotated} is given in Section \ref{subsubsec:consistency_lam} in the Appendix. Theorem \ref{theo_consistency_lam_rotated} shows that the estimators of the latent factors, the factor loadings of the observed factors and the coefficients of the dynamic equation in \eqref{favar2} are consistently estimated under the Euclidean and Frobenius norm after rotating them based on matrix $\tilde{A}$. Note that an analysis of the estimated factor loadings of the latent factors $\tilde{\lambda}^f_i$ is redundant, as they are unaffected by the rotation.
 
%
\section{Implementation of the regularized FAVAR estimator}\label{sec:implem}
\subsection{Majorize-minimize expectation maximization algorithm }
For the implementation of the regularized FAVAR model, we employ the majorize-minimize expectation maximization (EM) algorithm by \cite{BienTibshiraniothers2011}. The idea of this algorithm is to replace the optimization of the nonconvex objective function in equation \eqref{neg_log_lik} by a sequence of convex problems that are numerically easy to solve by algorithms for convex optimization. In that respect, we majorize the log-likelihood function in \eqref{neg_log_lik} by the tangent plane of the concave part $\log \left| \tilde{\Sigma} \right|$, which leads to the following expression
\begin{align}
\mathcal{L}^*_m =& \; \log\, \left|  \tilde{\Sigma}_m \right| + \traces{ \left( 2 \tilde{\Sigma}_{H,m} \tilde{\Lambda}_m' \right)\tilde{\Sigma}_m^{-1} \left( \Lambda - \tilde{\Lambda}_m \right)} 
+ \traces{ S_x \Sigma^{-1} }, \label{major_ll}
\end{align}
where $\tilde{\Sigma}_m = \tilde{\Lambda}_m \tilde{\Sigma}_{H,m} \tilde{\Lambda}_m' + \hat{\Phi}_{e,m}$ and $\Sigma = \Lambda \tilde{\Sigma}_{H,m} \Lambda' + \hat{\Phi}_{e,m}$ with $\tilde{\Sigma}_{H,m} = \left[\begin{matrix} I_{r_1} & 0 \\ 0 & \tilde{\Sigma}_{G,m} \\ \end{matrix} \right]$. 
$\tilde{F}$ denotes a initial estimate for the latent factors that can be obtained by unpenalized maximum likelihood and the subscript $m$ is the $m$-th step in the iterative procedure. As we use standardized time series, the covariance matrix of the observed factors reduces to an identity matrix and $\tilde{\Sigma}_{H,m} = I_{r_1+r_2}$. In the following, we augment the majorized log-likelihood function in \eqref{major_ll} by $L_1$-penalty terms for the factor loadings $\Lambda^f$ and $\Lambda^g$, which leads to the following optimization problem for our regularized FAVAR model
\begin{align}
	\begin{split}\label{opt_favar}
	\underset{\{\Lambda \} }{\text{min}} \, &\log\, \left| \tilde{\Sigma}_m \right| + \traces{ \left( 2 \tilde{\Lambda}_m'\right)\tilde{\Sigma}_m^{-1} \left( \Lambda - \tilde{\Lambda}_m \right)} + \traces{ S_x \Sigma^{-1} } \\
	&\quad  + \mu_1 \sum_{i=1}^N \sum_{k=1}^{r_1} \left| \lambda^f_{ik} \right| + \mu_2 \sum_{i=1}^N \sum_{l=1}^{r_2} \left| \lambda^g_{il} \right|.
	\end{split}
\end{align}
As the optimization problem in equation \eqref{opt_favar} has the appealing property of being entirely convex it can be easily solved by a convex optimizer. In this respect, we rely on an efficient projected gradient descent algorithm and optimize the following minimization problem
\begin{align}
\underset{\{\Lambda \} }{\text{min}} \, \frac{1}{2c} \sum_{i=1}^N \sum_{j=1}^r \left( \lambda_{ij} - \tilde{\lambda}_{ij,m} + c \cdot \tilde{D}_{ij,m} \right)^2 + \mu_1 \sum_{i=1}^N \sum_{k=1}^{r_1} \left| \lambda^f_{ik} \right| + \mu_2 \sum_{i=1}^N \sum_{l=1}^{r_2} \left| \lambda^g_{il} \right|, \label{grad_dec}
\end{align}
where $c$ determines the depth of the projection in the gradient decent algorithm\footnote{In all our applications we set $c = 0.01$.} and
\begin{align*}
	\tilde{D}_m = \left[ \tilde{\Sigma}_m^{-1} - \tilde{\Sigma}_m ^{-1} S_x \tilde{\Sigma}_m ^{-1}  \right] \left( 2 \tilde{\Lambda}_m'  \right),
\end{align*}
which corresponds to the first-derivative of $\mathcal{L}^*_m$ with respect to each element of $\Lambda$. The minimization of the objective function \eqref{grad_dec} leads to the following first order conditions with respect to each element of the quantities $\Lambda^f$ and $\Lambda^g$
\begin{align*}
\frac{\partial}{\partial \lambda_{ik}^f} &= \sum_{i=1}^N  \sum_{k=1}^{r_1} \left( \lambda_{ik} - \tilde{\lambda}_{ik} + c \cdot \tilde{D}_{ik,m} \right) + c \cdot \mu_1 \sum_{k=1}^{r_1} \sum_{i=1}^N \nu_{ik}^f \overset{!}{=} 0, \\
\frac{\partial}{\partial \lambda_{il}^g} &= \sum_{i=1}^N \sum_{l=r_1+1}^{r} \left( \lambda_{il} - \tilde{\lambda}_{il} + c \cdot \tilde{D}_{il,m} \right) + c \cdot \mu_2 \sum_{l=r_1+1}^{r} \sum_{i=1}^N \nu_{il}^g \overset{!}{=} 0,
\end{align*}
where $\nu_{ik}^f$ and $\nu_{il}^g$ are the subgradients of $\left| \lambda^f_{ik} \right|$ and $\left| \lambda^g_{il} \right|$, respectively. Hence, solving for a specific $\tilde{\lambda}^f_{ik}$ or $\tilde{\lambda}^g_{ik}$ leads to the following updating formulas for the factor loadings estimates
\begin{align}
\tilde{\lambda}^f_{ik,m+1} &= \mathcal{S} \left( \tilde{\lambda}_{ik,m}^f - c \cdot \tilde{D}_{ik,m} , c \cdot \mu_1 \right), \quad \text{ for } i=1,\dots,N; \; k = 1, \dots, r_1 \\
\tilde{\lambda}^g_{il,m+1} &= \mathcal{S} \left( \tilde{\lambda}_{il,m}^g - c \cdot \tilde{D}_{il,m} , c \cdot \mu_2 \right), \quad\;\, \text{ for } i=1,\dots,N; \; l = 1, \dots, r_2,
\end{align}
where $\mathcal{S}(\, \cdot \,)$ is the soft-thresholding function defined in \eqref{soft_t}.

To obtain an update for the estimate of the covariance matrix of the idiosyncratic error $\Phi_{e}$, we use the formula in the EM algorithm 
suggested by \cite{BaiLi2012}
\begin{align*}
\tilde{\Phi}_{e, m+1} = \text{diag}\left[S_{x} - \tilde{\Lambda}_{m+1} \tilde{\Lambda}_{m}' \left(\tilde{\Lambda}_{m}\tilde{\Lambda}_{m}' + \tilde{\Phi}_{e, m}\right)^{-1} S_{x}\right].
\end{align*}


Our iterative estimation procedure for the regularized FAVAR model is therefore described by the following steps:

Starting from the FAVAR model $X = \Lambda^f F' + \Lambda^g G' + e$

\begin{itemize}
    \item[\textit{Step 1:}] Transform the data by linearly projecting the observed factors $G$ from the observed data. This is achieved by post-multiplying the FAVAR model by $M = I_T - G(G'G)^{-1} G'$. We obtain the transformed data $XM = \Lambda^f F'M + eM$ or $\dot{X} = \Lambda^f \dot{F}' + \dot{e}$.

	\item[\textit{Step 2:}] Set $m = 1$ and obtain an initial estimate for the factor loading matrix $\tilde{\Lambda}_m^f$, factors $\tilde{F}_m$ and the diagonal idiosyncratic error covariance matrix $\tilde{\Phi}_{e,m}$, e.g.\ by using unpenalized MLE on $\dot{X}$. Get an initial estimate of the factor loadings of the observed factor by $\tilde{\Lambda}^g_m = \left(X - \tilde{\Lambda}^f_m \tilde{F}'_m\right) G\left(G'G\right)^{-1}$.
	\item[\textit{Step 3:}] Update $\tilde{\lambda}^f_{ik,m}$ by $ \tilde{\lambda}^f_{ik,m+1} = \mathcal{S}\left(\tilde{\lambda}^f_{ik, m} - c \cdot \tilde{D}_{ik, m}, \, c \cdot \mu_1 \right)$, for $i = 1, \dots, N; \; k = 1,\dots, r_1$ and $\tilde{\lambda}^g_{il,m}$ by $\tilde{\lambda}^g_{il,m+1} = 
	\mathcal{S} \left( \tilde{\lambda}_{il,m} - c \cdot \tilde{D}_{il,m} , c \cdot \mu_2 \right)$, for $i = 1, \dots, N; \; l = 1,\dots, r_2$.
	\item[\textit{Step 4:}] Update $\tilde{\Phi}_{e}$ using the EM algorithm in \cite{BaiLi2012}, according to \\
	$ \tilde{\Phi}_{e, m+1} = \text{diag}\left[S_{x} - \tilde{\Lambda}_{m+1} \tilde{\Lambda}_{m}' \left(\tilde{\Lambda}_{m}\tilde{\Lambda}_{m}' + \tilde{\Phi}_{e, m}\right)^{-1} S_{x}\right] $
	\item[\textit{Step 5:}] If $\spec{\tilde{\Lambda}^f_{m+1} - \tilde{\Lambda}^f_{m}}$ and $\spec{\tilde{\Phi}_{e,m+1} - \tilde{\Phi}_{e,m}}$ are sufficiently small, stop the procedure,  
	otherwise set $m = m + 1$ and return to \textit{Step 3}.
	\item[\textit{Step 6:}] Estimate the latent factors by $\tilde{f}_t = \left(\tilde{\Lambda}^{f'} \tilde{\Phi}_{e}^{-1} \tilde{\Lambda}^{f}\right)^{-1}\tilde{\Lambda}^{f'}\tilde{\Phi}_{e}^{-1}x_t$ ,
	where $\tilde{\Lambda}^{f}$ and $\tilde{\Phi}_{e}$ are the parameter estimates after convergence.
	\item[\textit{Step 7:}] Re-estimate the covariance matrix of the idiosyncratic errors based on the procedure introduced in Section \ref{sec:poet}.
	
	
	\item[\textit{Step 8:}] Based on the estimated and observed factors, $\tilde{h}_t = \left(\tilde{f}_t', g_t'\right)'$, we estimate the following VAR regression:
$\tilde{h}_t = \Phi_1 \tilde{h}_{t-1} + \Phi_2 \tilde{h}_{t-2} + \cdots + \Phi_p \tilde{h}_{t-p} + \eta_t$, we obtain the residuals $\tilde{\eta}_t$ from the previous regression and calculate $\tilde{\Omega} = \frac{1}{T-p} \sum_{t = p+1}^{T} \tilde{\eta}_t\tilde{\eta}_t'$.

	\item[\textit{Step 9:}] For the identification of our regularized FAVAR model, we use the identification restriction \textit{IRa}.  Thus, the factor loadings estimates for the unobserved factors are unchanged $\hat{\Lambda}^{f} = \tilde{\Lambda}^f$. Further, we estimate $\hat{\Lambda}^{g} = \tilde{\Lambda}^{g} + \tilde{\Lambda}^{f} \tilde{\Omega}_{fg}\tilde{\Omega}^{-1}_{gg}$, $\hat{F} = \tilde{F} - G\tilde{\Omega}_{gg}^{-1}\tilde{\Omega}_{gf}$ and the autoregressive parameters by $\hat{\Phi}_i = \tilde{A} \tilde{\Phi}_i \tilde{A}^{-1}$, for $i = 1,\dots,p$, with the rotation matrix
	$\tilde{A} = \left[\begin{array}{cc}
	I_{r_1} & -\tilde{\Omega}_{fg}\tilde{\Omega}_{gg}^{-1}	\\
	0 & I_{r_2} 
	\end{array} \right]$.	

\end{itemize}

\subsection{Selection of the regularization parameters} \label{sec:pen}

We provide a selection criterion to estimate the regularization parameters $\mu_1$ and $\mu_2$. The information criterion is based on an adaptation of the Bayesian information criterion, comparable to the criteria in \cite{Bai2002}, and takes the following form
\begin{align}
IC(\mu_1, \mu_2) = \mathcal{L}\left(\tilde{\Lambda}, S_{\tilde{H}}, \tilde{\Sigma}_e^{\tau}\right) + \kappa_{\mu_1,\mu_2} \sqrt{\frac{\log(2N)}{N} + \frac{\log N}{NT}}, \label{ic}
\end{align}
where $\kappa_{\mu_1,\mu_2}$ denotes the number of non-zero elements in the factor loadings matrix $\tilde{\Lambda}$ for specific $\mu_1$ and $\mu_2$. $\mathcal{L}\left(\tilde{\Lambda}, S_{\tilde{H}}, \tilde{\Sigma}_e^{\tau}\right)$ is the value of the log-likelihood function in equation \eqref{neg_log_lik} evaluated at the estimates of the factor loadings $\tilde{\Lambda}$, the sample covariance matrix of the factors $S_{\tilde{H}}$ and the covariance matrix of the idiosyncratic component $\tilde{\Sigma}_e^{\tau}$. The penalty function in \eqref{ic} relates to the convergence rate of the estimator $\tilde{\Lambda}$ and vanishes for $N,T \to \infty$. In order to select the optimal penalty parameters, we calculate the information criterion in \eqref{ic} for a grid of different values for $\mu_1$ and $\mu_2$ and select the ones that minimize the criterion. The grids for the regularization parameters are set to $\mu_1 = \left[0,0.05,\mu_{1,\max}\right]$ and $\mu_2 = \left[0,0.1,\mu_{2,\max}\right]$, where $\mu_{1,\max}$ and $\mu_{2,\max}$ denote the highest values of the penalty parameters such that we do not select an empty loadings matrix and the Assumptions \ref{sparsity_assum} in Section \ref{sec:theo} are still fulfilled.

%

\section{Empirical illustration: The effects of monetary policy shocks}\label{sec:appl}
In this section we use our regularized FAVAR model to investigate the effects of a monetary policy shock on economic variables. We use the Federal Funds rate (FFR) as the monetary policy instrument. The monetary policy shock is defined as a shock to the innovation of the FFR. In our FAVAR setting, we treat the FFR as an observable factor.

\subsection{Data description}
The data set $x_t$ consists of 126 macroeconomic and financial time series on monthly frequency. It spans the period from January 1985 until December 2016 which results in 384 monthly observations. 
The data set is mainly based on the FRED-MD database by \cite{McCrackenNg2016}. To represent a broad range of economic sectors we augment the data set by manufacturing data, stock market data and various short-term interest rate spreads. The resulting data set is comparable to the one used by \cite{ForniGambetti2010}, \cite{stock2016factor} and \cite{kerssenfischerpuzzling}. A detailed overview can be found in Table \ref{tab:data} in the Appendix. We transform the data such that $x_t$ contains only stationary time series.\footnote{The specific transformations can be found in Table \ref{tab:data} in the Appendix .} Furthermore, without loss of generality we standardize $x_t$.

\subsection{Overview of models}
This section provides an overview of the models that are compared in the empirical analysis.
The regularized FAVAR (\textbf{RFAVAR}) leads to a data driven identification of the latent factors. We obtain an economic structure by shrinking single elements in the factor loadings matrix to zero. Hence, the estimated factors load only on a subset of the observed time series that correspond to different sectors of the economy. For the identification of our model we use scheme \textit{IRa} in Section \ref{sec:ident} that relies on a data-driven identification of the FAVAR model that is invariant to the ordering of the data. Hence, we circumvent imposing a restrictive a priori assumption on the structure of the model. We initialize our model with eight latent factors according to the $IC_1$ selection criterion by \cite{Bai2002} and are left with five latent factors after the regularization.

We set the lag order for our regularized FAVAR model to $p = 12$. In a dynamic factor model setting, this lag order captures the dynamics in the observed and latent factors sufficiently. Concerning the choice of the regularization parameters we use the procedure described in Section \ref{sec:pen}.\footnote{Different robustness checks in our empirical application reveal that the structure of the regularized FAVAR model does not change much for different values of $\mu_2$. Hence, to reduce the computational time, we can fix $\mu_2$ to a value in the interval [0.7, 1.3] and optimize only for values of $\mu_1$.}

The commonly used scheme in the literature to obtain economically identified factors is based on the named factor identification. This corresponds to identification scheme \textit{IRb} in Section \ref{sec:ident}, which implies that the first $r_1$ time series drive the dynamics of the latent factors. In our analysis we denote this model as \textbf{FAVAR-NF} and the following time series are used as naming variables: civilian unemployment rate, consumer price index, industrial production index, S\&P 500 composite index and Baa corporate bond yield (see Table \ref{tab:named1}). This specific choice of time series is guided by the selection of our regularized FAVAR model. More specifically, these are the time series with the highest loadings in absolute value on the estimated factors. Moreover, the selection is also economically sensible as relevant sectors of the economy are represented by those time series. We use $p=12$ lags in the named factor FAVAR model.

In the following, we introduce two additional models that are restricted versions of the general dynamic factor model proposed by \cite{forni2000generalized}. Closely related models are studied by \citeauthor{Stock2002} (\citeyear{Stock2002a}, \citeyear{Stock2002}, \citeyear{stock2005implications}). More specifically, the DFM is given by
\begin{align}
    x_t &= \Lambda^f f_t + e_t,	\label{mes_dfm}\\
    f_t^\ddagger &= \Theta(L) f_{t-1}^\ddagger + u_t,\label{t_dfm}
\end{align}
where $f_t$ in equation \eqref{mes_dfm} is a $(q \times 1)$-vector of latent static factors and $\Theta(L)$ is a $p$-th order lag polynomial. The first model we consider is the pure dynamic factor model (\textbf{DFM}) by \cite{ForniGambetti2010}. In this model $f_t^\ddagger = f_t$, $u_t = R \varepsilon_t$, $R$ is a $(q \times q_1)$-dimensional matrix and $\varepsilon_t$ is $(q_1 \times 1)$-vector of primitive shocks. 

In the second model, we augment the factors in equation \eqref{t_dfm} by observable variables $y_t$. Hence, in this setting $f_t^\ddagger = [f_t', y_t']'$. We denote this model as \textbf{VAR-F}. In comparison to the FAVAR model introduced in equations \eqref{favar1} and \eqref{favar2}, the model specification in \eqref{mes_dfm} only includes latent variables on the right hand side.

For the model specification of the DFM, we follow \cite{ForniGambetti2010} by setting $q = 16$, $q_1 = 4$ and $p = 12$.\footnote{The selection criterion by \cite{Bai2007} also suggests 4 dynamic factors for our dataset.} The identification of the structural shocks relies on a recursive Cholesky scheme, which includes industrial production (IP), consumer prices (CPI), the FFR and the excess bond premium (EBP). A tight monetary policy shock increases the FFR and has no contemporaneous effect on IP and CPI. The FFR can only be affected by industrial production, consumer prices and itself contemporaneously, whereas EBP reacts to shocks to the three others. The model settings for the VAR-F are comparable to \cite{kerssenfischerpuzzling} and are set to $q_1 = 9$  and $p = 12$. As observable variables, we include IP, CPI and the FFR. For the identification of the structural shocks corresponding to the observable variables we use a recursive Cholesky scheme. As above, IP and CPI are not affected by a contractionary monetary policy shock on impact, whereas the FFR increases.

\subsection{Results for the factor model analysis}

To get some insights on the data, we start with an unregularized factor model estimated by PCA. The number of included factors is determined by the \cite{Bai2002} $IC_1$ criterion, where we set the maximum number of allowed factors $(r_{\max})$ to ten.  

Figure \ref{fig:hm_r2_pca} shows $R^2$ of univariate regressions of the latent factors on the observed time series. Of the eight factors that are selected by the criterion, only five factors have blocks with high explanatory power. It is well documented in the literature that the information criteria of \cite{Bai2002} tend to overestimate the true number of factors when there is remaining correlations in the idiosyncratic component (see e.g.\ \cite{Ahn2013} and \cite{Caner2014}). 
Furthermore, the standard factor model is only identified statistically and the estimated factors may not be economically meaningful. For the structural analysis we are interested in the dynamics of a model that is economically interpretable. Hence, it is crucial to economically identify the latent factors in the FAVAR model. 

In our regularized FAVAR model we enhance the economic interpretability by shrinking elements of the factor loadings matrix to zero.
We illustrate the $R^2$ results of univariate regressions associated with a sparse factor structure in Figure \ref{fig:hm_r2_4}. The separation of the latent factors is notably more distinct compared to the factors estimated by PCA. More specifically, we obtain a block structure in the factor loadings which leads to factors that are linked to different sectors in the economy. Furthermore, the RFAVAR model estimates five latent factors. Hence, it shrinks the factors that do not add additional explanatory power to zero. The economic groups associated to the latent factors are: the labor market, prices, industrial production, the stock market and credit spreads. In our context, the labor market factor is mostly linked to employment time series. The corresponding time series plots are shown in Figure \ref{fig:factor_rfavar}. The obtained factors are closely aligned with the underlying economic time series which is due to the fact that the informational content of various time series is used to construct the latent factor estimates. Table \ref{tab:data} provides an overview of which factor drives each variable. A comparable number of factors has been found by \cite{stock2016factor}.

The same analysis is repeated for the named factor FAVAR (FAVAR-NF) model. The resulting $R^2$ are plotted in Figures \ref{fig:hm_r2_s2}. The choice of the naming variables in the FAVAR-NF model leads to factors that have a comparable allocation to economic sectors as the RFAVAR model. This result is anticipated as the selection of the naming variables is guided by the RFAVAR. However, the obtained factor time series in Figures \ref{fig:factor_nfs2} are less aligned to the observed time series in comparison to the RFAVAR. This effect is more pronounced for naming variables that do not represent the entire sector. In these cases the resulting factors are distorted by sector unrelated variables, as in the labor market and credit spread factors.

It is important to note that the model identification crucially depends on the ordering of the observed time series. For example, \cite{stock2016factor} use the following naming scheme in an oil price application: Real personal consumption expenditures, industrial production index, civilian employment rate, S\&P 500 composite index, trade weighted US dollar index major currencies and producer price index.\footnote{As \cite{stock2016factor} work on quarterly frequency they can use data that is not available on monthly frequency. We deviate by using a different proxy for employment and by omitting government spending.} The resulting $R^2$ are plotted in Figure \ref{fig:hm_r2_s1}. Even though the naming variables are chosen based on economic reasoning the selection is rather arbitrary and leads to very different results. Hence, the selection of the naming variables constitutes a restrictive assumption on the structural model. 

We omit this analysis for the DFM and VAR-F models because these frameworks are not concerned with the interpretability of the factors.

\subsection{Results from the impulse response analysis} \label{sec:res_impl}
In the following, we analyze the effects of a monetary policy shock on the observed time series based on our model specification (RFAVAR). Additionally, we compare the results to the ones obtained for the FAVAR-NF, DFM and VAR-F. Hence, our focus lies on the dynamic responses to a shock in the monetary policy instrument (FFR). This corresponds to a partial identification of the structural model, where we identify the monetary policy shock. More specifically, the innovations corresponding to the FFR are contemporaneously uncorrelated with the latent factor innovations, whereas the latter can be contemporaneously correlated. Our RFAVAR framework allows for the structural analysis of the dynamics for both the factors and the underlying observed time series as pointed out in Section \ref{sec:impl}.

In a first step, we elaborate on the effect on the estimated factors for the RFAVAR and the FAVAR-NF model. In both settings, the obtained factors are economically identified and serve as proxies of economic aggregates. The contemporaneous structural effect on the factors are depicted in Table \ref{tab:imp}, where the last column in both panels is associated with the monetary policy shock. 

For our RFAVAR model, the strength of the impact is given in Panel A of Table \ref{tab:imp}. The sign of the contemporaneous impact is in line with economic reasoning. More specifically, the price, industrial production and the labor market factors react negatively on impact in response to a tight monetary policy shock. 

The impulse responses of the factors to a 100bp shock in the innovation of the Federal Funds rate are illustrated in Figure \ref{fig:irf_per_4_1}. The responses of all factors are transitory and stabilize to a new level after one or two years. More precisely, the labor market reacts significantly negative on impact and no longer reacts after 12 months. 
Moreover, as expected by economic rationale, the price level, the level of industrial production and the credit spread factor are impacted negatively by an exogenous increase in the monetary policy rate. The stock market factor does not react significantly to a monetary policy shock for most of the periods. It slightly increases after five months for about one month.




\begin{figure}[!t]
    \centering
    \includegraphics[width=1\linewidth]{irf_per_4_1}
    \caption{{Accumulated impulse responses to a monetary policy shock on the factors for the regularized FAVAR} \label{fig:irf_per_4_1}\\
        \footnotesize The graph shows accumulated impulse responses to a 100bp shock in the innovation of the FFR. The dashed lines correspond to 68\% bootstrap confidence intervals.}
\end{figure}


For the FAVAR-NF model, the impact matrix is reported in Panel B in Table \ref{tab:imp}. Note that we only provide standard errors for the estimates in the last column of Panel B. This follows from the structure of the rotation matrix $\tilde{A}$ given at the end of Section \ref{sec:ident}. As the first $r_1 \times r_1$ block of $\tilde{A}$ depends on $\tilde{\Lambda}_1$, which is kept fixed in the residual-bootstrap based on the dynamic factor equation \eqref{favar2}, we cannot compute bootstrap-based standard errors for the estimated quantities in the first $r_1 \times r_1$ block of the impact matrix. However, this is not harmful for the upcoming analysis, as we are interested in the dynamic effects of a monetary policy shock, whose contemporaneous impacts are given by the last column of the impact matrix.

Furthermore, in Figure \ref{fig:irf_per_3s1_1}, we plot the impulse response functions associated with the named factor scheme. Both, the contemporaneous impacts and the resulting impulse responses on the estimated factors are not always in line with the economic theory. For example, unemployment reacts negatively in response to a tight monetary policy shock, whereas industrial production increases for about ten months. Moreover, we obtain relatively large confidence intervals for the remaining factors leading to statistically insignificant impulse responses. 
Even though the choice of the naming variables is economically motivated, the implied structural dynamics fail to span the monetary policy shock. 

In a second step, we investigate the impulse responses on the observable variables $x_t$ for all models. For the RFAVAR and FAVAR-NF models the contemporaneous impact matrix on $x_t$ is given by $\hat{\Lambda}$ corresponding to the rotated factor loadings estimate, as outlined in Section \ref{sec:impl}. The impact matrix for our regularized FAVAR model is given in Figure \ref{fig:hm_lam_e_4}. The contemporaneous impact of a tightening monetary policy shock is in accordance with economic reasoning. Following an exogenous increase in the FFR, the short term interest rates go up on impact. Further, price and IP as well as employment time series react negatively, whereas unemployment series rise. The stock market plummets and credit conditions deteriorate. 

\begin{figure}[!t]
    \vspace*{-0.8cm}
    \centering
    \includegraphics[width=0.75\linewidth]{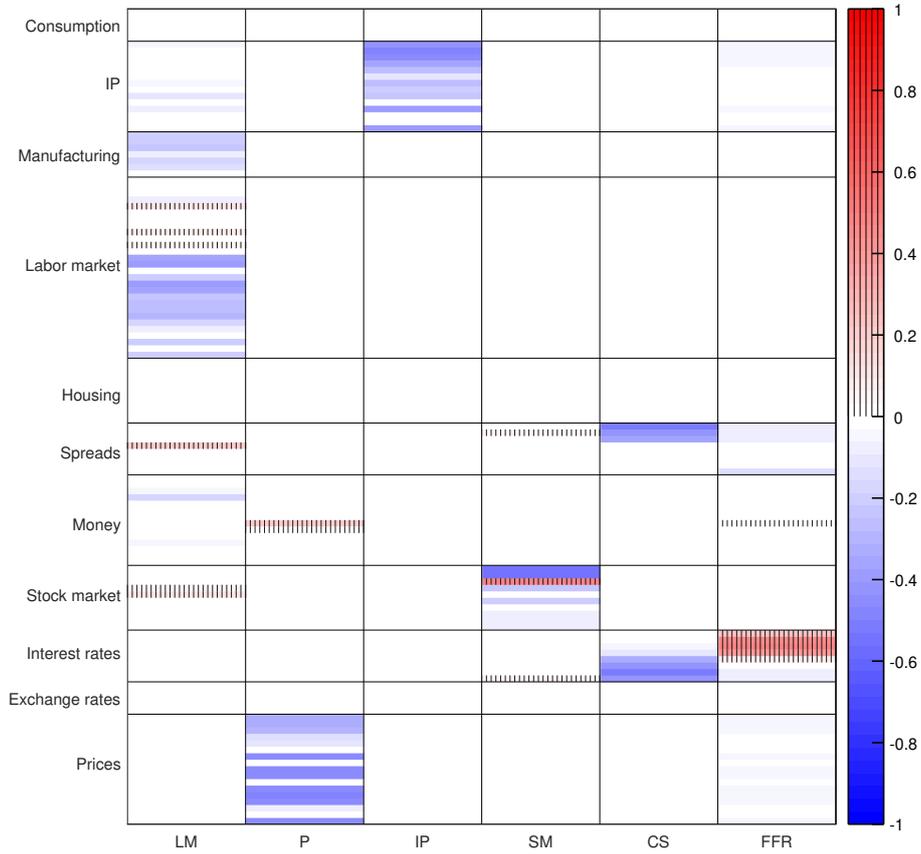}
    \caption{{Impact matrix on the observed variables $x_t$ for the regularized FAVAR} \label{fig:hm_lam_e_4} \\
        \footnotesize This graph shows the contemporaneous impact matrix of a 100bp shock in the factor innovations to the observed time series for the regularized FAVAR model. The factors are abbreviated as follows: 'LM' labor market, 'P' price, 'IP' industrial production, 'SM' stock market, 'CS' credit spread and 'FFR' Federal Funds rate.}
\end{figure}


The contemporaneous impact matrices for the FAVAR-NF model is depicted in Figure \ref{fig:hm_lam_e_3s1}. The economic implications are implausible, as the sign of the effects is often incorrect. For example, the labor market, industrial production and prices react positively on impact for a contractionary monetary policy shock. In contrast to our RFAVAR model, the factor loadings matrix of the FAVAR-NF model is more dense and an economic sector association is not possible.  
Hence, even though the $R^2$ plot of univariate regressions of the named factors on the observed time series in Figure \ref{fig:hm_r2_s1} shows a high explanatory power of the latent factors, the implied rotation scheme yields a factor loadings matrix which we can no longer explain economically.

In the following, we analyze the impulse responses for selected time series, i.e. we concentrate on the effects of a tight monetary policy shock on the consumer price index (CPI), the civilian unemployment rate, the IP index and the three months Treasury bill. Figure \ref{fig:irf_x_per_4_1} shows the level effects for our RFAVAR model on the specific variables. CPI and the IP index series react negatively for around two years, respectively and stabilize to a new level for the remaining periods. The unemployment rate increases for approximately 24 months, whereas the short-term interest rate reacts positively on impact for 12 months. 




\begin{figure}[!t]
    
    \centering
    \includegraphics[width=1\linewidth]{irf_x_per_4_1}
    \caption{{Accumulated impulse responses to a monetary policy shock on the observed variables $x_t$ for the regularized FAVAR} \label{fig:irf_x_per_4_1} \\
        \footnotesize The graph shows accumulated impulse responses to a 100bp shock in the innovation of the FFR. The dashed lines correspond to 68\% bootstrap confidence intervals.}
\end{figure}


The impulse responses for the FAVAR-NF are provided in Figure \ref{fig:irf_x_per_3s1_1}. Overall, the results show contradicting effects contemporaneously as well as over time for the unemployment rate and the IP index. 

For the DFM, the impulse responses for CPI, unemployment, industrial production and the three month Treasury bill rate are illustrated in Figure \ref{fig:irf_x_per_6}. The point estimates for CPI, the IP index and unemployment are economically not reasonable. Moreover, the DFM is sensitive to the number of included latent factors and lags, which leads to a high degree of estimation noise. 

The accumulated impulse responses for the VAR-F model are given in Figure \ref{fig:irf_x_per_5}. The responses of unemployment and IP are in line with economic rationale. Moreover, the impulse response of CPI is very volatile and statistically insignificant. The three month Treasury bill rate reacts negatively over all horizons. The results are very sensitive to the number of included factors and lags. If we do not include enough factors or lags, the impulse response patterns are hardly interpretable (e.g.\ a price puzzle is obtained). 

\subsection{Robustness checks}

We evaluate the robustness of our findings in various ways: First, we use the shadow rate provided by \cite{wu2016measuring} instead of the Federal Funds rate when the policy rate is at the zero lower bound between December 16, 2008 and December 15, 2015.\footnote{The shadow rate is retrieved from Jing Cynthia Wu's website \url{https://sites.google.com/view/jingcynthiawu/shadow-rates}.} Figure \ref{fig:SR} shows both, the shadow rate and the Federal Funds rate, in one graph. The replacement of the policy rate does not affect the results qualitatively for our regularized FAVAR model. In Figure \ref{fig:irf_x_per_4_1_p1_sr} we plot the impulse responses of the observed variables to a shock in the innovation of the shadow rate. The impulse response patterns are very similar to those in Figure \ref{fig:irf_x_per_4_1}. 
For the named factor scheme, the positive impact on the IP index is not significant in Figure \ref{fig:irf_x_per_3s1_1_sr} and the other three impulse responses are not substantially different.
The impulse response patterns of the DFM model in Figure \ref{fig:irf_x_per_6_sr} are similar to those obtained in the previous section. 
Finally, for the VAR-F model in Figure \ref{fig:irf_x_per_5_sr}, the impulse responses of CPI, the unemployment rate and the IP index remain economically plausible when we use the shadow rate instead of the FFR. 
More precisely, CPI reacts significantly negative for the first 6 months, whereas the IP index decreases significantly for approximately three years.  
Overall, using the shadow rate in periods where the FFR is at the zero lower bound improves the results of our competing models slightly but does not affect our RFAVAR model qualitatively.

Second, we augment the maximum number of latent factors used in the initial steps to $r_{\max} = 20$. Figure \ref{fig:hm_r2_pca_r} contains the $R^2$ of univariate regressions of the principal components factors on the observed time series. Qualitatively, it is similar to the one obtained for $r^*=10$ factors depicted in Figure \ref{fig:hm_r2_pca}. In both settings, there are five factors which have separate block-wise explanatory power. Once we employ shrinkage onto the factor loadings, we obtain five latent factors for both initial number of factors. Our method is robust against alternative initial number of factors. This can be seen in Figures \ref{fig:hm_r2_4} and \ref{fig:hm_r2_4_r} which are based on different starting points for the shrinkage but yield similar latent factors and factor loadings estimates. 

Third, we alter the lag order of the regularized FAVAR model to allow for different dynamics in the impulse responses. The resulting response functions of the factors and observed time series of a RFAVAR(2) model can be found in Figures \ref{fig:irf_per_4_1_p2} and \ref{fig:irf_x_per_4_1_p2}. Moreover, the results for a RFAVAR(3) and RFAVAR(6) model are depicted in Figures \ref{fig:irf_per_4_1_p3},  \ref{fig:irf_x_per_4_1_p3} and \ref{fig:irf_per_4_1_p6}, \ref{fig:irf_x_per_4_1_p6}, respectively. The shape of the impulse responses and the evolution over the horizons does not change qualitatively.

Forth, we shorten the time span to the period prior to the global financial crisis. Hence, the data set spans a period for January 1985 to December 2006. We set the initial number of factors to $r_{\max} = 10$. Based on the \cite{Bai2002} $IC_1$ criterion, we initialize the factor model with $r^* = 9$ factors. The $R^2$ of univariate regressions of the PCA based factors onto the observed time series can be found in Figure \ref{fig:hm_r2_pca_2006}. For the pre-crisis period, the block-wise dependence is pronounced for four or five factors. When we impose sparsity onto the factor loadings, we retrieve four latent factors: prices, credit spreads, real activity, and the stock market. Figure \ref{fig:hm_r2_4_2006} shows the $R^2$ of univariate regressions of the regularized factors on the observed time series and Figure \ref{fig:factor_rfavar_2006} shows the time series of the latent factors. A block-structure is still present, however, the real activity sector is represented by the industrial production sector. The labor market sector is no longer a separate latent factor in the pre-crisis period because the two real activity proxies covary heavily. For the entire sample, the labor market factor reacts more sluggish in comparison to industrial production which leads to two separate real activity factors (i.e.\ Figure \ref{fig:factor_rfavar}). 
Our regularized factor model yields economically sensible latent factors and factor loadings estimates in the pre-crisis period. 
Figure \ref{fig:hm_lam_e_4_2006} contains the contemporaneous effects of a 100bp structural shock to the FFR onto the observed time series. The direction of the impact is economically plausible for all cases. Lastly, Figure \ref{fig:irf_x_per_4_1_2006} shows impulse responses of four observed time series: the consumer price index and the industrial production index react negatively in response to a tight monetary policy shock, the level of employment in the manufacturing sector plummets, whereas the three months Treasury bill increases.

%
\section{Conclusions}
\label{sec:con}
In this paper, we propose a regularized factor-augmented vector autoregressive model that enables the factor identification and their economic interpretation. Our estimation procedure relies on a penalized quasi-maximum likelihood approach and is based on a $L_1$-norm regularization of the factor loadings matrix.
The named factor identification scheme conventionally used in the FAVAR literature to identify the factors, imposes specific relations between the factors and observed time series. More specifically, the ordering of the variables determines the structure of the model. The sparse factor loadings structure in our regularized FAVAR model allows for a direct factor identification. Hence, we are able to identify the latent factors in a data-driven manner without implicitly assuming their form a priori in our identification scheme. In this framework, the effects of structural shocks can be investigated on economically meaningful estimated factors and on all observed time series included in the model.

We prove consistency under the Frobenius norm for the estimators of the factor loadings, the latent factors and the covariance matrix of the idiosyncratic component based on the regularized FAVAR model. The factors estimated based on GLS are shown to be consistent. Moreover, the autoregressive parameters in the dynamic equation are consistently estimated as well.

In an empirical application, we investigate the effects of structural monetary policy shocks on a broad range of economically relevant variables. We choose to identify this shock using a joint identification of the factor model and the structural innovations in the vector autoregressive model. We extract five latent factors that relate to the labor market, prices, industrial production, the stock market and credit spreads. Furthermore, the Federal Funds rate is used as an observed factor.
We find impulse response functions which are in line with economic rationale both on the factor aggregates and the observed time series level. More specifically, we do not observe a price puzzle and the obtained impulse response patterns are economically plausible. In particular, following a tight monetary policy, industrial production falls, credit conditions deteriorate and the level of employment decreases.

\setlength{\bibsep}{2pt plus 0.3ex}
\bibliographystyle{ecta}
\bibliography{references_SVAR}

\newpage
\renewcommand\thefigure{\thesection.\arabic{figure}}
\renewcommand{\thetable}{\thesection.\arabic{table}}
\appendix
\begin{appendix}
\section*{Appendix}
\label{sec:appendix}
\section{Proofs}
\label{sec:A_proofs}
\subsection{Consistency of the regularized FAVAR Model Estimator} \label{subsubsec:consistency_lam}
\textbf{Proof.} Theorem \ref{theo_consistency_lam} (Consistency of the estimators of the regularized FAVAR model before rotation)
To establish the consistency of the regularized FAVAR model, we proceed in a similar fashion as in \cite{Daniele2018}. Initially, we define the following penalized log-likelihood
\begin{align}
    \mathcal{L}_p(\bar{\Lambda}, \bar{\Sigma}_{e}) = Q_1(\bar{\Lambda}, \bar{\Sigma}_{e}) + Q_2(\bar{\Lambda}, \bar{\Sigma}_{e}) + Q_3(\bar{\Lambda}, \bar{\Sigma}_{e}), \label{const_lik}
\end{align}
where
\begin{align*}
    Q_1(\bar{\Lambda}, \bar{\Sigma}_{e}) &= \frac{1}{N} \log\left|\bar{\Sigma}_{e}\right| + \frac{1}{N} \trace{S_{e} \bar{\Sigma}_{e}^{-1}} - \frac{1}{N} \log\left|\Sigma_{e}\right| - \frac{1}{N} \trace{S_e \Sigma_{e}^{-1}}\\
    &\quad + \frac{1}{N} \mu_1 \sum_{i = 1}^{N} \sum_{k = 1}^{r_1} \left(\left|\bar{\lambda}_{ik}^f\right| - \left|\lambda_{ik}^f\right| \right) + \frac{1}{N} \mu_2 \sum_{i = 1}^{N} \sum_{l = 1}^{r_2} \left(\left|\bar{\lambda}_{il}^g\right| - \left|\lambda_{il}^g\right| \right),	\\
    Q_2(\bar{\Lambda}, \bar{\Sigma}_{e}) &= \frac{1}{N} \traces{\left(\bar{\Lambda} - \Lambda\right)' \bar{\Sigma}_{e}^{-1} \left(\bar{\Lambda} - \Lambda\right) - \left(\bar{\Lambda} - \Lambda\right)' \bar{\Sigma}_{e}^{-1} \bar{\Lambda} \left(\bar{\Lambda}'\bar{\Sigma}_{e}^{-1} \bar{\Lambda}\right)^{-1} \bar{\Lambda}'\bar{\Sigma}_{e}^{-1}\left(\bar{\Lambda} - \Lambda\right)},	\\
    Q_3(\bar{\Lambda}, \bar{\Sigma}_{e}) &= \frac{1}{N} \log\left|\bar{\Lambda}\bar{\Lambda}' + \bar{\Sigma}_{e}\right| + \frac{1}{N} \trace{S_x\left(\bar{\Lambda}\bar{\Lambda}' + \bar{\Sigma}_{e}\right)^{-1}} - Q_2(\bar{\Lambda}, \bar{\Sigma}_{e}) \\
    &\quad - \frac{1}{N} \log\left|\bar{\Sigma}_{e}\right| - \frac{1}{N} \trace{S_e \bar{\Sigma}_{e}^{-1}},
\end{align*}
and $\Lambda = \left[\Lambda^f \; \Lambda^g\right]$. Hence, the penalized log-likelihood in \eqref{const_lik} can be written as
\begin{align}
    \begin{split}
        \mathcal{L}_p(\bar{\Lambda}, \bar{\Sigma}_{e}) &= \frac{1}{N} \log\left|\bar{\Lambda}\bar{\Lambda}' + \bar{\Sigma}_{e}\right| + \frac{1}{N} \trace{S_x\left(\bar{\Lambda}\bar{\Lambda}' + \bar{\Sigma}_{e}\right)^{-1}} \\
        &\quad - \frac{1}{N} \log\left|\Sigma_{e}\right| - \frac{1}{N} \trace{S_{e} \Sigma_{e}^{-1}}\\
        &\quad + \frac{1}{N} \mu_1 \sum_{i = 1}^{N} \sum_{k = 1}^{r_1} \left(\left|\bar{\lambda}_{ik}^f\right| - \left|\lambda_{ik}^f\right| \right) + \frac{1}{N} \mu_2 \sum_{i = 1}^{N} \sum_{l = 1}^{r_2} \left(\left|\bar{\lambda}_{il}^g\right| - \left|\lambda_{il}^g\right|\right). \label{const_lik2}
    \end{split}
\end{align}
Consider the following set,
\begin{align*}
    \Psi_\delta = \left\{ \left(\Lambda, \Sigma_{e}\right):\;\right. &\delta^{-1} < \pi_{\min}\left(\frac{\Lambda'\Lambda}{N^{\beta}}\right) \leq \pi_{\max}\left(\frac{\Lambda'\Lambda}{N^{\beta}}\right) < \delta, \\
    &\left. \delta^{-1} < \pi_{\min}\left(\Sigma_{e}\right) \leq \pi_{\max}\left(\Sigma_{e}\right) < \delta \right\}, \quad \text{for } 1/2 \leq \beta \leq 1.
\end{align*}
Further, we define $\Phi_e = \text{diag}\left(\Sigma_{e}\right)$, which corresponds to a covariance matrix that contains only the diagonal elements of $\Sigma_{e}$ on its main diagonal.
To control the sparsity in both $\Lambda$ and $\Sigma_{e}$, we impose the following sparsity assumptions as in Assumption \ref{sparsity_assum}:
\begin{align*}
    L_N &= \sum_{k = 1}^{r} \sum_{i = 1}^{N}\1\left\{\lambda_{ik} \neq 0 \right\} = \mathcal{O}\left(N\right),	\\
    S_N &= \max_{i \leq N} \sum_{j = 1}^{N} \1\left\{\sigma_{e,ij} \neq 0 \right\},
\end{align*} 
where $\1\left\{\cdot\right\}$ denotes the indicator function that is equal to one if the boolean argument in braces is true. Hence, $L_N$ is the number of non-zero elements in the factor loadings matrix $\Lambda$ and $S_N$ denotes the maximum number of non-zero elements in each row of $\Sigma_{e}$.

The following lemma will useful for the forthcoming derivations.
\begin{lemma}\label{lemma_fan}
    \leavevmode
    \begin{itemize}
        \item [(i)] $\underset{i,j\leq N}{\max}\left|\frac{1}{T} \sum_{t = 1}^{T} u_{it}u_{jt} - \E{u_{it}u_{jt}}\right| = \mathcal{O}_p\left(\sqrt{(\log N) / T}\right)$,
        \item[(ii)] $\underset{i \leq r, j\leq N}{\max} \left|\frac{1}{T} \sum_{t = 1}^{T} f_{it}u_{jt}\right| = \mathcal{O}_p\left(\sqrt{(\log N) / T}\right)$.
    \end{itemize}
\end{lemma}
\begin{proof}
    See \textit{Lemmas A.3 and B.1} in \cite{Fan2011a}.\\
\end{proof}

Under Assumptions \ref{data_assum} and \ref{sparsity_assum} and \textit{Lemma S.1.2.} in \cite{Daniele2018}, we have that
\begin{align*}
    \sup_{\left(\Lambda, \Sigma_{e}\right) \in \Psi_\delta} \left|Q_3(\Lambda, \Sigma_{e})\right| = \mathcal{O}_p\left(\frac{\log N^{\beta}}{N} + \frac{1}{N^{\beta}}\frac{\log N}{T}\right).
\end{align*}
Furthermore, by \textit{Lemma S.1.3.} in \cite{Daniele2018}, we obtain
\begin{align}
    Q_1\left(\tilde{\Lambda},\tilde{\Sigma}_e\right) + Q_2\left(\tilde{\Lambda},\tilde{\Sigma}_e\right) \leq d_T, \label{q1_q2}
\end{align}
where $d_T = \frac{\log N^{\beta}}{N} + \frac{1}{N^{\beta}} \frac{\log N}{T}$.

In the following, we establish the consistency results for the diagonal idiosyncratic error covariance matrix estimator $\tilde{\Phi}_e$ and the factor loadings estimator $\tilde{\Lambda}$.

\begin{lemma}\label{phi_hat}
    \begin{align*}
        \frac{1}{N} \frob{\tilde{\Phi}_e -\Phi_{e}}^2 = \mathcal{O}_p\left(\frac{\log N}{T} + d_T\right) = o_p(1).
    \end{align*}
\end{lemma}
\begin{proof}
    By the definition of $Q_1\left(\tilde{\Lambda},\tilde{\Sigma}_e\right)$ and $Q_2\left(\tilde{\Lambda},\tilde{\Sigma}_e\right)$ and equation \eqref{q1_q2} we define
    \begin{align}
        B_1 + B_2 \leq d_T, \label{b1_b2}
    \end{align} 
    where
    \begin{align*}
        B_1 &= \frac{1}{N} \log\left|\tilde{\Sigma}_{e}\right| + \frac{1}{N} \trace{S_e \tilde{\Sigma}_e^{-1}} - \frac{1}{N} \log\left|\Sigma_{e}\right| - \frac{1}{N} \trace{S_e \Sigma_{e}^{-1}},\\
        B_2 &= \frac{1}{N} \traces{\left(\tilde{\Lambda} - \Lambda\right)' \tilde{\Sigma}_e^{-1} \left(\tilde{\Lambda} - \Lambda\right) - \left(\tilde{\Lambda} - \Lambda\right)' \tilde{\Sigma}_e^{-1} \tilde{\Lambda} \left(\tilde{\Lambda}'\tilde{\Sigma}_e^{-1} \tilde{\Lambda}\right)^{-1} \tilde{\Lambda}'\tilde{\Sigma}_e^{-1}\left(\tilde{\Lambda} - \Lambda\right)}	\\
        &\quad + \frac{1}{N} \mu_1 \sum_{i = 1}^{N} \sum_{k = 1}^{r_1} \left(\left|\tilde{\lambda}_{ik}^f\right| - \left|\lambda_{ik}^f\right| \right)+ \frac{1}{N} \mu_2 \sum_{i = 1}^{N} \sum_{l = 1}^{r_2} \left(\left|\tilde{\lambda}_{il}^g\right| - \left|\lambda_{il}^g\right|\right).
    \end{align*}
    The result follows by the same argument as in \textit{Lemma S.1.4.} in \cite{Daniele2018}.\\
\end{proof}

To establish the consistency of the factor loadings matrix $\tilde{\Lambda}$, we analyze both sets of factor loadings corresponding to the latent and observed factors separately. Initially, we lower bound the first term in $B_2$ as in \textit{Lemma S.1.5.} in \cite{Daniele2018} and obtain
\begin{align}
    \begin{split}
        &\frac{1}{N} \traces{\left(\tilde{\Lambda} - \Lambda\right)' \tilde{\Sigma}_e^{-1} \left(\tilde{\Lambda} - \Lambda\right) - \left(\tilde{\Lambda} - \Lambda\right)' \tilde{\Sigma}_e^{-1} \tilde{\Lambda} \left(\tilde{\Lambda}'\tilde{\Sigma}_e^{-1} \tilde{\Lambda}\right)^{-1} \tilde{\Lambda}'\tilde{\Sigma}_e^{-1}\left(\tilde{\Lambda} - \Lambda\right)}\\
        &\geq \mathcal{O}_p\left(\frac{L_N}{N}\right) \max_{i \leq N} \spec{\tilde{\lambda}_i - \lambda_{i}}^2. \label{b_2_1}
    \end{split}
\end{align}
Furthermore, the consistency results for $\tilde{\Lambda}^f$ and $\tilde{\Lambda}^g$ are summarized in the following lemma.

\begin{lemma}\label{lem_est_load}
    \begin{align*}
        \max_{i \leq N} \spec{\tilde{\lambda}_i - \lambda_{i}} &= \mathcal{O}_p\left(\max(\mu_1,\mu_2) + \sqrt{\frac{N (d_T + \max(\mu_1,\mu_2))}{L_N}}\right),\\
    \end{align*}
    where $d_T = \frac{\log N^{\beta}}{N} + \frac{1}{N^{\beta}} \frac{\log N}{T}$.
\end{lemma}
\begin{proof}
    By the definition of $B_2$, and equations \eqref{b1_b2} and \eqref{b_2_1} we have
    \begin{align}
        \mathcal{O}_p\left(\frac{L_N}{N}\right)\max_{i \leq N} \spec{\tilde{\lambda}_i - \lambda_{i}}^2 + \frac{1}{N} \mu_1 \sum_{i = 1}^{N} \sum_{k = 1}^{r_1} \left(\left|\tilde{\lambda}_{ik}^f\right| - \left|\lambda_{ik}^f\right| \right)+ \frac{1}{N} \mu_2 \sum_{i = 1}^{N} \sum_{l = 1}^{r_2} \left(\left|\tilde{\lambda}_{il}^g\right| - \left|\lambda_{il}^g\right|\right) \leq d_T. \label{lam_est_l}
    \end{align}
    We start by analyzing $\tilde{\Lambda}^f$. The left hand side of \eqref{lam_est_l} can be further lower bounded by
    \begin{align*}
        \mathcal{O}_p\left(\frac{L_N}{N}\right)\max_{i \leq N} \spec{\tilde{\lambda}_i^f - \lambda_{i}^f}^2 - \frac{1}{N} \mu_1 \sum_{i = 1}^{N} \sum_{k = 1}^{r_1}\left(\left|\lambda_{ik}^f\right| -  \left|\tilde{\lambda}_{ik}^f\right|\right) &\leq d_T + \mu_2 r_2	\\
        \mathcal{O}_p\left(\frac{L_N}{N}\right)\max_{i \leq N} \spec{\tilde{\lambda}_i^f - \lambda_{i}^f}^2 - \frac{1}{N} \mu_1 \sum_{i = 1}^{N} \sum_{k = 1}^{r_1} \left|\tilde{\lambda}_{ik}^f - \lambda_{ik}^f\right| &\leq d_T + \mu_2 r_2	\\
        \mathcal{O}_p\left(\frac{L_N}{N}\right)\max_{i \leq N} \spec{\tilde{\lambda}_i^f - \lambda_{i}^f}^2 - \mathcal{O}\left(\frac{L_N}{N}\right)\mu_1 \sqrt{r_1} \sqrt{\max_{i \leq N} \spec{\tilde{\lambda}_i^f - \lambda_{i}^f}^2} & \leq d_T + \mu_2 r_2
    \end{align*} 
    Solving for $\underset{i \leq N}{\max} \spec{\tilde{\lambda}_i^f - \lambda_{i}^f}$ yields
    \begin{align*}
        \max_{i \leq N} \spec{\tilde{\lambda}_i^f - \lambda_{i}^f} &\leq \mu_1 + \sqrt{\mu_1^2 + \mathcal{O}_p\left(\frac{N (d_T + \mu_2)}{L_N}\right)}	\\
        &\leq \mathcal{O}_p\left(\mu_1 + \sqrt{\frac{N (d_T + \mu_2)}{L_N}}\right).
    \end{align*}
    Equivalently, by the same argument as above for $\tilde{\Lambda}^g$, we obtain
    \begin{align*}
        \max_{i \leq N} \spec{\tilde{\lambda}_i^g - \lambda_{i}^g} = \mathcal{O}_p\left(\mu_2 + \sqrt{\frac{N (d_T + \mu_1)}{L_N}}\right),
    \end{align*}
    which completes the proof.\\	
\end{proof}

The consistency of the latent factor estimator is established in the following lemma.
\begin{lemma}\label{lem_est_factor}
    \begin{align*}
        \frac{1}{T} \sum_{t = 1}^{T} \spec{\tilde{f}_t - f_t}^2 = o_p(1).
    \end{align*}
\end{lemma}
\begin{proof}
    The latent factor estimator in equation \eqref{gls_factors} yields
    \begin{align}\label{diff_f}
        \tilde{f}_t - f_t = -\left(\tilde{\Lambda}^{f'}\tilde{\Phi}_e^{-1}\tilde{\Lambda}^f\right)^{-1}\tilde{\Lambda}^{f'}\tilde{\Phi}_e^{-1}\left(\tilde{\Lambda}^f - \Lambda^f\right)f_t + \left(\tilde{\Lambda}^{f'}\tilde{\Phi}_e^{-1}\tilde{\Lambda}^f\right)^{-1}\tilde{\Lambda}^{f'}\tilde{\Phi}_e^{-1}e_t.
    \end{align}
    As $L_N = \mathcal{O}\left(N\right)$, the first term on the right-hand side is upper bounded by
    \begin{align}
        &\mathcal{O}_p\left(N^{-\beta}\right) \sqrt{\sum_{i = 1}^{N} \spec{\left(\tilde{\Lambda}^{f'}\tilde{\Phi}_e^{-1}\right)_i \left(\tilde{\lambda}^f_i - \lambda_{i}^f\right)}^2}\spec{f_t} \notag\\
        \quad & \leq \mathcal{O}_p\left(N^{-\beta}\right) \sqrt{\mathcal{O}_p\left(\sum_{i = 1}^{N}\spec{\tilde{\lambda}_i^f - \lambda_{i}^f}^2\right)} \notag\\
        & \leq \mathcal{O}_p\left(N^{-\beta}\right) \sqrt{\mathcal{O}_p\left(L_N\max_{i \leq N}\spec{\tilde{\lambda}_i^f - \lambda_{i}^f}^2\right)}\notag\\
        &= \mathcal{O}_p\left(\frac{\sqrt{L_N}}{N^{\beta}}\right) \mathcal{O}_p\left(\mu_1+\sqrt{\frac{N(d_T + \mu_2)}{L_N}}\right) = o_p(1).\label{f1_f_hat}
    \end{align}
    In the following, we bound the second term on the right-hand side of \eqref{diff_f}. For this we analyze the term $\tilde{\Lambda}^{f'}\tilde{\Phi}_e^{-1}e_t$.
    \begin{align*}
        &\mathcal{O}_p\left(N^{-\beta}\right) \frob{\left(\tilde{\Lambda}^{f'}\tilde{\Phi}_e^{-1} - \Lambda^{f'}\Phi_{e}^{-1}\right)e_t}  \\
        & \quad \leq \mathcal{O}_p\left(N^{-\beta}\right) \frob{\left(\tilde{\Lambda}^f - \Lambda^f\right)'\tilde{\Phi}_e^{-1}e_t} +\mathcal{O}_p\left(N^{-\beta}\right) \frob{\Lambda^{f'}\left(\tilde{\Phi}_e^{-1} - \Phi_{e}^{-1}\right)e_t}.
    \end{align*}
    By \lemref{lem_est_load}., the first term is bounded by
    \begin{align}\label{f1}
        &\mathcal{O}_p\left(N^{-\beta}\right) \sqrt{\sum_{i = 1}^{N} \spec{\left(\tilde{\lambda}_i^f - \lambda_{i}^f\right)\left(\tilde{\Phi}_e^{-1}e_t\right)_i}^2} \notag	\\
        &\quad\leq \mathcal{O}_p\left(N^{-\beta}\right) \sqrt{L_N \max_{i \leq N} \spec{\tilde{\lambda}_i^f - \lambda_{i}^f}^2\mathcal{O}_p(1)} \notag\\
        &\quad = \mathcal{O}_p\left(\frac{\sqrt{L_N}}{N^{\beta}}\right)o_p(1) = o_p(1).
    \end{align}
    The second term can be bounded using \lemref{phi_hat}. according to
    \begin{align}\label{f2}
        \mathcal{O}_p\left(N^{-\beta}\right) \frob{\Lambda^{f'}\left(\tilde{\Phi}_e^{-1} - \Phi_{e}^{-1}\right)e_t} &= \mathcal{O}_p\left(N^{-\beta}\right) \sqrt{\sum_{i = 1}^{N} \spec{\left(\Lambda^{f'}\Phi_e^{-1}\right)_i \left(\phi_{ie} - \tilde{\phi}_{ie}\right)\left(\tilde{\Phi}_e^{-1}e_t\right)_i}^2} \notag\\
        &= \mathcal{O}_p \left(\frac{\log N}{N^{\beta}} \frob{\tilde{\Phi}_e - \Phi_{e}}\right) = o_p(1).
    \end{align}
    Hence, equations \eqref{f1_f_hat}, \eqref{f1} and \eqref{f2} yield
    \begin{align*}
        \spec{\tilde{f}_t - f_t} &= \mathcal{O}_p\left(N^{-\beta}\right) \sum_{i = 1}^{N}\spec{\left(\Lambda^{f'}\Phi_{e}^{-1}\right)_i e_{it}} + \mathcal{O}_p\left(\frac{\mu_1 \sqrt{L_N}}{N^{\beta}}+\sqrt{\frac{N(d_T + \mu_2)}{N^{2\beta}}}\right) \\
        &= \mathcal{O}_p\left(N^{-\beta/2}\right) + o_p(1) = o_p(1).
    \end{align*}
\end{proof}

To establish the consistency of the second step estimator of the idiosyncratic error covariance matrix $\tilde{\Sigma}_e^{\tau}$, we first compute the convergence rate of idiosyncratic errors $e_{it}$.

\begin{lemma}\label{lem_uhat}
    \begin{align*}
        \max_{i \leq N} \frac{1}{T} \sum_{t = 1}^{T} \left|\tilde{e}_{it} - e_{it}\right|^2 = \mathcal{O}_p\left(\max(\mu_1,\mu_2)^2 + \frac{N (d_T + \max(\mu_1,\mu_2))}{L_N}\right).
    \end{align*}
\end{lemma}
\begin{proof}
    As $\tilde{e}_{it} - e_{it} = \left(\tilde{\lambda}_i - \lambda_{i}\right)\tilde{f}_t' + \lambda_{i}\left(\tilde{f}_t - f_t\right)'$, we obtain by \lemref{lem_est_load}. and \lemref{lem_est_factor}.
    \begin{align*}
        \max_{i \leq N} \frac{1}{T} \sum_{t = 1}^{T} \left|\tilde{e}_{it} - e_{it}\right|^2 &\leq 2 \max_{i \leq N} \spec{\tilde{\lambda}_i - \lambda_{i}}^2 \frac{1}{T} \sum_{t = 1}^{T} \spec{\tilde{f}_t}^2 + 2 \max_{i \leq N} \spec{\lambda_{i}}^2\frac{1}{T} \sum_{t = 1}^{T} \spec{\tilde{f}_t - f_t}^2	\\
        &\leq \mathcal{O}_p\left(\max_{i \leq N} \spec{\tilde{\lambda}_i - \lambda_{i}}^2\right) + \mathcal{O}_p\left(\frac{1}{T} \sum_{t = 1}^{T} \spec{\tilde{f}_t - f_t}^2\right)	\\
        &= \mathcal{O}_p\left(\max(\mu_1,\mu_2)^2 + \frac{N (d_T + \max(\mu_1,\mu_2))}{L_N}\right).
    \end{align*}
\end{proof}

By \lemref{lem_uhat} we have shown that $\underset{i \leq N}{\max} \,\frac{1}{T} \sum_{t = 1}^{T} \left|\tilde{e}_{it} - e_{it}\right|^2 = o_p(1)$. Hence, by a similar argument as in \textit{Lemma S.1.9.} in \cite{Daniele2018}, we have
\begin{align}
    \max_{i,j\leq N} \left|\tilde{\sigma}_{ij} - \sigma_{ij}\right| =  \mathcal{O}_p\left(\sqrt{\max(\mu_1,\mu_2)^2 + \frac{N (d_T + \max(\mu_1,\mu_2))}{L_N}}\right). \label{sig_e}
\end{align}

In what follows, we are going to determine the convergence rate of the idiosyncratic error covariance matrix estimator based on second step soft-thresholding estimator introduced in Section \ref{sec:poet}.
\begin{lemma}\label{lem_idio}
    \begin{align*}
        \spec{\tilde{\Sigma}_e^{\tau} - \Sigma_e} = \mathcal{O}_p\left(S_N\sqrt{\max(\mu_1,\mu_2)^2 + \frac{N (d_T + \max(\mu_1,\mu_2))}{L_N}}\right).
    \end{align*}
\end{lemma}
\begin{proof}
    The result follows from equation \eqref{sig_e} and Theorem A.1. of \cite{FanLiaoMincheva2013}.\\
\end{proof}

In the following, we derive the convergence rates for the autoregressive matrices $\Phi_i$, for $i = 1, \dots p$, in the dynamic equation \eqref{favar2}. We proceed in the same fashion as in \textit{Proposition A.2.} in \cite{bai2016estimation}.
In order to improve the readability of the upcoming technicalities, we define $\bar{p} = p + 1$ and $\bar{T} = T - p - 1$.

\begin{lemma}\label{lem_phi}
    \leavevmode
    For $\frac{N \log N}{T} = o(1)$, we have
    \begin{enumerate}[label=(\alph*)]
        \item $\frac{1}{\bar{T}}\sum_{t = \bar{p}}^T \tilde{h}_{t-q} \tilde{h}_{t-s}' - \frac{1}{\bar{T}}\sum_{t = \bar{p}}^T h_{t-q} h_{t-s}' = \mathcal{O}_p\left(\frac{\sqrt{L_N}}{N^{\beta}} \left(\mu_1 + \sqrt{\frac{N(d_T + \mu_2)}{L_N}}\right)\right), \; \text{ for } q,s = 0, \dots, p.$ \label{h_a}
        \item $\frac{1}{\bar{T}}\sum_{t = \bar{p}}^T  \left(\tilde{h}_{t-q} - h_{t-q}\right)\tilde{h}_{t-s}' = \mathcal{O}_p\left(\frac{\sqrt{L_N}}{N^{\beta}} \left(\mu_1 + \sqrt{\frac{N(d_T + \mu_2)}{L_N}}\right)\right), \; \text{ for } q,s = 0, \dots, p.$ \label{h_b}
        \item $\frac{1}{\bar{T}}\sum_{t = \bar{p}}^T \eta_t \tilde{h}_{t-q}' - \frac{1}{\bar{T}}\sum_{t = \bar{p}}^T \eta_t h_{t-q}' = \mathcal{O}_p\left(\frac{\sqrt{L_N}}{N^{\beta}} \left(\mu_1 + \sqrt{\frac{N(d_T + \mu_2)}{L_N}}\right)\right), \; \text{ for } q = 0, \dots, p.$ \label{h_c}
    \end{enumerate}
\end{lemma}

\begin{proof}
    We start with expression \ref{h_a}. As $\tilde{h}_t = \left[\tilde{f}_t', g_t'\right]'$ and $h_t = \left[f_t', g_t'\right]'$, the left hand side of \ref{h_a} can be expressed as $\left[\begin{array}{cc}
    W_{11} & W_{12}		\\
    W_{21} & 0
    \end{array}\right]$, where $W_{11} = \frac{1}{\bar{T}}\sum_{t = \bar{p}}^T \left(\tilde{f}_{t-q} - f_{t-q}\right) \left(\tilde{f}_{t-s} - f_{t-s}\right)'$ $+ \frac{1}{\bar{T}}\sum_{t = \bar{p}}^T \left(\tilde{f}_{t-q} - f_{t-q}\right)f_{t-s}' + \frac{1}{\bar{T}}\sum_{t = \bar{p}}^T f_{t-q}\left(\tilde{f}_{t-s} - f_{t-s}\right)'$, $W_{12} = \frac{1}{\bar{T}}\sum_{t = \bar{p}}^T \left(\tilde{f}_{t-q} - f_{t-q}\right)g_{t-s}'$ and $W_{21} = \frac{1}{\bar{T}}\sum_{t = \bar{p}}^T g_{t-q}\left(\tilde{f}_{t-s} - f_{t-s}\right)'$. Now, we analyse each of the three quantities separately.
    Based on \lemref{lem_est_factor}, we can see that the first term on the right hand side of $W_{11}$ is $\mathcal{O}_p\left(\frac{\mu_1^2 L_N + N(d_T + \mu_2)}{N^{2\beta}}\right)$.
    
    Using equation \eqref{diff_f}, we obtain the following expression for the second term in $W_{11}$
    \begin{align}
        \begin{split}\label{f_lag}
            \frac{1}{\bar{T}}\sum_{t = \bar{p}}^T \left(\tilde{f}_{t-q} - f_{t-q}\right)f_{t-s}' = 
            &- \left(\tilde{\Lambda}^{f'} \tilde{\Phi}_e^{-1} \tilde{\Lambda}^f\right)^{-1} \tilde{\Lambda}^{f'} \tilde{\Phi}_e^{-1} \left(\tilde{\Lambda}^f - \Lambda^f\right) \frac{1}{\bar{T}}\sum_{t = \bar{p}}^T f_{t-q} f_{t-s}'		\\
            & + \left(\tilde{\Lambda}^{f'} \tilde{\Phi}_e^{-1} \tilde{\Lambda}^f\right)^{-1} \tilde{\Lambda}^{f'} \tilde{\Phi}_e^{-1} \frac{1}{\bar{T}}\sum_{t = \bar{p}}^T e_{t-q} f_{t-s}'.
        \end{split}
    \end{align} 
    
    The first term on the right hand side of equation \eqref{f_lag} is bounded by the expression $\mathcal{O}_p\left(\frac{\sqrt{L_N}}{N^{\beta}} \left(\mu_1 + \sqrt{\frac{N(d_T + \mu_2)}{L_N}}\right)\right)$, similarly as in \lemref{lem_est_factor}, as $\frac{1}{\bar{T}}\sum_{t = \bar{p}}^T f_{t-q} f_{t-s}' = \mathcal{O}_p(1)$. To bound the second term in equation \eqref{f_lag}, we first analyze the expression $\tilde{\Lambda}^{f'} \tilde{\Phi}_e^{-1} \frac{1}{\bar{T}}\sum_{t = \bar{p}}^T e_{t-q} f_{t-s}'$.	
    \begin{align}
        &\mathcal{O}_p\left(N^{-\beta}\right) \frob{\left(\tilde{\Lambda}^{f'} \tilde{\Phi}_e^{-1} - \Lambda^{f'} \Phi_e^{-1}\right) \frac{1}{\bar{T}}\sum_{t = \bar{p}}^T e_{t-q} f_{t-s}'} 	\notag\\
        &\leq \mathcal{O}_p\left(N^{-\beta}\right) \left[\frob{\left(\tilde{\Lambda}^f - \Lambda^f\right)'\tilde{\Phi}_e^{-1} \frac{1}{\bar{T}}\sum_{t = \bar{p}}^T e_{t-q} f_{t-s}'} + \frob{\Lambda^{f'}\left(\tilde{\Phi}_e^{-1} - \Phi_{e}^{-1}\right)\frac{1}{\bar{T}}\sum_{t = \bar{p}}^T e_{t-q} f_{t-s}'}\right]. \label{f_2}
    \end{align}
    Based on \lemref{lemma_fan} and \lemref{lem_est_load}, the first term in \eqref{f_2} is upper bounded by
    \begin{align*}
        \mathcal{O}_p\left(\frac{\sqrt{L_N}}{N^{\beta}}\right) \sqrt{\max_{i \leq N} \spec{\tilde{\lambda}_{i}^f - \lambda_{i}^f}^2} \mathcal{O}_p\left(\sqrt{\frac{N \log N}{T}}\right).
    \end{align*}
    Hence, as $\frac{N \log N}{T} = o(1)$, the first term of equation \eqref{f_2} is $o_p(1)$. Using \lemref{phi_hat}, the second term in \eqref{f_2} is upper bounded by
    \begin{align*}
        &\mathcal{O}_p\left(N^{-\beta}\right) \frob{\tilde{\Phi}_e - \Phi_{e}} \frob{\frac{1}{\bar{T}}\sum_{t = \bar{p}}^T e_{t-q} f_{t-s}'}	\\
        & \quad \leq o_p(1) \sqrt{\frac{N \log N}{T}} = o_p(1).
    \end{align*} 
    Finally, the second term on the right hand side of equation \eqref{f_lag} is bounded by
    \begin{align*}
        &\mathcal{O}_p\left(N^{-\beta}\right) \sqrt{\max_{i \leq N}\left(\Lambda^{f'}\Phi_{e}^{-1}\right)_i} \frob{\frac{1}{\bar{T}}\sum_{t = \bar{p}}^T e_{t-q} f_{t-s}'}	\\
        & \quad \leq \mathcal{O}_p\left(N^{-\beta}\right) \sqrt{\frac{N \log N}{T}} = o_p(1).
    \end{align*}
    Thus, the second term in $W_{11}$ is $\mathcal{O}_p\left(\frac{\sqrt{L_N}}{N^{\beta}} \left(\mu_1 + \sqrt{\frac{N(d_T + \mu_2)}{L_N}}\right)\right)$. It can be shown that the last term in $W_{11}$ is of the same order. By summarizing these results, we have that $W_{11} = \mathcal{O}_p\left(\frac{\sqrt{L_N}}{N^{\beta}} \left(\mu_1 + \sqrt{\frac{N(d_T + \mu_2)}{L_N}}\right)\right)$. By similar arguments as for $W_{11}$, it can be shown that $W_{12}$ and $W_{21}$ are $\mathcal{O}_p\left(\frac{\sqrt{L_N}}{N^{\beta}} \left(\mu_1 + \sqrt{\frac{N(d_T + \mu_2)}{L_N}}\right)\right)$ as well. Hence, \ref{h_a} follows by this.
    
    Similar as in \cite{bai2016estimation}, expression \ref{h_b} can be written as
    \begin{align*}
        \left[\begin{array}{cc}
            \frac{1}{\bar{T}}\sum_{t = \bar{p}}^T \left(\tilde{f}_{t-q} - f_{t-q}\right)\tilde{f}_{t-s}' & \frac{1}{\bar{T}}\sum_{t = \bar{p}}^T \left(\tilde{f}_{t-q} - f_{t-q}\right)g_{t-s}'	\\
            0 & 0
        \end{array}\right].
    \end{align*}
    Both terms above are of order $\mathcal{O}_p\left(\frac{\sqrt{L_N}}{N^{\beta}} \left(\mu_1 + \sqrt{\frac{N(d_T + \mu_2)}{L_N}}\right)\right)$, as shown in \ref{h_a}. \ref{h_b} follows by this result.
    
    The left hand side of part \ref{h_c} can be expressed as
    \begin{align*}
        \left[\begin{array}{c}
            \frac{1}{\bar{T}}\sum_{t = \bar{p}}^T \eta_t \left(\tilde{f}_{t-q} - f_{t-q}\right)'\\
            0
        \end{array}\right].
    \end{align*}
    By using equation \eqref{diff_f}, $\frac{1}{\bar{T}}\sum_{t = \bar{p}}^T \eta_t \left(\tilde{f}_{t-q} - f_{t-q}\right)'$ can be written as
    \begin{align*}
        \frac{1}{\bar{T}}\sum_{t = \bar{p}}^T \eta_t \left(\tilde{f}_{t-q} - f_{t-q}\right)' = &- \frac{1}{\bar{T}}\sum_{t = \bar{p}}^T \eta_t f_{t-q}' \left(\tilde{\Lambda}^f - \Lambda^f\right)' \tilde{\Phi}_{e}^{-1} \tilde{\Lambda}^f \left(\tilde{\Lambda}^{f'} \tilde{\Phi}_{e}^{-1} \tilde{\Lambda}^f\right)^{-1}	\\
        &+ \frac{1}{\bar{T}}\sum_{t = \bar{p}}^T \eta_t e_{t-q}' \tilde{\Phi}_{e}^{-1} \tilde{\Lambda}^f \left(\tilde{\Lambda}^{f'} \tilde{\Phi}_{e}^{-1} \tilde{\Lambda}^f\right)^{-1},
    \end{align*}	
    which is bounded by
    \begin{align*}
        &\frob{\frac{1}{\bar{T}}\sum_{t = \bar{p}}^T \eta_t f_{t-q}' \left(\tilde{\Lambda}^f - \Lambda^f\right)' \tilde{\Phi}_{e}^{-1} \tilde{\Lambda}^f \left(\tilde{\Lambda}^{f'} \tilde{\Phi}_{e}^{-1} \tilde{\Lambda}^f\right)^{-1}}\\
        & \quad \leq \mathcal{O}_p\left(N^{-\beta}\right) \sqrt{\sum_{i = 1}^N \spec{\left(\tilde{\Lambda}^f_i - \Lambda^f_i\right)'}} = \mathcal{O}_p\left(\frac{\sqrt{L_N}}{N^{\beta}} \left(\mu_1 + \sqrt{\frac{N(d_T + \mu_2)}{L_N}}\right)\right) = o_p(1)
    \end{align*}
    Hence, \ref{h_c} follows from this result.
    
\end{proof}

\begin{lemma}\label{lem_psi}
    \begin{align*}
        \frob{\tilde{\Phi}_i - \Phi_i} = \left(\sum_{t = \bar{p}}^T \eta_t \psi_t'\right) \left(\frac{1}{\bar{T}}\sum_{t = \bar{p}}^T \psi_t \psi_t'\right)^{-1} \left(\iota_i \otimes I_r\right) + \mathcal{O}_p\left(\frac{\sqrt{L_N}}{N^{\beta}} \left(\mu_1 + \sqrt{\frac{N(d_T + \mu_2)}{L_N}}\right)\right),
    \end{align*}	
    where $\Phi = \left(\Phi_1, \dots, \Phi_p\right)$ and $\psi = \left(h_{t-1}', \dots, h_{t-p}'\right)'$.
\end{lemma}

\begin{proof}
    We denote $\tilde{\Phi}$ the estimator of $\Phi$, which is obtained by estimating the regression
    \begin{align*}
        \tilde{h}_{t} = \Phi_1 \tilde{h}_{t-1} + \cdots + \Phi_p \tilde{h}_{t-p} + error,
    \end{align*}
    which yields the estimator
    \begin{align*}
        \tilde{\Phi} = \left(\sum_{t = \bar{p}}^T \tilde{h}_t \tilde{\psi}_t'\right)\left(\sum_{t = \bar{p}}^T \tilde{\psi}_t \tilde{\psi}_t'\right)^{-1},
    \end{align*}
    with $\psi_t = \left(h_{t-1}', \dots, h_{t-p}'\right)'$ and $h_t = \Phi \psi_t + \eta_t$.
    
    By \lemref{lem_phi} and the same arguments as in the proof of \textit{Proposition A.2.} in \cite{bai2016estimation}, we have that
    \begin{align*}
        \frob{\tilde{\Phi} - \Phi} = \left(\sum_{t = \bar{p}}^T \eta_t \psi_t'\right) \left(\frac{1}{\bar{T}}\sum_{t = \bar{p}}^T \psi_t \psi_t'\right)^{-1} + \mathcal{O}_p\left(\frac{\sqrt{L_N}}{N^{\beta}} \left(\mu_1 + \sqrt{\frac{N(d_T + \mu_2)}{L_N}}\right)\right).
    \end{align*}
    The result follows by post-multiplying $\iota_i \otimes I_r$ on both sides. 
    
\end{proof}

In the following, we will focus on the consistency of the parameter estimates after rotation. As the rotation matrix $A$ in Section \ref{sec:ident} depends on the covariance matrix $\Omega$ of the VAR innovations $\eta_t$ in equation \eqref{favar2}, we first concentrate on the consistency of $\Omega$. For this, we introduce the following two lemmas that are similarly established by \cite{bai2016estimation}.
\begin{lemma}\label{mat_al}
    For any two compatible matrices $\mathcal{B}$ and $\mathcal{C}$ and the corresponding estimates $\tilde{\mathcal{B}}$ and $\tilde{\mathcal{C}}$, we have
    \begin{align*}
        \tilde{\mathcal{B}} \tilde{\mathcal{C}}^{-1}\tilde{\mathcal{B}}' - \mathcal{B} \mathcal{C}^{-1}\mathcal{B}' = \left(\tilde{\mathcal{B}} - \mathcal{B}\right)\mathcal{C}^{-1}\mathcal{B}' + \mathcal{B}\mathcal{C}^{-1}\left(\tilde{\mathcal{B}} - \mathcal{B}\right)' - \mathcal{B}\mathcal{C}^{-1}\left(\tilde{\mathcal{C}} - \mathcal{C}\right)\mathcal{C}^{-1}\mathcal{B}' + \mathcal{R},
    \end{align*}
    where
    \begin{align*}
        \mathcal{R} = &- \left(\tilde{\mathcal{B}} - \mathcal{B}\right)\tilde{\mathcal{C}}^{-1}\left(\tilde{\mathcal{C}} - \mathcal{C}\right)\mathcal{C}^{-1}\mathcal{B}' + \left(\tilde{\mathcal{B}} - \mathcal{B}\right)\tilde{\mathcal{C}}^{-1}\left(\tilde{\mathcal{B}} - \mathcal{B}\right)'	\\
        &+\mathcal{B}\tilde{\mathcal{C}}^{-1}\left(\tilde{\mathcal{C}} - \mathcal{C}\right)\mathcal{C}^{-1}\left(\tilde{\mathcal{C}} - \mathcal{C}\right)\mathcal{C}^{-1}\mathcal{B}' - \mathcal{B}\tilde{\mathcal{C}}^{-1}\left(\tilde{\mathcal{C}} - \mathcal{C}\right)\mathcal{C}^{-1}\left(\tilde{\mathcal{B}} - \mathcal{B}\right)'.
    \end{align*}
\end{lemma}
\begin{proof}
    See \textit{Lemma B.1.} in \cite{bai2016estimation}.\\
\end{proof}
\begin{lemma}\label{lem_hmh}
    \begin{align*}
        \frac{1}{\bar{T}}\tilde{H}'M_{\tilde{\Phi}}\tilde{H} - \frac{1}{\bar{T}}H'M_{\Phi}H = \mathcal{O}_p\left(\frac{\sqrt{L_N}}{N^{\beta}} \left(\mu_1 + \sqrt{\frac{N(d_T + \mu_2)}{L_N}}\right)\right),
    \end{align*}
    where
    \begin{align*}
        \frac{1}{\bar{T}} \tilde{H}'M_{\tilde{\Phi}}\tilde{H} = \frac{1}{\bar{T}} \sum_{t = \bar{p}}^T\tilde{h}_t\tilde{h}_t' - \left( \frac{1}{\bar{T}} \sum_{t = \bar{p}}^T \tilde{h}_t\tilde{\psi}_t'\right)\left( \frac{1}{\bar{T}} \sum_{t = \bar{p}}^T \tilde{\psi}_t\tilde{\psi}_t'\right)^{-1}\left( \frac{1}{\bar{T}} \sum_{t = \bar{p}}^T \tilde{\psi}_t\tilde{h}_t'\right),
    \end{align*}
    and $\frac{1}{\bar{T}}H'M_{\Phi}H$ is defined similarly.
\end{lemma}
\begin{proof}
    The proof is conducted as in \cite{bai2016estimation}. Hereby, we consider the following three expressions that are bounded using \lemref{lem_phi} \ref{h_a}.
    \begin{align*}
        \frac{1}{\bar{T}} \sum_{t = \bar{p}}^T \tilde{h}_t\tilde{h}_t' - \frac{1}{\bar{T}} \sum_{t = \bar{p}}^T h_t h_t' &= \mathcal{O}_p\left(\frac{\sqrt{L_N}}{N^{\beta}} \left(\mu_1 + \sqrt{\frac{N(d_T + \mu_2)}{L_N}}\right)\right), \\
        \frac{1}{\bar{T}} \sum_{t = \bar{p}}^T \tilde{h}_t\tilde{\psi}_t' - \frac{1}{\bar{T}} \sum_{t = \bar{p}}^T h_t\psi_t' &=	\mathcal{O}_p\left(\frac{\sqrt{L_N}}{N^{\beta}} \left(\mu_1 + \sqrt{\frac{N(d_T + \mu_2)}{L_N}}\right)\right), \\
        \frac{1}{\bar{T}} \sum_{t = \bar{p}}^T \tilde{\psi}_t\tilde{\psi}_t' - \frac{1}{\bar{T}} \sum_{t = \bar{p}}^T \psi_t\psi_t' &= \mathcal{O}_p\left(\frac{\sqrt{L_N}}{N^{\beta}} \left(\mu_1 + \sqrt{\frac{N(d_T + \mu_2)}{L_N}}\right)\right).
    \end{align*}
    The result follows based on the above expressions and \lemref{mat_al}.
    
\end{proof}

The covariance matrix estimator of the innovations $\eta_t$ can be bounded according to the following lemma.
\begin{lemma}\label{lem_omega}
    \begin{align*}
        \frob{\tilde{\Omega} - \Omega} = \mathcal{O}_p\left(\frac{\sqrt{L_N}}{N^{\beta}} \left(\mu_1 + \sqrt{\frac{N(d_T + \mu_2)}{L_N}}\right)\right).
    \end{align*}
\end{lemma}
\begin{proof}
    The estimator of $\tilde{\eta}_t$ is defined as
    \begin{align*}
        \tilde{\Omega} = \frac{1}{\bar{T}} \sum_{t = \bar{p}}^T \tilde{\eta}_t\tilde{\eta}_t',
    \end{align*}
    where $\tilde{\eta}_t$ are the residuals of the regression $\tilde{h}_t = \Phi_1 \tilde{h}_{t-1} + \cdots + \Phi_p \tilde{h}_{t-p} + error$. Hence,
    \begin{align*}
        \tilde{\eta}_t = \tilde{h}_t - \left(\sum_{t = \bar{p}}^T  \tilde{h}_t \tilde{\psi}_t'\right)\left(\sum_{t = \bar{p}}^T  \tilde{\psi}_t \tilde{\psi}_t'\right)^{-1} \tilde{\psi}_t,
    \end{align*}
    where $\tilde{\psi}_t = \left(\tilde{h}_{t-1}', \dots, \tilde{h}_{t-p}'\right)'$. Based on the previous result we have
    \begin{align}
        \tilde{\Omega} = \frac{1}{\bar{T}} \sum_{t = \bar{p}}^T \tilde{h}_t\tilde{h}_t' - \left(\frac{1}{\bar{T}} \sum_{t = \bar{p}}^T  \tilde{h}_t \tilde{\psi}_t'\right)\left(\frac{1}{\bar{T}} \sum_{t = \bar{p}}^T  \tilde{\psi}_t \tilde{\psi}_t'\right)^{-1}\left(\frac{1}{\bar{T}} \sum_{t = \bar{p}}^T  \tilde{\psi}_t \tilde{h}_t'\right).\label{omeg_est}
    \end{align} 
    Hence, given \eqref{omeg_est} we obtain
    \begin{align}\label{omega_est1}
        \begin{split}
            \tilde{\Omega} - \Omega &= \frac{1}{\bar{T}} \sum_{t = \bar{p}}^T \tilde{h}_t\tilde{h}_t' - \left(\frac{1}{\bar{T}} \sum_{t = \bar{p}}^T  \tilde{h}_t \tilde{\psi}_t'\right)\left(\frac{1}{\bar{T}} \sum_{t = \bar{p}}^T  \tilde{\psi}_t \tilde{\psi}_t'\right)^{-1}\left(\frac{1}{\bar{T}} \sum_{t = \bar{p}}^T  \tilde{\psi}_t \tilde{h}_t'\right) \\
            &- \frac{1}{\bar{T}} \sum_{t = \bar{p}}^T h_t h_t' + \left(\frac{1}{\bar{T}} \sum_{t = \bar{p}}^T  h_t \psi_t'\right)\left(\frac{1}{\bar{T}} \sum_{t = \bar{p}}^T  \psi_t \psi_t'\right)^{-1}\left(\frac{1}{\bar{T}} \sum_{t = \bar{p}}^T  \psi_t h_t'\right).
        \end{split}
    \end{align}
    The expression \eqref{omega_est1} is bounded by \lemref{lem_hmh}. These results yield
    \begin{align*}
        \frob{\tilde{\Omega} - \Omega} = \mathcal{O}_p\left(\frac{\sqrt{L_N}}{N^{\beta}} \left(\mu_1 + \sqrt{\frac{N(d_T + \mu_2)}{L_N}}\right)\right).
    \end{align*} 
\end{proof}

The estimated rotation matrix $\tilde{A} = \left[\begin{array}{cc}
I_{r_1} &  -\tilde{\Omega}_{fg}\tilde{\Omega}_{gg}^{-1}	\\
0 & I_{r_2} 
\end{array} \right]$ can be bounded using the following lemma.
\begin{lemma}\label{lem_a21}
    \begin{align*}
        \frob{\tilde{A}_{21} - A_{21}} =\mathcal{O}_p\left(\frac{\sqrt{L_N}}{N^{\beta}} \left(\mu_1 + \sqrt{\frac{N(d_T + \mu_2)}{L_N}}\right)\right), 
    \end{align*}
    where $\tilde{A}_{21} = - \tilde{\Omega}_{fg}\tilde{\Omega}_{gg}^{-1}$ and $A_{21} = - \Omega_{fg}\Omega_{gg}^{-1}$.
\end{lemma}
\begin{proof}
    Given the definitions of $\tilde{A}_{21}$ and $A_{21}$, we obtain
    \begin{align}
        \tilde{A}_{21} - A_{21} &= \Omega_{fg}\Omega_{gg}^{-1} - \tilde{\Omega}_{fg}\tilde{\Omega}_{gg}^{-1}	\notag\\
        &= - \left(\tilde{\Omega}_{fg} - \Omega_{fg}\right)\left(\tilde{\Omega}_{gg}^{-1} - \Omega_{gg}^{-1}\right) - \Omega_{fg} \left(\tilde{\Omega}_{gg}^{-1} - \Omega_{gg}^{-1}\right) - \left(\tilde{\Omega}_{fg} - \Omega_{fg}\right) \Omega_{gg}^{-1}\label{a21_1}.
    \end{align}
    Based on \lemref{lem_omega}, we can upper bound the expression \eqref{a21_1} by
    \begin{align*}
        \frob{\tilde{A}_{21} - A_{21}} &\leq - \frob{\tilde{\Omega}_{fg} - \Omega_{fg}}\frob{\tilde{\Omega}_{gg}^{-1} - \Omega_{gg}^{-1}} - \spec{\Omega_{fg}} \frob{\tilde{\Omega}_{gg}^{-1} - \Omega_{gg}^{-1}} - \frob{\tilde{\Omega}_{fg} - \Omega_{fg}} \spec{\Omega_{gg}^{-1}}    \\
        &= \mathcal{O}_p\left(\frac{\sqrt{L_N}}{N^{\beta}} \left(\mu_1 + \sqrt{\frac{N(d_T + \mu_2)}{L_N}}\right)\right).
    \end{align*}
\end{proof}

\begin{lemma}
    \begin{align*}
        \max_{i \leq N}\spec{\hat{\lambda}_i^g - \lambda_i^{*g}} = \mathcal{O}_p\left(\sqrt{\mu_1} + \sqrt{\mu_2} + \sqrt{d_T}\right).
    \end{align*}
\end{lemma}
\begin{proof}
    The rotated factor loadings of the observed factors are defined as $\hat{\lambda}_i^g = \tilde{\lambda}_i^g - \tilde{\lambda}_i^f\tilde{A}_{21}$ and $\lambda_i^{*g} = \lambda_i^g - \lambda_i^f A_{21}$. Hence, their difference can be expressed as
    \begin{align}
        \hat{\lambda}_i^g - \lambda_i^{*g} = \tilde{\lambda}_i^g - \lambda_i^g - \left(\tilde{\lambda}_i^f - \lambda_i^f\right)A_{21} - \lambda_i^f \left(\tilde{A}_{21} - A_{21}\right) - \left(\tilde{\lambda}_i^f - \lambda_i^f\right) \left(\tilde{A}_{21} - A_{21}\right).
    \end{align}
    The above expression can be bounded using \lemref{lem_est_load} and \lemref{lem_a21} according to
    \begin{align*}
        \max_{i \leq N}\spec{\hat{\lambda}_i^g - \lambda_i^{*g}} &\leq \max_{i \leq N}\spec{\tilde{\lambda}_i^g - \lambda_i^g } - \max_{i \leq N}\spec{\tilde{\lambda}_i^f - \lambda_i^f}\frob{A_{21}} \\
        &- \max_{i \leq N}\spec{\lambda_i^f} \frob{\tilde{A}_{21} - A_{21}} - \max_{i \leq N}\spec{\tilde{\lambda}_i^f - \lambda_i^f} \frob{\tilde{A}_{21} - A_{21}}\\
        &= \mathcal{O}_p\left(\mu_1 + \mu_2 + \sqrt{\frac{N (d_T + \mu_1)}{L_N}} + \sqrt{\frac{N (d_T + \mu_2)}{L_N}}\right)    \\
        &\leq \mathcal{O}_p\left(\sqrt{\mu_1} + \sqrt{\mu_2} + \sqrt{d_T}\right). 
    \end{align*}
\end{proof}

\begin{lemma}
    \begin{align*}
        \spec{\hat{f}_t - f_t^*} = \mathcal{O}_p\left(N^{-\beta/2}\right) + \mathcal{O}_p\left(\frac{\mu_1 \sqrt{L_N}}{N^{\beta}}+\sqrt{\frac{N(d_T + \mu_2)}{N^{2\beta}}}\right).
    \end{align*}
\end{lemma}
\begin{proof}
    The rotated unobserved factors are defined as $\hat{f}_t = \tilde{f}_t + \tilde{A}_{21} g_t$ and $f_t^* = f_t + A_{21} g_t$. Thus, the difference of both terms is given by
    \begin{align}
        \hat{f}_t - f_t^* = \tilde{f}_t - f_t + \left(\tilde{A}_{21} - A_{21}\right) g_t.\label{f_diff}
    \end{align}
    Using the euclidean norm, we determine the upper bound of \eqref{f_diff}, using \lemref{lem_est_factor} and \lemref{lem_a21}, by
    \begin{align*}
        \spec{\hat{f}_t - f_t^*} &\leq \spec{\tilde{f}_t - f_t} + \frob{\tilde{A}_{21} - A_{21}} g_t.   \\
        &= \mathcal{O}_p\left(N^{-\beta/2}\right) + \mathcal{O}_p\left(\frac{\mu_1 \sqrt{L_N}}{N^{\beta}}+\sqrt{\frac{N(d_T + \mu_2)}{N^{2\beta}}}\right) = o_p(1).
    \end{align*}
\end{proof}

Finally, for the rotated coefficient matrices of the dynamic equation in \eqref{favar2} we establish the following lemma.
\begin{lemma}
    \begin{align*}
        \frob{\hat{\Phi}_i - \Phi_i^*} = \left(\sum_{t = \bar{p}}^T \eta_t \psi_t'\right) \left(\frac{1}{\bar{T}}\sum_{t = \bar{p}}^T \psi_t \psi_t'\right)^{-1} \left(\iota_i \otimes I_r\right)+ \mathcal{O}_p\left(\frac{\sqrt{L_N}}{N^{\beta}} \left(\mu_1 + \sqrt{\frac{N(d_T + \mu_2)}{L_N}}\right)\right),
    \end{align*}
    where $\iota_i$ is the $i$-th column of the $r \times r$ identity matrix.
\end{lemma}
\begin{proof}
    Given the definitions $\hat{\Phi}_i = \tilde{A}\tilde{\Phi}_i\tilde{A}^{-1}$ and $\Phi_i^* = A \Phi_i A^{-1}$, we obtain
    \begin{align}
        \begin{split}
            \hat{\Phi} - \Phi^* &= \left(\tilde{A} - A\right) \left(\tilde{\Phi} - \Phi\right)\left(\tilde{A}^{-1} - A^{-1}\right) + A\left(\tilde{\Phi} - \Phi\right)\left(\tilde{A}^{-1} - A^{-1}\right) \\
            &\quad + \left(\tilde{A} - A\right) \left(\tilde{\Phi} - \Phi\right) A^{-1} + A\left(\tilde{\Phi} - \Phi\right)A^{-1}	+ \left(\tilde{A} - A\right) \Phi A^{-1}\\
            &\quad + \left(\tilde{A} - A\right) \Phi \left(\tilde{A}^{-1} - A^{-1}\right) + A \Phi \left(\tilde{A}^{-1} - A^{-1}\right) 
        \end{split}
    \end{align}
    Based on \lemref{lem_psi} and \lemref{lem_a21}, we get the expression
    \begin{align*}
        \frob{\hat{\Phi} - \Phi^*} &\leq \frob{\tilde{A} - A} \frob{\tilde{\Phi} - \Phi} \frob{\tilde{A}^{-1} - A^{-1}} + \spec{A} \frob{\tilde{\Phi} - \Phi} \frob{\tilde{A}^{-1} - A^{-1}} \\
        &\quad + \frob{\tilde{A} - A} \frob{\tilde{\Phi} - \Phi} \spec{A^{-1}} + \spec{A}\frob{\tilde{\Phi} - \Phi}\spec{A^{-1}}	+ \frob{\tilde{A} - A} \spec{\Phi A^{-1}}\\
        &\quad + \frob{\tilde{A} - A}\spec{ \Phi} \frob{\tilde{A}^{-1} - A^{-1}} + \spec{A \Phi} \frob{\tilde{A}^{-1} - A^{-1}}  \\
        &= \left(\sum_{t = \bar{p}}^T \eta_t \psi_t'\right) \left(\frac{1}{\bar{T}}\sum_{t = \bar{p}}^T \psi_t \psi_t'\right)^{-1} + \mathcal{O}_p\left(\frac{\sqrt{L_N}}{N^{\beta}} \left(\mu_1 + \sqrt{\frac{N(d_T + \mu_2)}{L_N}}\right)\right).
    \end{align*}
    The result follows by post-multiplying $\iota_i \otimes I_r$ on both sides.
    
\end{proof}

\newpage
\section{Tables}
\setcounter{table}{0}
\label{sec:A_tables}
\renewcommand{\arraystretch}{0.8}
\begin{ThreePartTable}
	\setlength{\extrarowheight}{-2pt}
	\begin{TableNotes}
		\footnotesize
		\singlespacing
		\item \leavevmode\kern-\scriptspace\kern-\labelsep 
		Note: This table shows the different observed variables on monthly frequency for a sampling period January 1985 until December 2016. The data is retrieved from the \cite{McCrackenNg2016} FRED-MD data base and datastream. The transformation codes are labeled as follows: 1 =	no transformation, 2	= $\Delta x_t$, 3 = $\Delta^2 x_t$, 4 = $\text{log}(x_t)$, 5 = $\Delta \text{log}(x_t)$, 6	= $\Delta^2 \text{log}(x_t)$. The last two columns are associated with our regularized FAVAR model. The column ``Latent factor'' shows the latent factors the specific time series load on. The latent factors are denotes as 1 = labor market, 2 = price, 3 = industrial production, 4 = stock market, 5 = credit spread. The column ``FFR'' denotes the estimated loadings onto the observable factor.
	\end{TableNotes}

\begin{longtable}{lp{9cm}ccc}
	\caption{Data and Transformations}\\
	\hhline{=====}
	\multicolumn{2}{c}{\multirow{2}[0]{*}{Description of the variables}} & Transformation & Latent & \multirow{2}[0]{*}{FFR} \\
	\multicolumn{2}{c}{} & code  & factor &  \\
	\hline
	\endfirsthead
	\multicolumn{5}{c}%
	{\tablename\ \thetable\ -- \textit{cont.}} \\
	\hhline{=====}
	\multicolumn{2}{c}{\multirow{2}[0]{*}{Description of the variables}} & Transformation & Latent & \multirow{2}[0]{*}{FFR} \\
	\multicolumn{2}{c}{} & code  & factor &  \\
	\hline
	\endhead
	\hline \multicolumn{5}{r}{\textit{Continued on next page}}
	\endfoot
	\hhline{=====}
	\insertTableNotes
	\endlastfoot       
    1     & real personal income & 5     &       & 0 \\
	2     & real personal income ex transfer receipts  & 5     &       & 0 \\
	3     & real personal consumption expenditures & 5     &       & 0 \\
	4     & real manufacturing and trade industries sales & 5     & 5     & -0.0025 \\
	5     & retail and food services sales & 5     &       & 0 \\
	6     & IP index & 5     & 1,5   & -0.0444 \\
	7     & IP final products and nonindustrial supplies & 5     & 1,5   & -0.0494 \\
	8     & IP final products & 5     & 5     & -0.0473 \\
	9     & IP consumer goods & 5     & 5     & -0.0362 \\
	10    & IP durable consumer goods & 5     & 5     & -0.0267 \\
	11    & IP nondurable consumer goods & 5     & 5     & -0.0103 \\
	12    & IP business equipment & 5     & 1,5   & -0.0283 \\
	13    & IP materials & 5     & 1,5   & -0.0230 \\
	14    & IP durable materials & 5     & 1,5   & -0.0300 \\
	15    & IP nondurable materials & 5     & 5     & -0.0026 \\
	16    & IP manufacturing & 5     & 1,5   & -0.0445 \\
	17    & IP residuential utilities & 5     &       & 0 \\
	18    & IP fuels & 5     &       & 0 \\
	19    & capacity utilization & 2     & 1,5   & -0.0402 \\
	20    & US ISM Manufacturers survey: production index & 1     & 1     & -0.0110 \\
	21    & US ISM Manufacturers survey: employment index & 1     & 1     & -0.0119 \\
	22    & US ISM Purchasing Managers Index & 1     & 1     & -0.0138 \\
	23    & US ISM Manufacturers survey: supplier delivery index & 1     & 1     & -0.0038 \\
	24    & US ISM Manufacturers survey: new orders index  & 1     & 1     & -0.0096 \\
	25    & US ISM Manufacturers survey: inventories index & 1     & 1     & -0.0075 \\
	26    & US ISM Manufacturers survey: prices paid index & 1     & 1,2   & -0.0021 \\
	27    & help wanted index & 2     &       & 0 \\
	28    & ratio of help wanted/number unemployed & 2     & 1     & -0.0005 \\
	29    & civilian labor force & 5     &       & 0 \\
	30    & civilian employment & 5     & 1     & -0.0041 \\
	31    & civilian unemployment rate & 2     & 1     & 0.0060 \\
	32    & average duration of unemployment & 2     &       & 0 \\
	33    & civilians unemployed - less than 5 weeks & 5     &       & 0 \\
	34    & civilians unemployed for 5-14 weeks & 5     &       & 0 \\
	35    & civilians unemployed - 15 weeks and over & 5     & 1     & 0.0050 \\
	36    & civilians unemployed for 15-26 weeks & 5     &       & 0 \\
	37    & civilians unemployed for 27 weeks and over & 5     & 1     & 0.0030 \\
	38    & initial claims & 5     &       & 0 \\
	39    & All employees: total nonfarm & 5     & 1     & -0.0194 \\
	40    & All employees: goods-producing industries & 5     & 1     & -0.0222 \\
	41    & All employees: mining  & 5     &       & 0 \\
	42    & All employees: construction & 5     & 1     & -0.0120 \\
	43    & All employees: manufacturing & 5     & 1,5   & -0.0217 \\
	44    & All employees: durable goods & 5     & 1,5   & -0.0219 \\
	45    & All employees: nondurable goods & 5     & 1     & -0.0130 \\
	46    & All employees: service-providing industries & 5     & 1     & -0.0139 \\
	47    & All employees: trade, transportation and utilities & 5     & 1     & -0.0152 \\
	48    & All employees: wholesale trade & 5     & 1     & -0.0159 \\
	49    & All employees: retail trade & 5     & 1     & -0.0098 \\
	50    & All employees: financial activities & 5     & 1     & -0.0044 \\
	51    & All employees: government & 5     &       & 0 \\
	52    & Average hourly earnings: goods-producing & 1     & 1     & -0.0110 \\
	53    & Average weekly overtime hours: manufacturing & 2     &       & 0 \\
	54    & Average weekly hours: manufacturing & 1     & 1     & -0.0121 \\
	55    & Housing starts: total new privately owed & 5     &       & 0 \\
	56    & Housing starts: NE & 5     &       & 0 \\
	57    & Housing starts: MW & 5     &       & 0 \\
	58    & Housing starts: S & 5     &       & 0 \\
	59    & Housing starts: W & 5     &       & 0 \\
	60    & New private housing permits & 5     &       & 0 \\
	61    & New private housing permits: NE & 5     &       & 0 \\
	62    & New private housing permits: MW & 5     &       & 0 \\
	63    & New private housing permits: S & 5     &       & 0 \\
	64    & New private housing permits: W & 5     &       & 0 \\
	65    & Moodys Seasoned Aaa Corporate Bond Yield, Percent & 1     & 3     & -0.0912 \\
	66    & Moodys Seasoned Baa Corporate Bond Yield, Percent & 2     & 3,4   & -0.0722 \\
	67    & 30-Year Fixed Rate Mortgage Average in the United States, Percent & 2     & 3     & -0.0644 \\
	68    & Excess Bond Premium & 1     & 1     & 0.0119 \\
	69    & Spread UK-US & 2     &       & 0 \\
	70    & Spread CAN-US & 2     &       & 0 \\
	71    & Spread SW-US & 2     &       & 0 \\
	72    & Spread JPN-US & 2     &       & -0.1383 \\
	73    & New orders for durable goods & 5     &       & 0 \\
	74    & New orders for nondefense capital goods & 5     &       & 0 \\
	75    & Unfilled orders for durable goods & 5     & 1     & -0.0019 \\
	76    & Business Inventories & 5     & 1     & -0.0100 \\
	77    & Inventories to sales ratio & 2     &       & 0 \\
	78    & M1 money stock & 5     &       & 0 \\
	79    & M2 money stock & 5     &       & 0 \\
	80    & Real M2 money stock & 5     & 2     & 0.0232 \\
	81    & St. Louis adjusted monetary base & 5     & 2     & 0.0023 \\
	82    & Total reserves of depository institutions & 5     &       & 0 \\
	83    & Commercial and industrial loans & 5     & 1     & -0.0029 \\
	84    & Real estate loans at all commercial banks & 5     &       & 0 \\
	85    & Total nonrevolving credit & 5     &       & 0 \\
	86    & nonrevolving comsumer credit to personal income & 2     &       & 0 \\
	87    & S\&P 500 composite & 5     & 4     & 0.0092 \\
	88    & S\&P industrials & 5     & 4     & 0.0091 \\
	89    & S\&P common stock dividend yields & 2     & 4     & -0.0083 \\
	90    & S\&P common stock price-earnings ratio & 5     & 1,4   & 0.0066 \\
	91    & VXO   & 1     & 1,4   & 0.0063 \\
	92    & Dow Jones Industrial & 5     & 4     & 0.0034 \\
	93    & Dow Jones Utilities & 5     &       & 0 \\
	94    & Nasdaq Composite & 5     & 4     & 0.0015 \\
	95    & Dow Jones & 5     & 4     & 0.0015 \\
	96    & Nasdaq Industrial & 5     & 4     & 0.0013 \\
	97    & 3 month financial commercial paper rate & 2     & 1     & 0.2361 \\
	98    & 3 month Treasury bill & 2     &       & 0.4455 \\
	99    & 6 month Treasury bill & 2     & 1,3   & 0.4731 \\
	100   & 1 year Treasury rate & 2     & 3     & 0.4023 \\
	101   & 5 year Treasury rate & 2     & 3     & 0.0871 \\
	102   & 10 year Treasury rate & 2     & 3     & -0.0207 \\
	103   & Moody AAA corporate bond yield & 2     & 3     & -0.0912 \\
	104   & Moody BAA corporate bond yield & 2     & 3,4   & -0.0722 \\
	105   & Trade weighted US Dollar index major currencies & 5     &       & 0 \\
	106   & Switzerland US foreign exchange rate & 5     &       & 0 \\
	107   & Japan US foreign exchange rate & 5     &       & 0 \\
	108   & US UK foreign exchange rate & 5     &       & 0 \\
	109   & Canada US foreign exchange rate & 5     &       & 0 \\
	110   & PPI: finished goods & 5     & 2     & -0.0349 \\
	111   & PPI: finished consumer goods & 5     & 2     & -0.0353 \\
	112   & PPI: intermediate materials & 5     & 2     & -0.0332 \\
	113   & PPI: crude materials & 5     & 2     & -0.0145 \\
	114   & crude oil & 5     & 2     & -0.0113 \\
	115   & PPI: metals & 5     &       & 0 \\
	116   & CPI: all items & 5     & 2     & -0.0505 \\
	117   & CPI: apparel & 5     &       & 0 \\
	118   & CPI: transportation & 5     & 2     & -0.0480 \\
	119   & CPI: commodities & 5     & 2     & -0.0508 \\
	120   & CPI: durables & 5     &       & 0 \\
	121   & CPI: all items less food & 5     & 2     & -0.0501 \\
	122   & CPI: all items less shelter & 5     & 2     & -0.0516 \\
	123   & CPI: all items less medical care & 5     & 2     & -0.0509 \\
	124   & Personal consumption expenditure: chain index & 6     & 2     & -0.0090 \\
	125   & Personal consumption expenditure: durable goods & 5     &       & 0 \\
	126   & Personal consumption expenditure: nondurable goods & 5     & 2     & -0.0491
    \label{tab:data}
\end{longtable}
\end{ThreePartTable}

\renewcommand{\arraystretch}{1}
\begin{table}[h]
	\setlength{\extrarowheight}{2pt}
	\centering
	\caption{Naming variables scheme}
	\begin{threeparttable}
		\begin{tabular}{ll}
			\hhline{==}
			Index & Description of the variables \\\hline
			6      & Industrial Production (IP) index \\    
			31     & Civilian unemployment rate \\
			87     & S\&P 500 composite \\    
			66     & Moody's Seasoned Baa Corporate Bond Yield
			\\
			116    & Consumer Price Index (CPI): all items \\
			\hhline{==}
		\end{tabular}
		\begin{tablenotes}
		\footnotesize
		\singlespacing
		\item \leavevmode\kern-\scriptspace\kern-\labelsep 
		Note: This table lists the observed time series that are used as naming variables for the named factor identification case (\textit{IRb}).
		\end{tablenotes}
\end{threeparttable}
	\label{tab:named1}%
\end{table}%


\begin{table}[htbp]
	\setlength{\extrarowheight}{2pt}
  \centering
  \caption{Contemporaneous impact matrices}
  \begin{threeparttable}
	    \begin{tabular}{cccccc}
	    \hhline{======}
	    \multicolumn{6}{c}{\textbf{Panel A: regularized FAVAR model - RFAVAR}} \\	    
		\multirow{2}[0]{*}{1} & \multirow{2}[0]{*}{0} & \multirow{2}[0]{*}{0} & \multirow{2}[0]{*}{0} & \multirow{2}[0]{*}{0} & -0.0555 \\
		&       &       &       &       & (0.0316) \\
		\multirow{2}[0]{*}{0} & \multirow{2}[0]{*}{1} & \multirow{2}[0]{*}{0} & \multirow{2}[0]{*}{0} & \multirow{2}[0]{*}{0} & -0.1093 \\
		&       &       &       &       & (0.0676) \\
		\multirow{2}[0]{*}{0} & \multirow{2}[0]{*}{0} & \multirow{2}[0]{*}{1} & \multirow{2}[0]{*}{0} & \multirow{2}[0]{*}{0} & -0.1014 \\
		&       &       &       &       & (0.0721) \\
		\multirow{2}[0]{*}{0} & \multirow{2}[0]{*}{0} & \multirow{2}[0]{*}{0} & \multirow{2}[0]{*}{1} & \multirow{2}[0]{*}{0} & 0.0169 \\
		&       &       &       &       & (0.0895) \\
		\multirow{2}[0]{*}{0} & \multirow{2}[0]{*}{0} & \multirow{2}[0]{*}{0} & \multirow{2}[0]{*}{0} & \multirow{2}[0]{*}{1} &  -0.1742\\
		&       &       &       &       &  (0.0836) \\
		0     & 0     & 0     & 0     & 0     & 1 \\\\
	    
	    \multicolumn{6}{c}{\textbf{Panel B: named factors scheme - FAVAR-NF}} \\
		\multirow{2}[0]{*}{-1.2426} & \multirow{2}[0]{*}{0.3996} & \multirow{2}[0]{*}{0.2934} & \multirow{2}[0]{*}{0.2029} & \multirow{2}[0]{*}{0.0824} & -0.4235 \\
		&       &       &       &       & (0.0449) \\
		\multirow{2}[0]{*}{-0.8293} & \multirow{2}[0]{*}{-0.9389} & \multirow{2}[0]{*}{0.0315} & \multirow{2}[0]{*}{-0.2053} & \multirow{2}[0]{*}{0.0276} & -0.1201 \\
		&       &       &       &       & (0.0738) \\
		\multirow{2}[0]{*}{-0.6111} & \multirow{2}[0]{*}{-0.0810} & \multirow{2}[0]{*}{-0.2619} & \multirow{2}[0]{*}{-0.1330} & \multirow{2}[0]{*}{-1.1376} & 0.0337 \\
		&       &       &       &       & (0.0870) \\
		\multirow{2}[0]{*}{1.7036} & \multirow{2}[0]{*}{-0.3395} & \multirow{2}[0]{*}{0.9813} & \multirow{2}[0]{*}{1.1365} & \multirow{2}[0]{*}{0.3018} & 0.1589 \\
		&       &       &       &       & (0.0752) \\
		\multirow{2}[0]{*}{-1.1194} & \multirow{2}[0]{*}{-0.1157} & \multirow{2}[0]{*}{-0.9908} & \multirow{2}[0]{*}{0.7985} & \multirow{2}[0]{*}{0.5888} & -0.1626 \\
		&       &       &       &       & (0.0836) \\
		0     & 0     & 0     & 0     & 0     & 1 \\
	    \hhline{======}
	    \end{tabular}
		\begin{tablenotes}
			
			\footnotesize
			\singlespacing
			\item \leavevmode\kern-\scriptspace\kern-\labelsep 
			Note: This tables shows the contemporaneous impact matrices associated with the different identification schemes. Panel A shows the impact matrix for our structural regularized FAVAR model (\textit{IRa}). Panel B gives the impact matrix associated with the named factor identification scheme (\textit{IRb}). Standard errors are given in brackets.
		\end{tablenotes}
\end{threeparttable}

  \label{tab:imp}%
\end{table}%


\newpage
\graphicspath{ {figures/p1_3s1/} }
\setcounter{figure}{0}
\section{Figures}
\label{sec:A_figures}

\begin{figure}[H]
	\centering
	\includegraphics[width=0.75\linewidth]{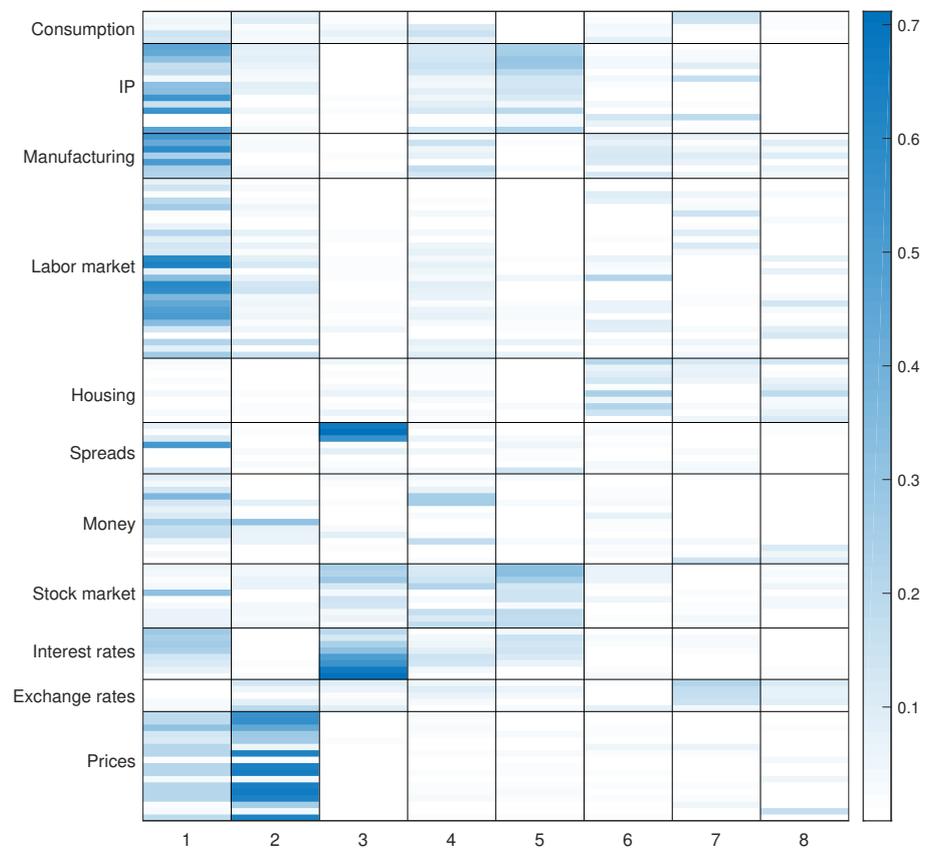}
	\caption{{$R^2$ plot for a factor model based on principal components analysis}\label{fig:hm_r2_pca} \\
	\footnotesize This graph shows $R^2$ for univariate regressions of the observed variables on each of the eight factors. The number of factors is chosen based on the \cite{Bai2002} $IC_1$ criterion.}
\end{figure}


\begin{figure}[H]
	\centering
	\includegraphics[width=0.75\linewidth]{hm_lam_r2_4}
	\caption{{$R^2$ plot for the regularized FAVAR}\label{fig:hm_r2_4} \\
	\footnotesize This graph shows $R^2$ for univariate regressions of the observed variables on each of the factors. The factors are abbreviated as follows: 'LM' labor market, 'P' price, 'IP' industrial production, 'SM' stock market, 'CS' credit spread and 'FFR' Federal Funds rate.}
\end{figure}


\begin{figure}[H]
	\centering
	\begin{minipage}{.4\textwidth}
		\centering
		\captionsetup{width=0.95\linewidth}
		\subfloat[Labor Market]{\includegraphics[width=0.95\linewidth]{factor_4_1}}
	\end{minipage}\hspace*{1cm}
	\begin{minipage}{.4\textwidth}
		\centering
		\captionsetup{width=0.95\linewidth}
		\subfloat[Price]{\includegraphics[width=0.95\linewidth]{factor_4_2}}
	\end{minipage}\\\vspace{0.3cm}
	
	\begin{minipage}{.4\textwidth}
		\centering
		\captionsetup{width=0.95\linewidth}
		\subfloat[Industrial Production]{\includegraphics[width=0.95\linewidth]{factor_4_3}}
	\end{minipage}\hspace*{1cm}
	\begin{minipage}{.4\textwidth}
		\centering
		\captionsetup{width=0.95\linewidth}
		\subfloat[Stock Market]{\includegraphics[width=0.95\linewidth]{factor_4_4}}
	\end{minipage}\\\vspace{0.3cm}
	
	\begin{minipage}{.4\textwidth}
		\centering
		\captionsetup{width=0.95\linewidth}
		\subfloat[Credit Spread]{\includegraphics[width=0.95\linewidth]{factor_4_5}}
	\end{minipage}\hspace*{1cm}
	\begin{minipage}{.4\textwidth}
		\centering
		\captionsetup{width=0.95\linewidth}
		\subfloat[Federal Funds rate]{\includegraphics[width=0.95\linewidth]{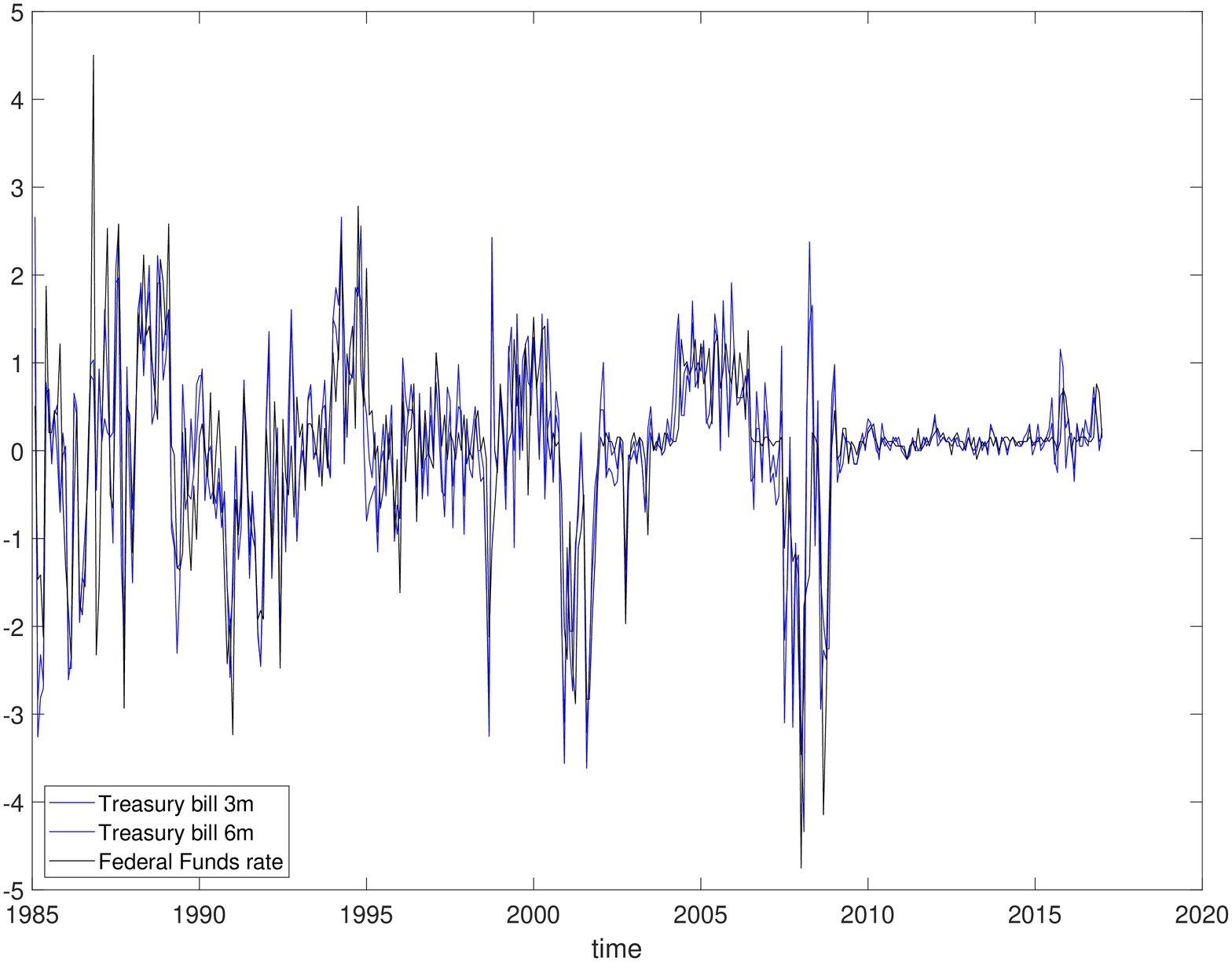}}
	\end{minipage}\\
	\caption{Factor plots for the regularized FAVAR} \label{fig:factor_rfavar}
\end{figure}

\begin{figure}[H]
	\centering
	\includegraphics[width=0.75\linewidth]{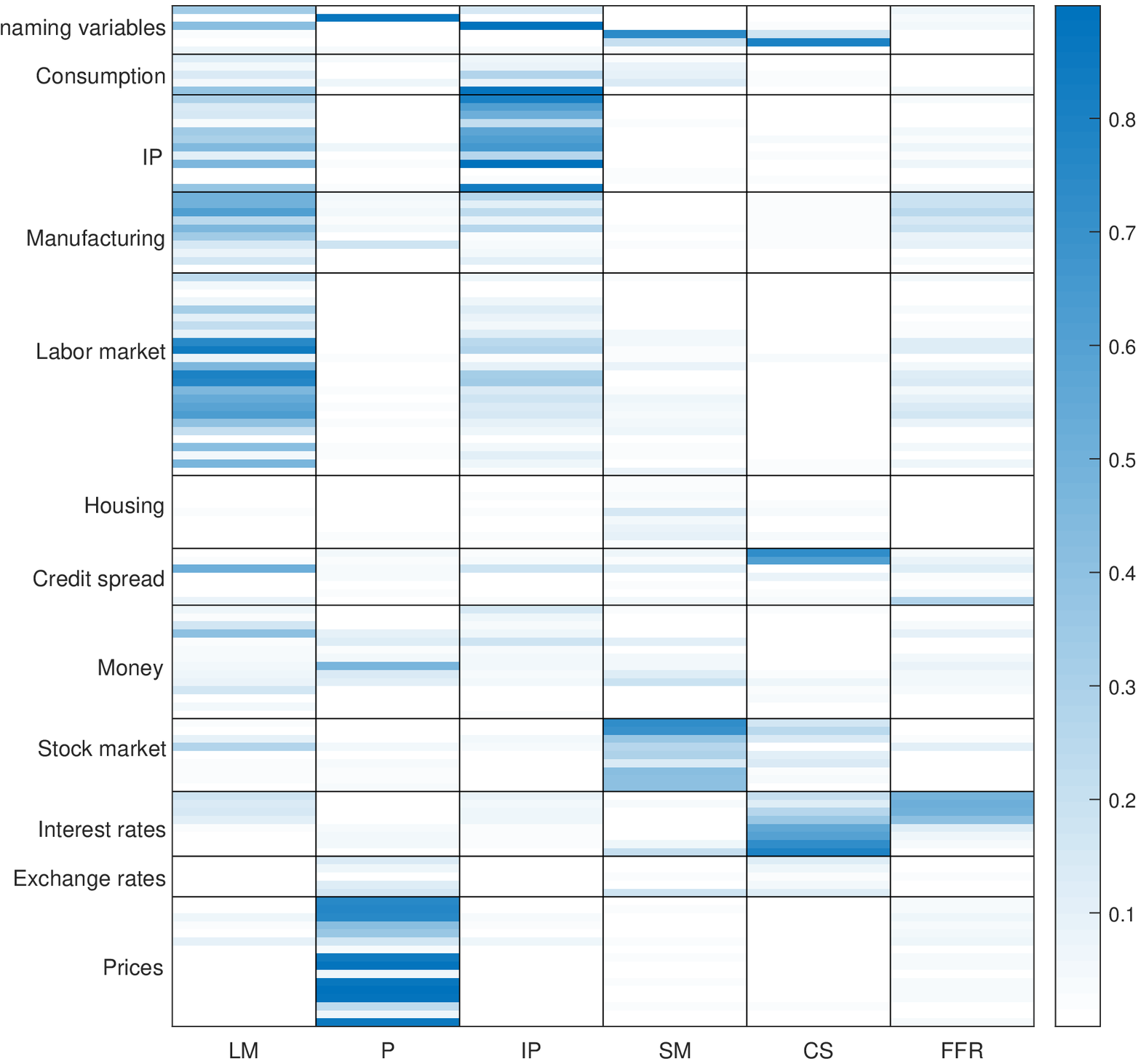}
	\caption{{$R^2$ plot for the named factor FAVAR} \label{fig:hm_r2_s2}\\
	\footnotesize This graph shows $R^2$ for univariate regressions of the observed variables on each of the factors. The factors are abbreviated as follows: 'LM' labor market, 'P' price, 'IP' industrial production, 'SM' stock market, 'CS' credit spread and 'FFR' Federal Funds rate. The naming variables are given in Table \ref{tab:named1}.}
\end{figure}


\begin{figure}[H]
	\centering
	\begin{minipage}{.4\textwidth}
		\centering
		\captionsetup{width=0.95\linewidth}
		\subfloat[Labor Market]{\includegraphics[width=0.95\linewidth]{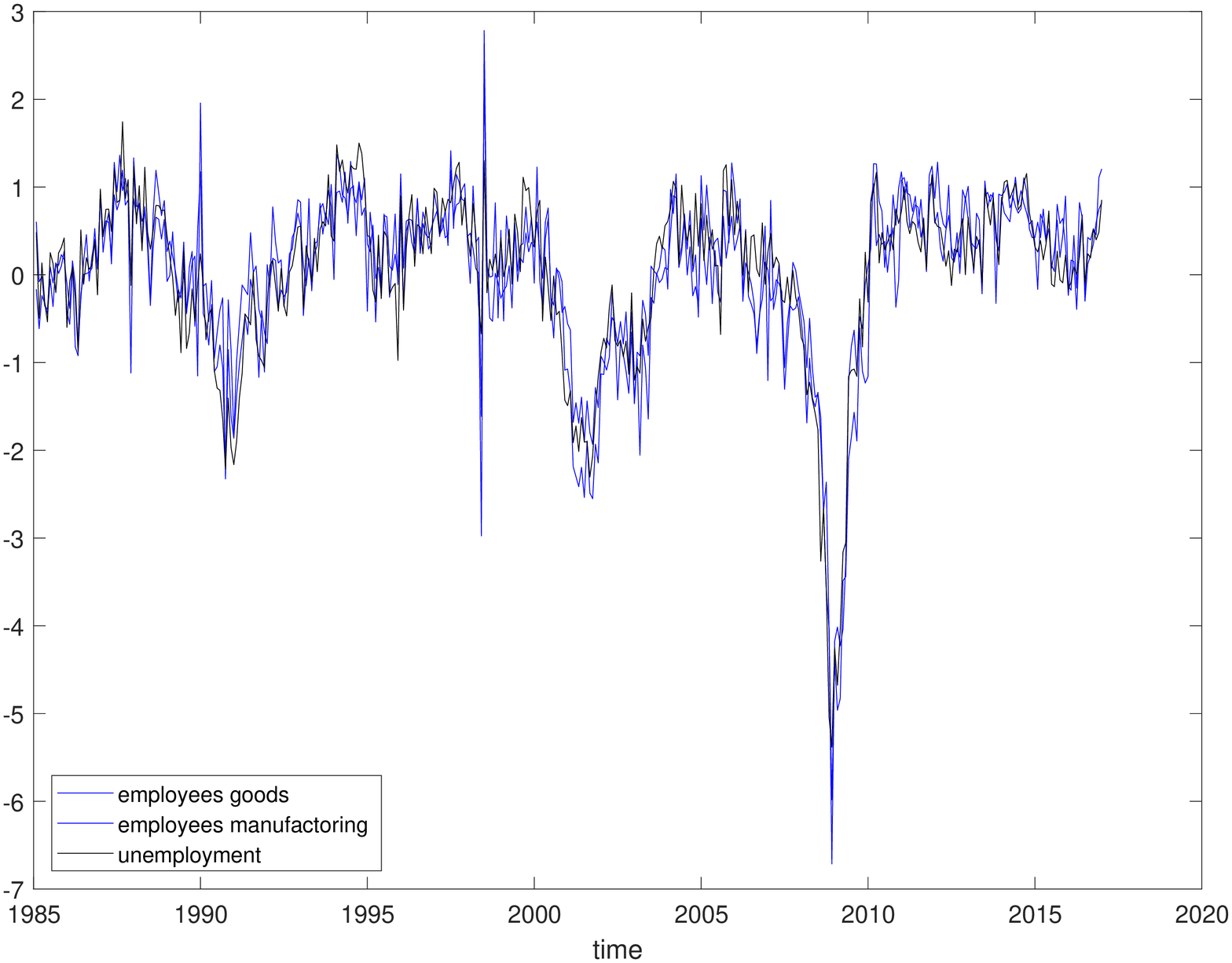}}
	\end{minipage}\hspace*{1cm}
	\begin{minipage}{.4\textwidth}
		\centering
		\captionsetup{width=0.95\linewidth}
		\subfloat[Price]{\includegraphics[width=0.95\linewidth]{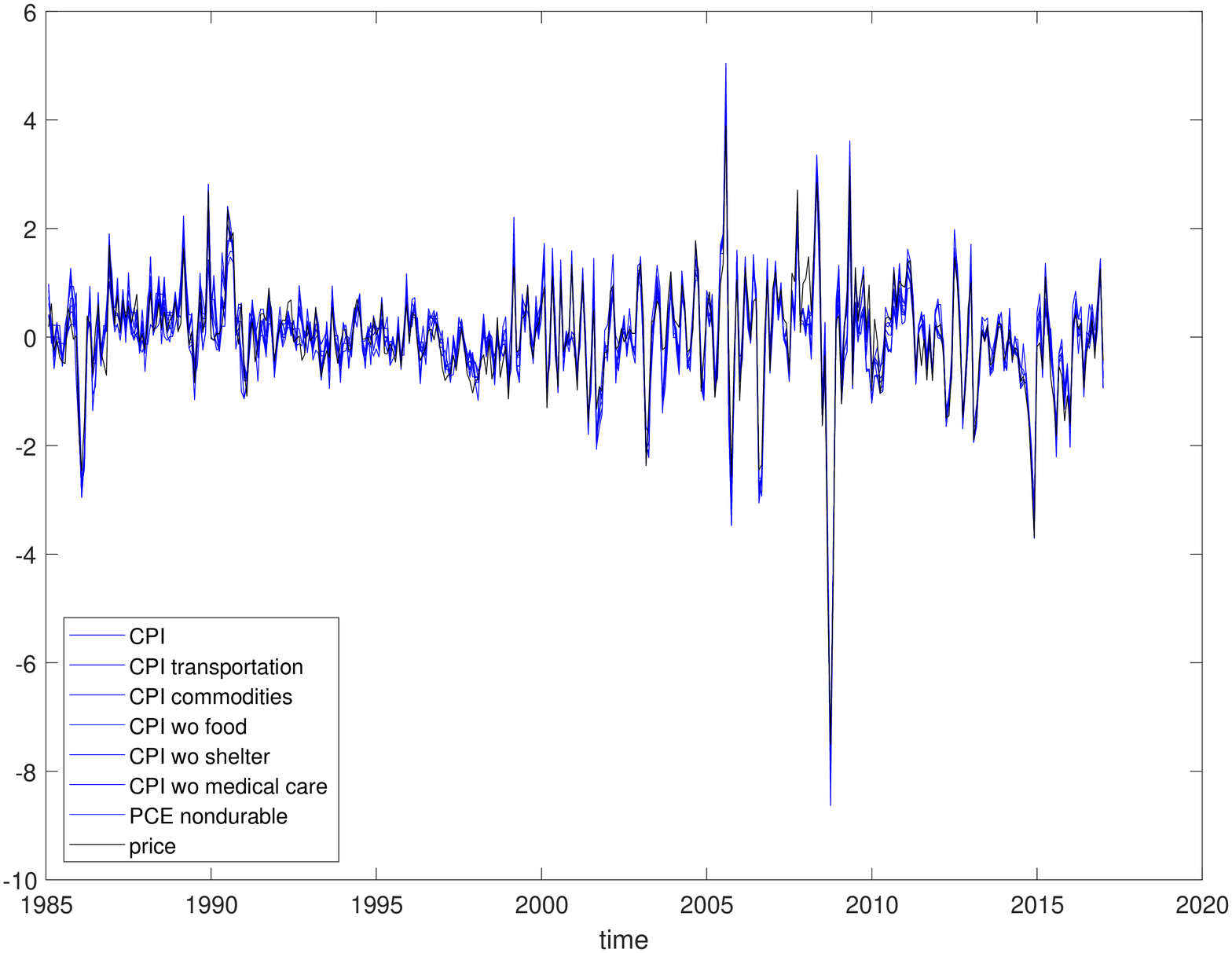}}
	\end{minipage}\\\vspace{0.3cm}
	
	\begin{minipage}{.4\textwidth}
		\centering
		\captionsetup{width=0.95\linewidth}
		\subfloat[Industrial Production]{\includegraphics[width=0.95\linewidth]{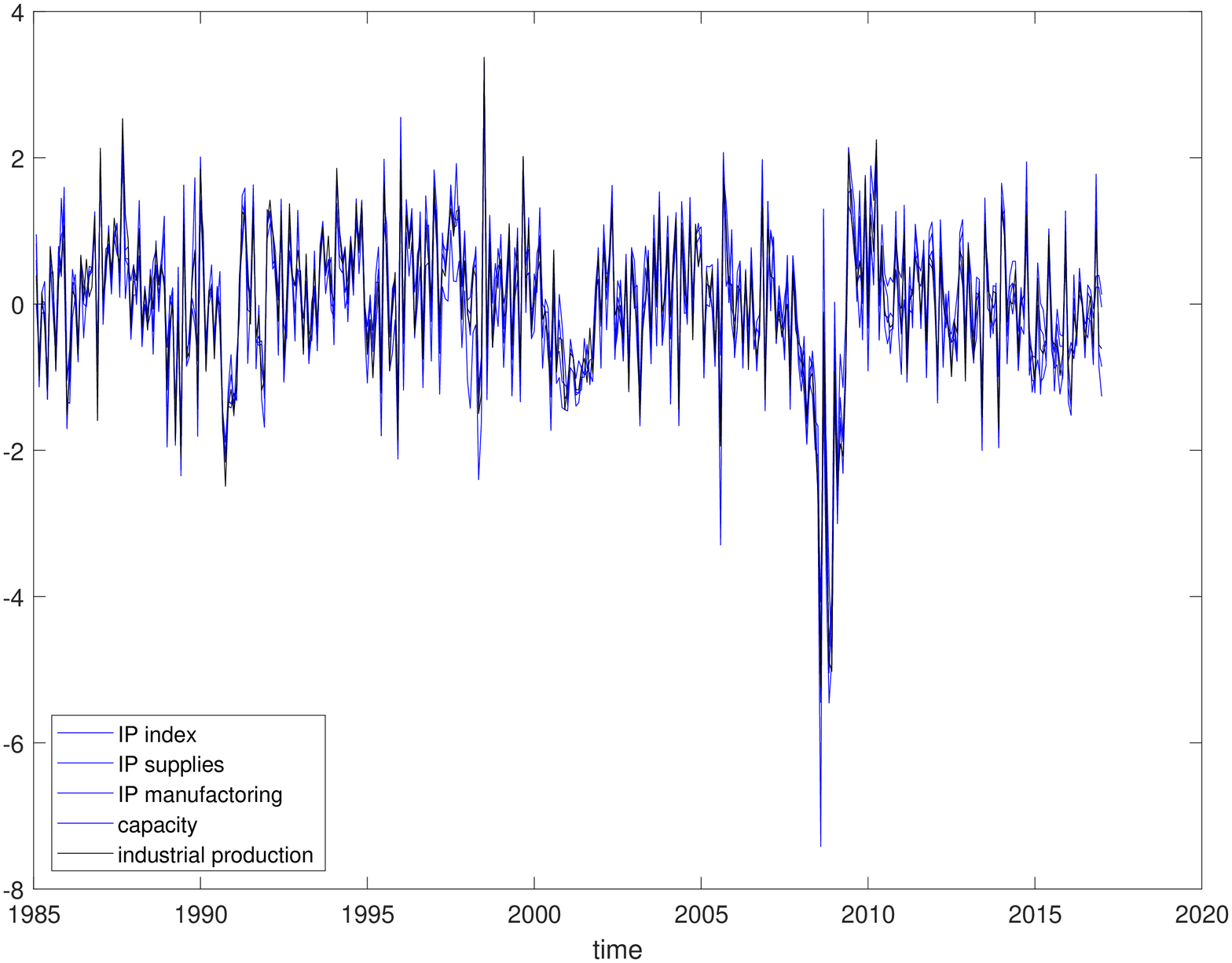}}
	\end{minipage}\hspace*{1cm}
	\begin{minipage}{.4\textwidth}
		\centering
		\captionsetup{width=0.95\linewidth}
		\subfloat[Stock Market]{\includegraphics[width=0.95\linewidth]{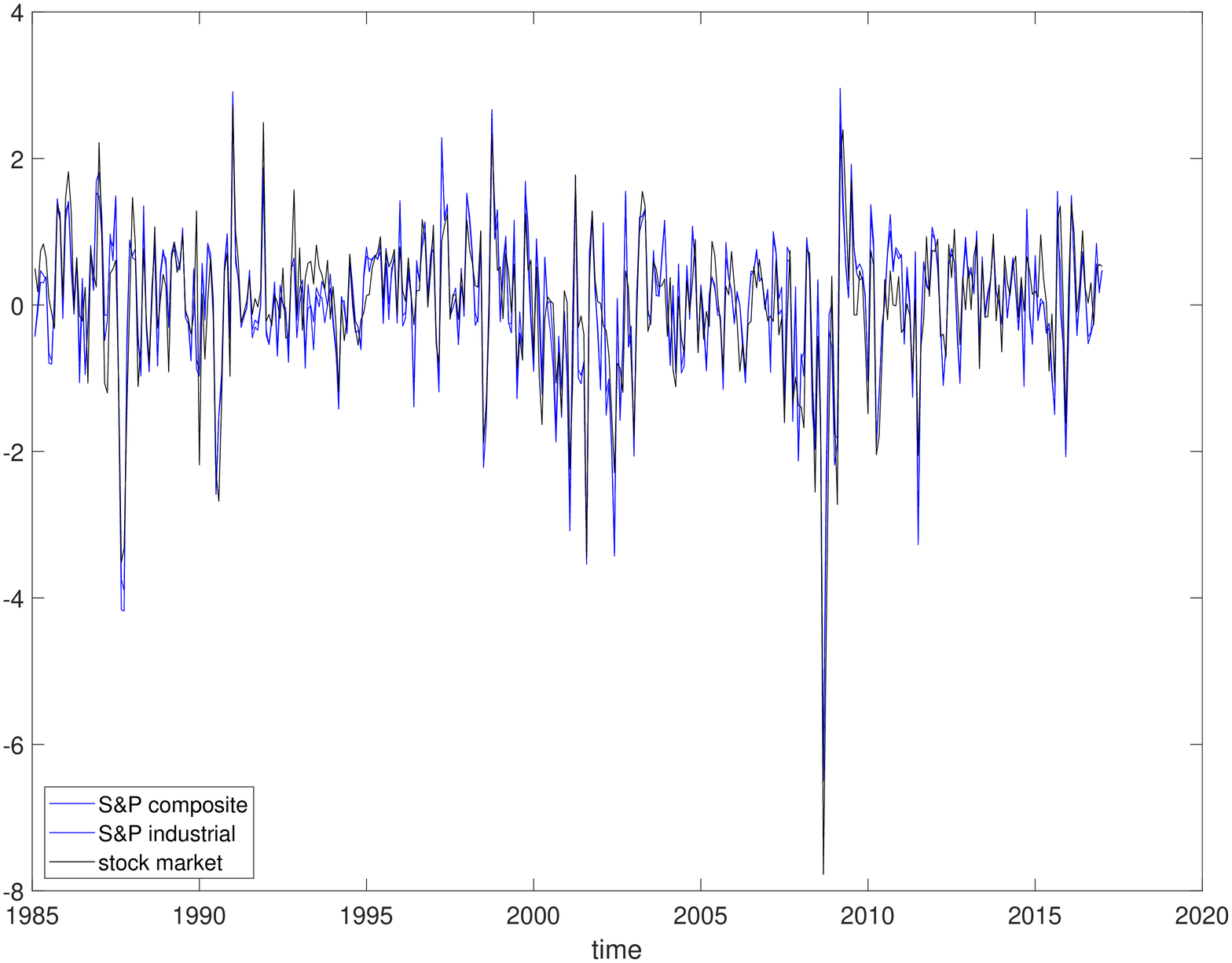}}
	\end{minipage}\\\vspace{0.3cm}
	
	\begin{minipage}{.4\textwidth}
		\centering
		\captionsetup{width=0.95\linewidth}
		\subfloat[Credit Spread]{\includegraphics[width=0.95\linewidth]{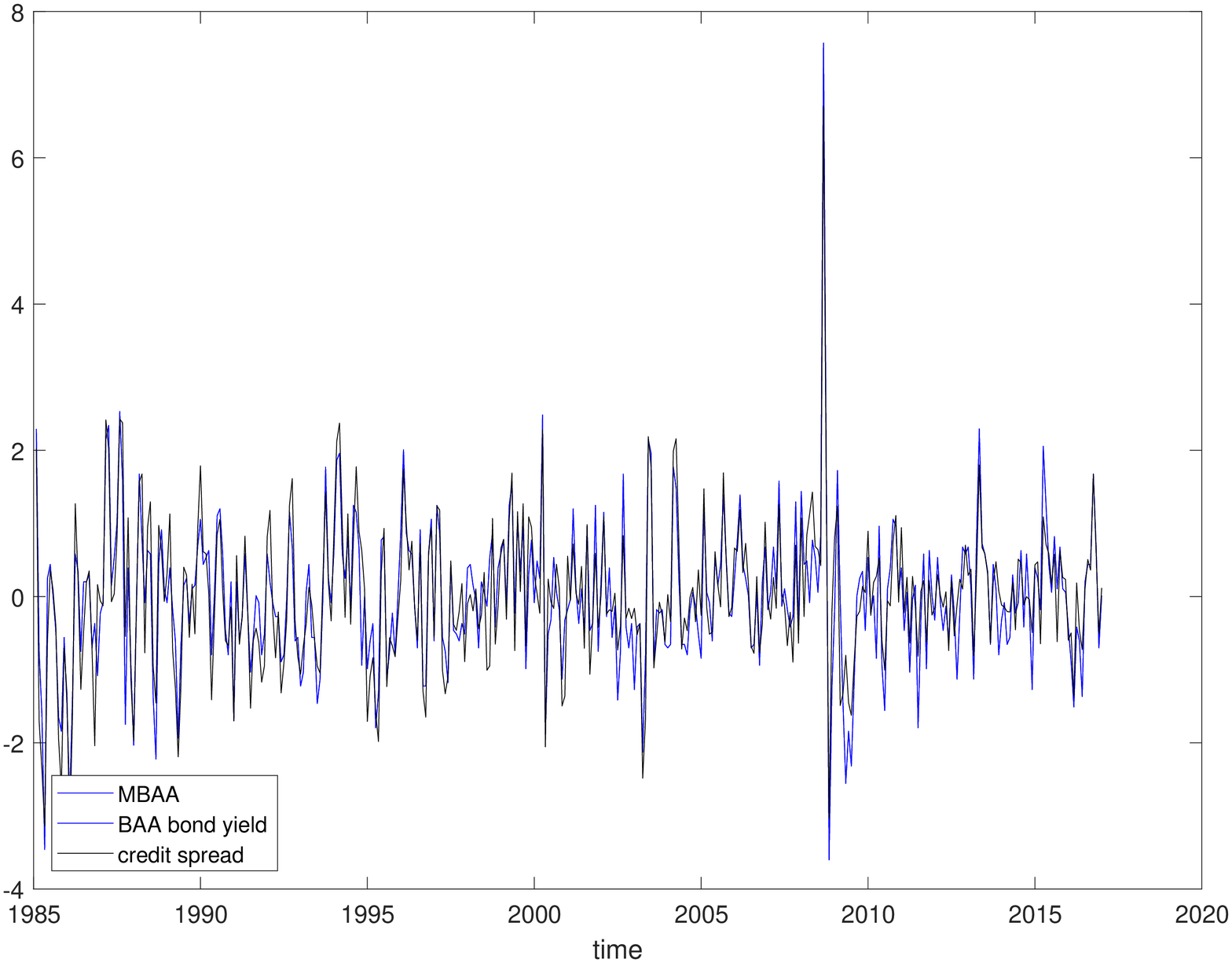}}
	\end{minipage}\hspace*{1cm}
	\begin{minipage}{.4\textwidth}
		\centering
		\captionsetup{width=0.95\linewidth}
		\subfloat[Federal Funds rate]{\includegraphics[width=0.95\linewidth]{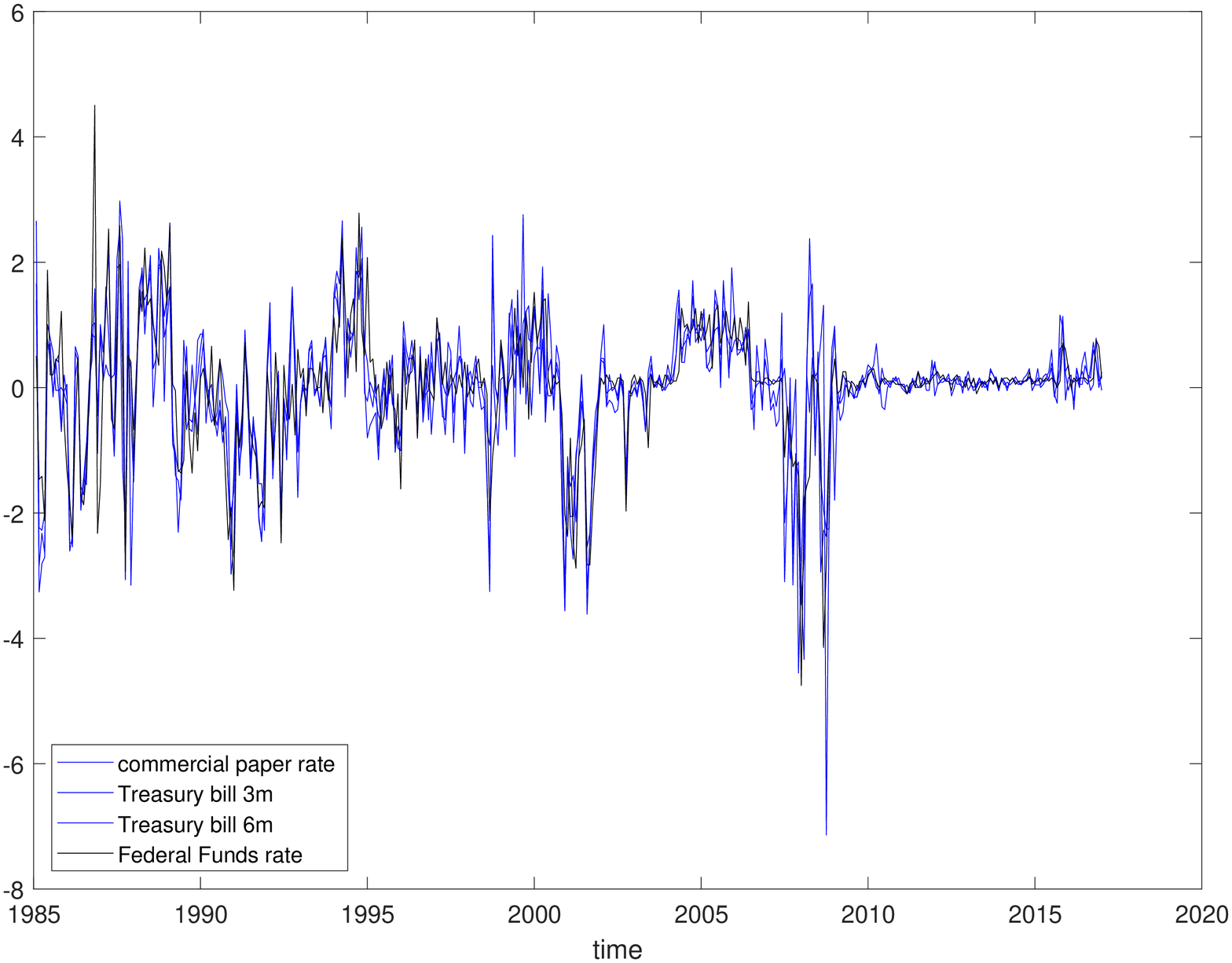}}
	\end{minipage}\\\vspace{0.3cm}
	\caption{Factor plots for the named factor FAVAR}\label{fig:factor_nfs2}
\end{figure}


\begin{figure}[H]
	\centering
	\includegraphics[width=0.75\linewidth]{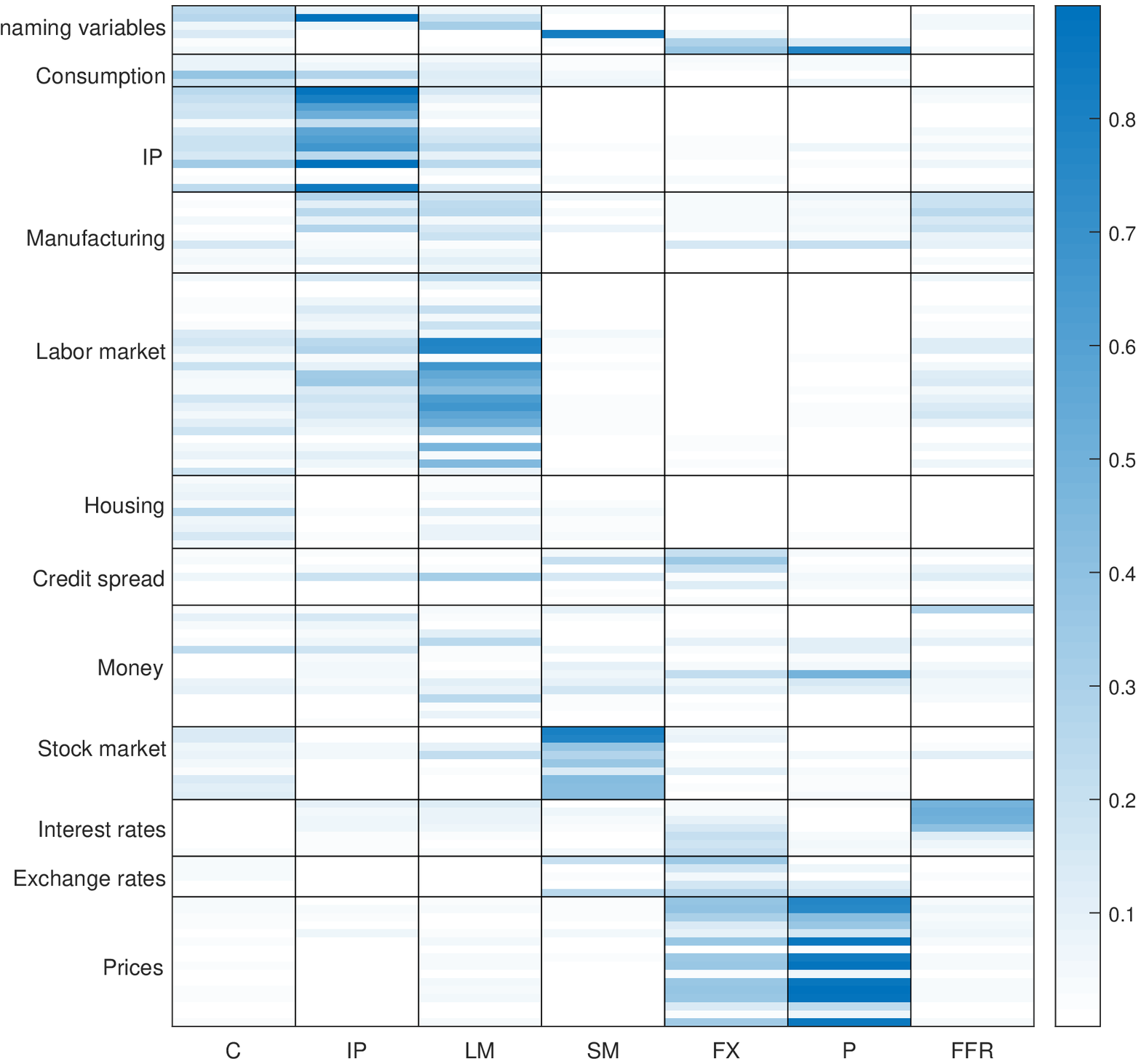}
	\caption{{$R^2$ plot for the named factor FAVAR: oil price application} \label{fig:hm_r2_s1}\\
		\footnotesize This graph shows $R^2$ for univariate regressions of the observed variables on each of the factors. The factors are abbreviated as follows: 'LM' labor market, 'P' price, 'IP' industrial production, 'SM' stock market, 'C' consumption, 'FX' is a trade weighted currency index and 'FFR' Federal Funds rate. The naming variables are based on \cite{stock2016factor}.}
\end{figure}

\vspace{1cm}

\begin{figure}[H]
	\centering
	\includegraphics[width=1\linewidth]{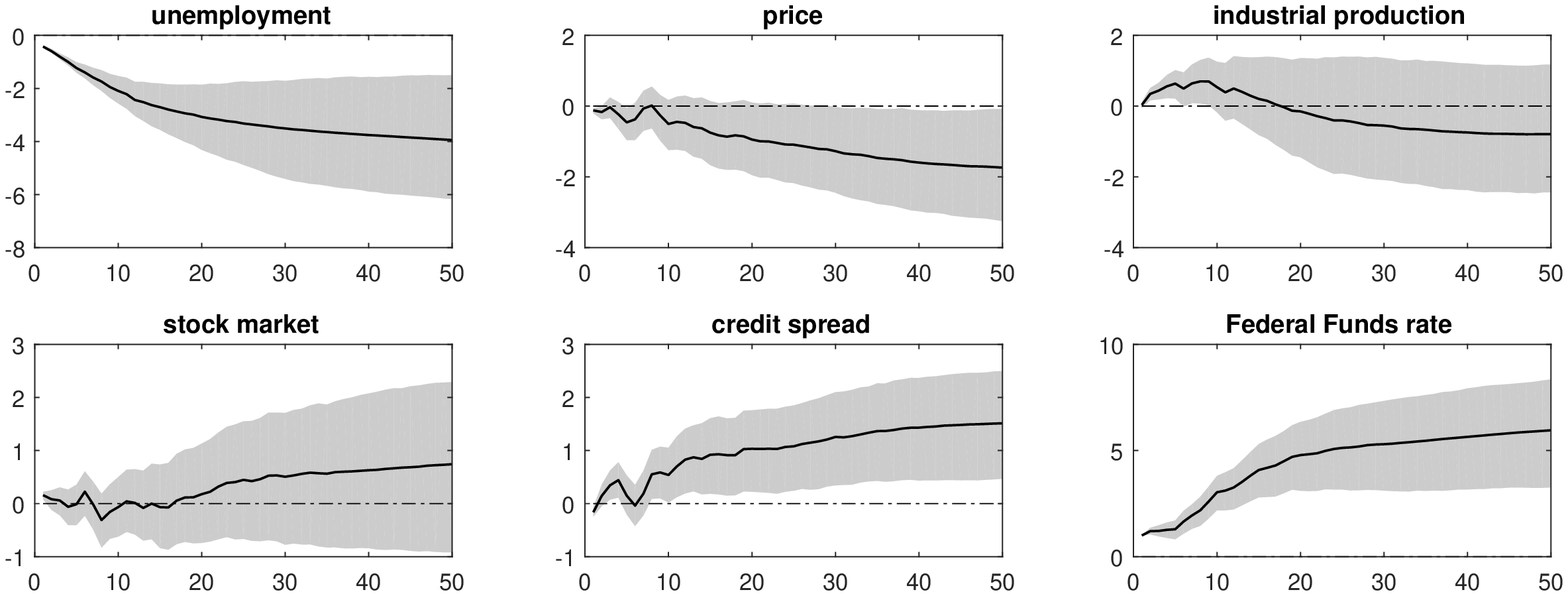}
	\caption{{Accumulated impulse responses to a monetary policy shock on the factors for the named factor FAVAR} \label{fig:irf_per_3s1_1}\\
		\footnotesize The graph shows accumulated impulse responses to a 100bp shock in the innovation of the FFR. The dashed lines correspond to 68\% bootstrap confidence intervals.}
\end{figure}


\begin{figure}[H]
	\centering
	\includegraphics[width=0.75\linewidth]{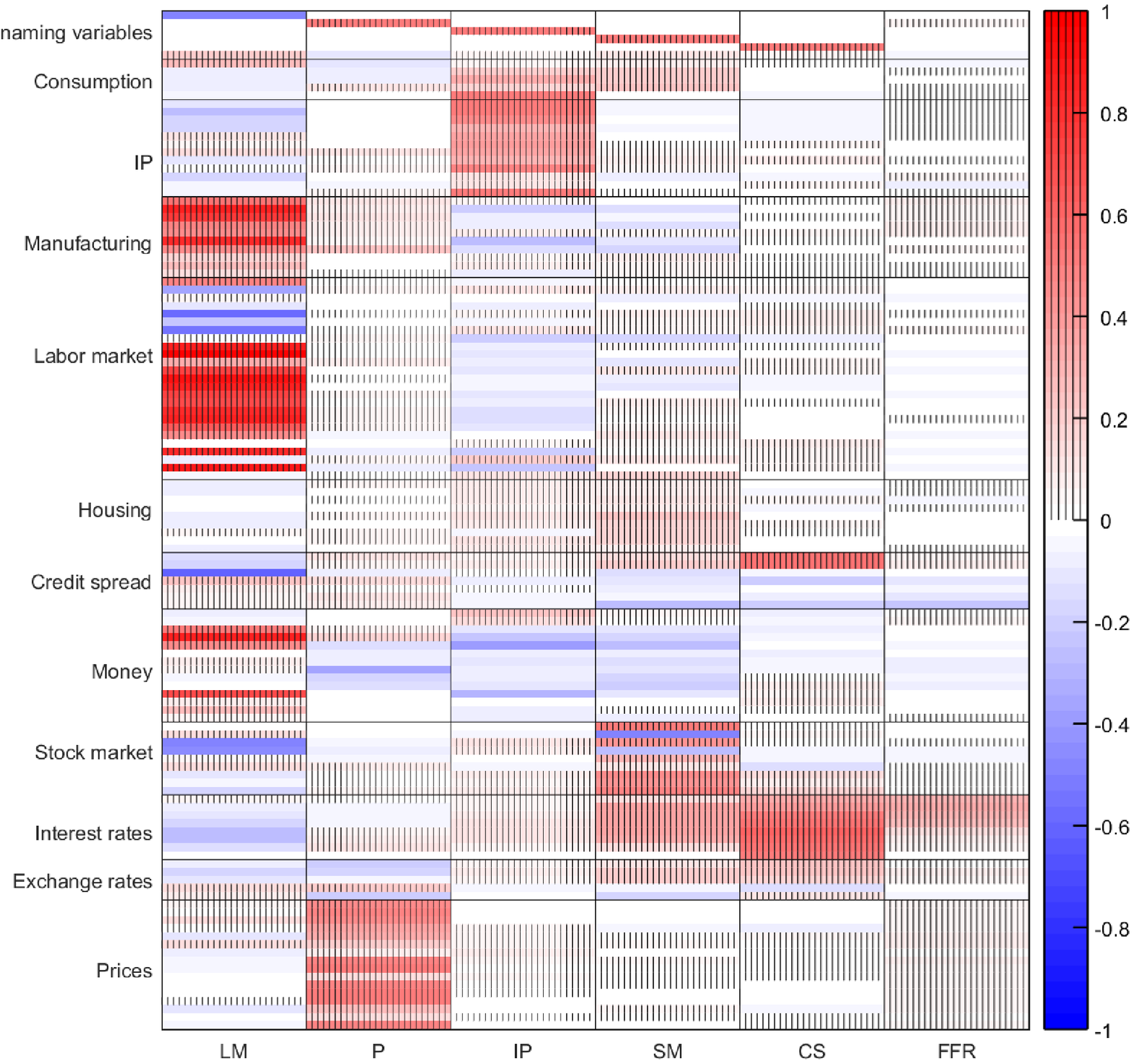}\vspace*{-1.5cm}
	\caption{{Impact matrix on the observed variables $x_t$ for the named factor FAVAR} \label{fig:hm_lam_e_3s1}\\
		\footnotesize This graph shows the contemporaneous impact matrix of a 100bp shock in the factor innovations to the observed time series for a factor model with named factor identification. The factors are abbreviated as follows: 'LM' labor market, 'P' price, 'IP' industrial production, 'SM' stock market, 'CS' credit spread and 'FFR' Federal Funds rate. The naming variables for the first scheme are given in Table \ref{tab:named1}.}
\end{figure}

\newpage



\begin{figure}[H]
	\centering
	\includegraphics[width=1\linewidth]{irf_x_per_3s1_1}
	\caption{{Accumulated impulse responses to a monetary policy shock on the observed variables $x_t$ for the named factor FAVAR}\label{fig:irf_x_per_3s1_1}\\
		\footnotesize The graph shows accumulated impulse responses to a 100bp shock in the innovation of the FFR. The dashed lines correspond to 68\% bootstrap confidence intervals.}
\end{figure}



\begin{figure}[H]
	\centering
	\includegraphics[width=1\linewidth]{FAVAR_6}
	\caption{{Accumulated impulse responses on the observed variables $x_t$ for the dynamic factor model}\label{fig:irf_x_per_6}\\
		\footnotesize The graph shows accumulated impulse responses to a 100bp shock in the innovation of the FFR. The dashed lines correspond to 68\% bootstrap confidence intervals.}
\end{figure}


\begin{figure}[H]
	\centering
	\includegraphics[width=1\linewidth]{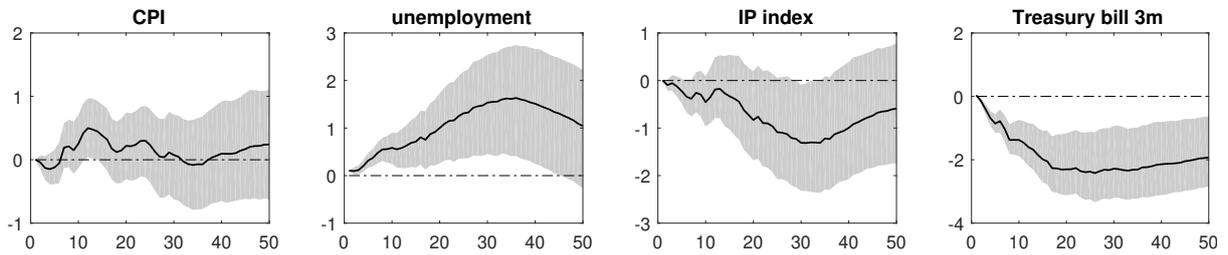}
	\caption{{Accumulated impulse responses on the observed variables $x_t$ for the VAR-F model}\label{fig:irf_x_per_5}\\
		\footnotesize The graph shows impulse responses to a 100bp shock in the innovation of the FFR. The dashed lines correspond to 68\% bootstrap confidence intervals.}
\end{figure}

\section*{Robustness Checks}
\label{sec:Figures_robust}
\subsection{Shadow rate}
\graphicspath{ {figures/shadow_rate/} }

\begin{figure}[H]
	\centering
	\includegraphics[width=1\linewidth]{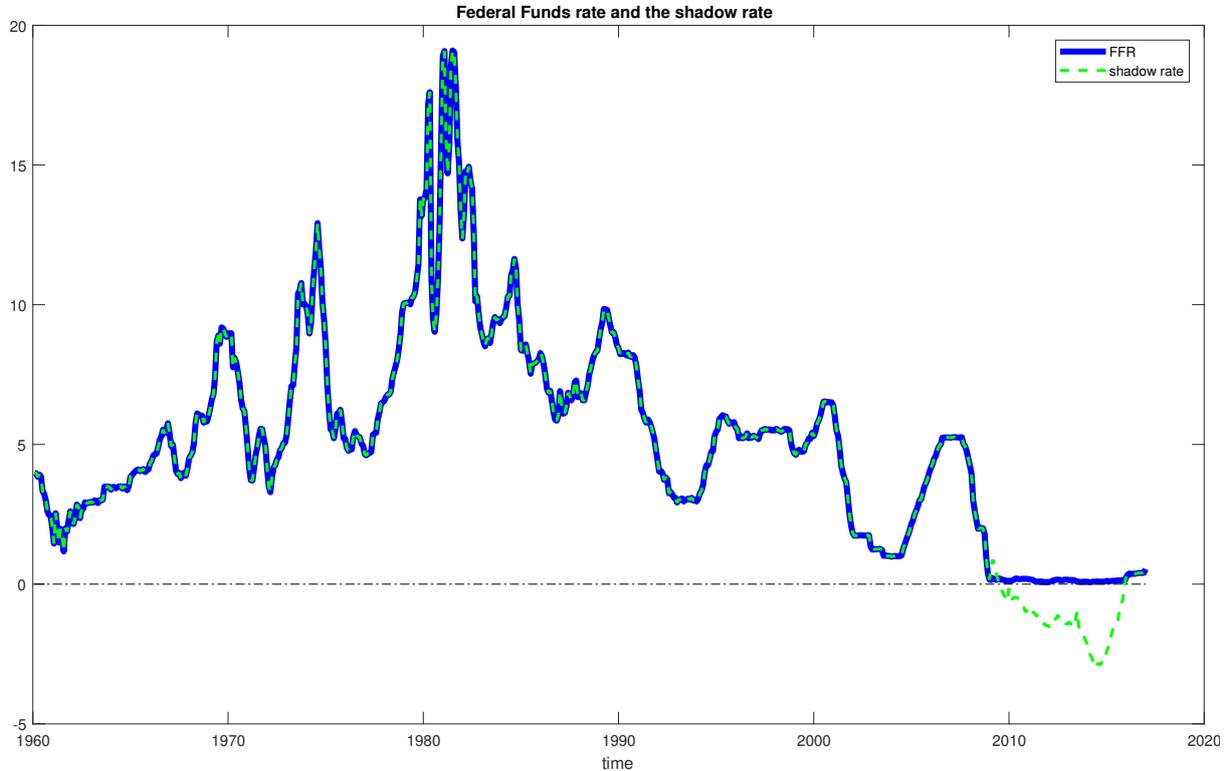}
	\caption{{Shadow and Federal funds rate}\label{fig:SR} \\
	\footnotesize This graph shows the shadow rate and the Federal Funds rate. The shadow rate takes effect whenever the FFR is at the zero lower bound, i.e. when it is below 25bp.}
\end{figure}


\begin{figure}[H]
	
	\centering
	\includegraphics[width=1\linewidth]{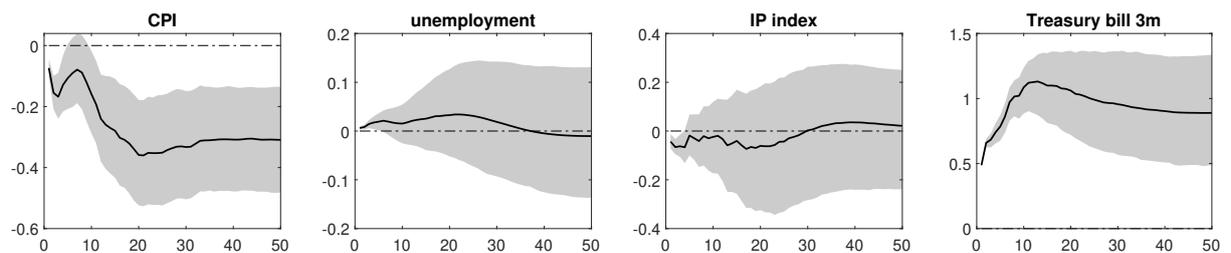}
	\caption{{Accumulated impulse responses to a monetary policy shock on the observed variables $x_t$ for the regularized FAVAR} \label{fig:irf_x_per_4_1_p1_sr}\\
		\footnotesize The graph shows accumulated impulse responses to a 100bp shock in the innovation of the shadow rate. The dashed lines correspond to 68\% bootstrap confidence intervals. The impulse responses are based on a RFAVAR(12) model.}
\end{figure}


%


\begin{figure}[H]
	
	\centering
	\includegraphics[width=1\linewidth]{irf_x_per_3s1_1}
	\caption{{Accumulated impulse responses to a monetary policy shock on the observed variables $x_t$ for the named factor FAVAR} \label{fig:irf_x_per_3s1_1_sr}\\
		\footnotesize The graph shows accumulated impulse responses to a 100bp shock in the innovation of the shadow rate. The dashed lines correspond to 68\% bootstrap confidence intervals.}
\end{figure}

\begin{figure}[H]
	\centering
	\includegraphics[width=1\linewidth]{FAVAR_6}
	\caption{{Accumulated impulse responses on the observed variables $x_t$ for the dynamic factor model}\label{fig:irf_x_per_6_sr}\\
		\footnotesize The graph shows accumulated impulse responses to a 100bp shock in the innovation of the shadow rate. The dashed lines correspond to 68\% bootstrap confidence intervals.}
\end{figure}


\begin{figure}[H]
	\centering
	\includegraphics[width=1\linewidth]{FAVAR_5}
	\caption{{Accumulated impulse responses on the observed variables $x_t$ for the VAR-F model}\label{fig:irf_x_per_5_sr}\\
		\footnotesize The graph shows impulse responses to a 100bp shock in the innovation of the shadow rate. The dashed lines correspond to 68\% bootstrap confidence intervals.}
\end{figure}

\graphicspath{ {figures/p1_r20/} }
\subsection{Different initial number of factors}
\label{sec:A_figures_different_r}
\begin{figure}[H]
	\centering
	\includegraphics[width=0.75\linewidth]{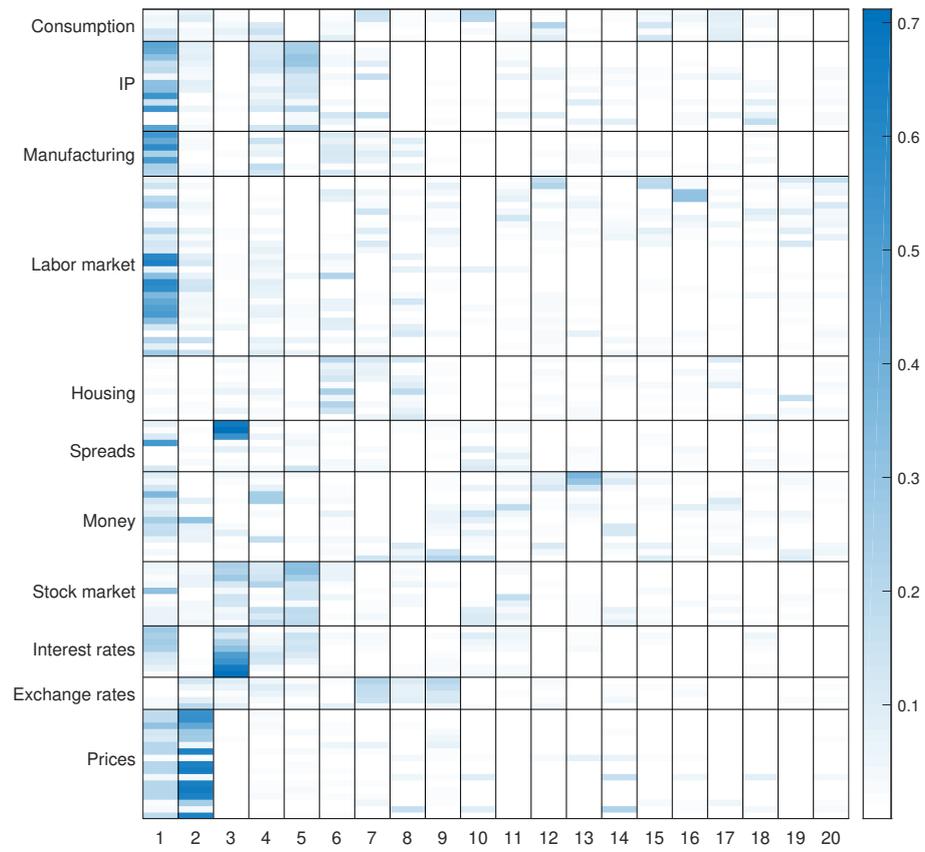}
	\caption{{$R^2$ plot for a factor model based on principal components analysis}\label{fig:hm_r2_pca_r} \\
	\footnotesize This graph shows $R^2$ for univariate regressions of the observed variables on each of the 20 factors.}
\end{figure}


\begin{figure}[H]
	\centering
	\includegraphics[width=0.75\linewidth]{hm_lam_r2_4}
	\caption{{$R^2$ plot for the regularized FAVAR}\label{fig:hm_r2_4_r} \\
	\footnotesize This graph shows $R^2$ for univariate regressions of the observed variables on each of the factors. The factors are abbreviated as follows: 'LM' labor market, 'P' price, 'IP' industrial production, 'SM' stock market, 'CS' credit spread and 'FFR' Federal Funds rate.}
\end{figure}

\graphicspath{ {figures/p2_r10/} }
\subsection{Different lag order}
\subsubsection{RFAVAR(2) model}
\label{sec:A_figures_different_p2}




\begin{figure}[H]
	\centering
	\includegraphics[width=1\linewidth]{irf_per_4_1}
	\caption{{Accumulated impulse responses to a monetary policy shock on the factors for the regularized FAVAR} \label{fig:irf_per_4_1_p2}\\
		\footnotesize The graph shows accumulated impulse responses to a 100bp shock in the innovation of the FFR. The dashed lines correspond to 68\% bootstrap confidence intervals. The impulse responses are based on a RFAVAR(2) model.}
\end{figure}



\vspace{1cm}

\begin{figure}[H]
	
	\centering
	\includegraphics[width=1\linewidth]{irf_x_per_4_1}
	\caption{{Accumulated impulse responses to a monetary policy shock on the observed variables $x_t$ for the regularized FAVAR} \label{fig:irf_x_per_4_1_p2}\\
		\footnotesize The graph shows accumulated impulse responses to a 100bp shock in the innovation of the FFR. The dashed lines correspond to 68\% bootstrap confidence intervals. The impulse responses are based on a RFAVAR(2) model.}
\end{figure}

\subsubsection{RFAVAR(3) model}
\graphicspath{ {figures/p3_r10/} }
\label{sec:A_figures_different_p3}




\begin{figure}[H]
	\centering
	\includegraphics[width=1\linewidth]{irf_per_4_1}
	\caption{{Accumulated impulse responses to a monetary policy shock on the factors for the regularized FAVAR} \label{fig:irf_per_4_1_p3}\\
		\footnotesize The graph shows accumulated impulse responses to a 100bp shock in the innovation of the FFR. The dashed lines correspond to 68\% bootstrap confidence intervals. The impulse responses are based on a RFAVAR(3) model.}
\end{figure}



\vspace{1cm}

\begin{figure}[H]
	
	\centering
	\includegraphics[width=1\linewidth]{irf_x_per_4_1}
	\caption{{Accumulated impulse responses to a monetary policy shock on the observed variables $x_t$ for the regularized FAVAR} \label{fig:irf_x_per_4_1_p3}\\
		\footnotesize The graph shows accumulated impulse responses to a 100bp shock in the innovation of the FFR. The dashed lines correspond to 68\% bootstrap confidence intervals. The impulse responses are based on a RFAVAR(3) model.}
\end{figure}

\subsubsection{RFAVAR(6) model}
\graphicspath{ {figures/p6_r10/} }
\label{sec:A_figures_different_p6}




\begin{figure}[H]
	\centering
	\includegraphics[width=1\linewidth]{irf_per_4_1}
	\caption{{Accumulated impulse responses to a monetary policy shock on the factors for the regularized FAVAR} \label{fig:irf_per_4_1_p6}\\
		\footnotesize The graph shows accumulated impulse responses to a 100bp shock in the innovation of the FFR. The dashed lines correspond to 68\% bootstrap confidence intervals. The impulse responses are based on a RFAVAR(6) model.}
\end{figure}



\vspace{1cm}

\begin{figure}[H]
	
	\centering
	\includegraphics[width=1\linewidth]{irf_x_per_4_1}
	\caption{{Accumulated impulse responses to a monetary policy shock on the observed variables $x_t$ for the regularized FAVAR} \label{fig:irf_x_per_4_1_p6}\\
		\footnotesize The graph shows accumulated impulse responses to a 100bp shock in the innovation of the FFR. The dashed lines correspond to 68\% bootstrap confidence intervals. The impulse responses are based on a RFAVAR(6) model.}
\end{figure}


\graphicspath{ {figures/p1_2006/} }
\subsection{Pre-crisis time span}

\label{sec:A_figures_different_period}
\begin{figure}[H]
	\centering
	\includegraphics[width=0.75\linewidth]{hm_lam_r2_pca}
	\caption{{$R^2$ plot for the FAVAR based on principal components analysis}\label{fig:hm_r2_pca_2006} \\
	\footnotesize This graph shows $R^2$ for univariate regressions of the observed variables on each of the nine factors. The number of factors is chosen based on the \cite{Bai2002} $IC_1$ criterion. The sample spans the pre-crisis period from January 1985 to December 2006.}
\end{figure}


\begin{figure}[H]
	\centering
	\includegraphics[width=0.75\linewidth]{hm_lam_r2_4}
	\caption{{$R^2$ plot for the regularized FAVAR}\label{fig:hm_r2_4_2006} \\
	\footnotesize This graph shows $R^2$ for univariate regressions of the observed variables on each of the factors. The factors are abbreviated as follows: 'P' price, 'CS' credit spreads, 'IP' industrial production, 'SM' stock market  and 'FFR' Federal Funds rate. The sample spans the pre-crisis period from January 1985 to December 2006.}
\end{figure}


\begin{figure}[H]
	\centering
	\begin{minipage}{.4\textwidth}
		\centering
		\captionsetup{width=0.95\linewidth}
		\subfloat[Price]{\includegraphics[width=0.95\linewidth]{factor_4_1}}
	\end{minipage}\hspace*{1cm}
	\begin{minipage}{.4\textwidth}
		\centering
		\captionsetup{width=0.95\linewidth}
		\subfloat[Credit Spread]{\includegraphics[width=0.95\linewidth]{factor_4_2}}
	\end{minipage}\\\vspace{0.3cm}
	
	\begin{minipage}{.4\textwidth}
		\centering
		\captionsetup{width=0.95\linewidth}
		\subfloat[Industrial Production]{\includegraphics[width=0.95\linewidth]{factor_4_3}}
	\end{minipage}\hspace*{1cm}
	\begin{minipage}{.4\textwidth}
		\centering
		\captionsetup{width=0.95\linewidth}
		\subfloat[Stock Market]{\includegraphics[width=0.95\linewidth]{factor_4_4}}
	\end{minipage}\\\vspace{0.3cm}
	\begin{minipage}{.4\textwidth}
		\centering
		\captionsetup{width=0.95\linewidth}
		\subfloat[Federal Funds Rate]{\includegraphics[width=0.95\linewidth]{factor_4_5}}
	\end{minipage}\\\vspace{0.3cm}
	\caption{Factor plots for the regularized FAVAR} \label{fig:factor_rfavar_2006}
\end{figure}

\begin{figure}[H]
	\centering
	\includegraphics[width=0.75\linewidth]{hm_lam_e_4}
	\caption{{Impact matrix on the observed variables $x_t$ for the regularized FAVAR} \label{fig:hm_lam_e_4_2006} \\
		\footnotesize This graph shows the contemporaneous impact matrix of a 100bp shock in the factor innovations to the observed time series for the regularized FAVAR model. The factors are abbreviated as follows: 'P' price, 'CS' credit spreads, 'IP' industrial production, 'SM' stock market  and 'FFR' Federal Funds rate. The sample spans the pre-crisis period from January 1985 to December 2006.}
\end{figure}




\begin{figure}[H]
	
	\centering
	\includegraphics[width=1\linewidth]{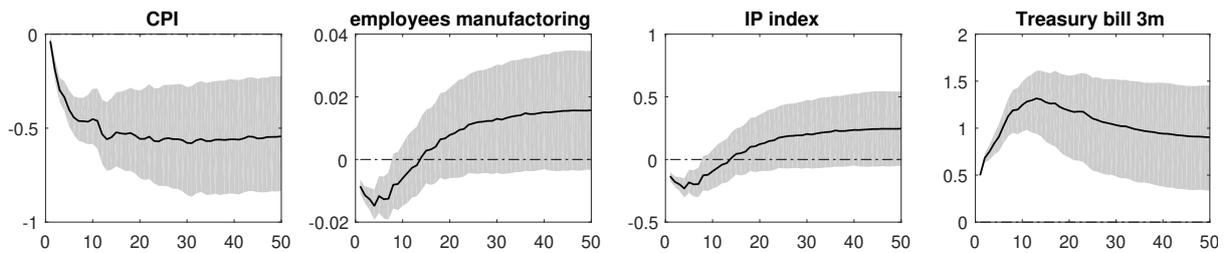}
	\caption{{Accumulated impulse responses to a monetary policy shock on the observed variables $x_t$ for the regularized FAVAR} \label{fig:irf_x_per_4_1_2006}\\
		\footnotesize The graph shows accumulated impulse responses to a 100bp shock in the innovation of the FFR. The dashed lines correspond to 68\% bootstrap confidence intervals. The impulse responses are based on a RFAVAR(12) model. The sample spans the pre-crisis period from January 1985 to December 2006.}
\end{figure}


\end{appendix}
\end{document}